\documentclass[11pt,letterpaper]{article}
\usepackage{amsmath,amssymb,amsthm}
\usepackage[margin=1in, centering]{geometry}
\usepackage{setspace}
\usepackage[numbers]{natbib}
\usepackage{xcolor}
\definecolor{HANADA}{RGB}{0, 98, 132}
\definecolor{KURENAI}{RGB}{203, 27, 69}
\usepackage[
    pdfstartview=FitH,
    pdfpagemode=UseNone,
    colorlinks=true,
    citecolor=KURENAI,
    linkcolor=HANADA,
    linktocpage=true
]{hyperref}
\usepackage[capitalise, nameinlink]{cleveref}
\usepackage{biolinum}
\usepackage{tgpagella}
\usepackage{mathpazo}
\usepackage{thm-restate}
\usepackage{graphicx}
\usepackage{xcolor}
\usepackage{enumitem}
\usepackage{booktabs}
\usepackage{tabularx}
\usepackage{array}
\usepackage{makecell}
\usepackage[margin=1in]{geometry}
\usepackage{authblk}

\renewcommand\Affilfont{\small}
\makeatletter
\renewcommand\AB@affilsepx{\protect\\[0.2em]\protect\Affilfont}
\makeatother

\usepackage{mathtools}
\usepackage{dsfont}

\newtheorem{theorem}{Theorem}[section]
\newtheorem{proposition}[theorem]{Proposition}
\newtheorem{fact}[theorem]{Fact}
\newtheorem{lemma}[theorem]{Lemma}
\newtheorem{corollary}[theorem]{Corollary}
\newtheorem{claim}[theorem]{Claim}

\newtheorem{definition}[theorem]{Definition}
\newtheorem{remark}[theorem]{Remark}

\crefname{assumption}{Black-box Theorem}{Black-box Theorems}
\Crefname{assumption}{Black-box Theorem}{Black-box Theorems}
\makeatletter
\newcommand*{\threshold@theHshared}{\thesection.\arabic{theorem}}
\providecommand*{\theHtheorem}{\threshold@theHshared}
\renewcommand*{\theHtheorem}{\threshold@theHshared}
\providecommand*{\theHproposition}{\threshold@theHshared}
\renewcommand*{\theHproposition}{\threshold@theHshared}
\providecommand*{\theHfact}{\threshold@theHshared}
\renewcommand*{\theHfact}{\threshold@theHshared}
\providecommand*{\theHlemma}{\threshold@theHshared}
\renewcommand*{\theHlemma}{\threshold@theHshared}
\providecommand*{\theHcorollary}{\threshold@theHshared}
\renewcommand*{\theHcorollary}{\threshold@theHshared}
\providecommand*{\theHclaim}{\threshold@theHshared}
\renewcommand*{\theHclaim}{\threshold@theHshared}
\providecommand*{\theHconjecture}{\threshold@theHshared}
\renewcommand*{\theHconjecture}{\threshold@theHshared}
\providecommand*{\theHdefinition}{\threshold@theHshared}
\renewcommand*{\theHdefinition}{\threshold@theHshared}
\providecommand*{\theHremark}{\threshold@theHshared}
\renewcommand*{\theHremark}{\threshold@theHshared}
\providecommand*{\theHassumption}{\threshold@theHshared}
\renewcommand*{\theHassumption}{\threshold@theHshared}
\makeatother

\newcommand{\appref}[1]{\hyperref[#1]{Appendix~\ref*{#1}}}

\usepackage{etoolbox}
\AtEndEnvironment{proof}{\qedhere}

\definecolor{lightcyan}{RGB}{0.88,1,1}
\definecolor{darkgreen}{RGB}{0, 128, 0}
\definecolor{darkblue}{RGB}{0, 0, 128}

\newcommand{\e}{\mathrm{e}}
\newcommand{\N}{\mathbb{N}}  \newcommand{\R}{\mathbb{R}}   \renewcommand{\d}{\mathrm{d}} \DeclareMathOperator*{\E}{\mathbb{E}}      \newcommand{\A}{\mathcal{A}}   \newcommand{\CC}{\mathcal{C}} \newcommand{\D}{\mathcal{D}}     \newcommand{\NN}{\mathcal{N}}    \renewcommand{\S}{\mathcal{S}}   \newcommand{\EE}{\mathcal{E}}  \newcommand{\RR}{\mathcal{R}}          \newcommand{\pmcube}[1]{\{-1,1\}^{#1}}      \newcommand{\expect}[2]{\E_{\substack{#1}}\!\Br{#2}}       \newcommand{\AND}{\newsf{AND}}
\newcommand{\OR}{\newsf{OR}}

\newcommand{\He}{\mathrm{He}}

\newcommand{\br}[1]{\left(#1\right)} \newcommand{\Br}[1]{\left[#1\right]}   \newcommand{\abs}[1]{\left|#1 \right|} \newcommand{\norm}[1]{\left\lVert #1 \right\rVert}     \newcommand{\poly}[1]{\mathrm{poly}\!\br{#1}}                            \newcommand{\ip}[1]{\langle #1 \rangle} 

\newcommand{\newsf}[1]{\textnormal{\textup{\textsf{#1}}}}
\newcommand{\ltf}{\newsf{LTF}}
\newcommand{\ptf}{\newsf{PTF}}
\newcommand{\thr}{\newsf{Thr}}
\newcommand{\prg}{\newsf{PRG}}
\newcommand{\CNF}{\newsf{CNF}}
\newcommand{\DNF}{\newsf{DNF}}
\newcommand{\gcdf}{\Phi}
\newcommand{\gpdf}{\phi}
\newcommand{\mollifier}{\widetilde{\mathcal O}}
\newcommand{\mollifierkernel}{\mathcal B}
\newcommand{\clauseprob}{w}

\newcommand{\clausemass}{q}
\newcommand{\clausemoment}{M}

\newcommand{\ort}{\mathcal{O}}
\newcommand{\ind}{\mathds{1}}
\newcommand{\onevec}{\mathbf{1}}
\newcommand{\zerovec}{\mathbf{0}}

\newcommand{\tildeo}{\widetilde{O}}
\newcommand{\gmr}{G_{\newsf{GMR}}}
\newcommand{\xor}{\newsf{XOR}}
\newcommand{\ns}{\newsf{NS}}
\newcommand{\gsa}{\newsf{GSA}}

\providecommand{\pmcube}[1]{\{-1,1\}^{#1}}
\providecommand{\CNF}{\ensuremath{\newsf{CNF}}}
\providecommand{\DNF}{\ensuremath{\newsf{DNF}}}
\providecommand{\Khead}{K_{\newsf{head}}}
\providecommand{\Wcnf}{W_{\newsf{CNF}}}
\providecommand{\whd}{w_{\newsf{head}}}

\providecommand{\dhyb}{d_{\newsf{hyb}}}
\providecommand{\deltaCNF}{\delta_{\newsf{CNF}}}
\providecommand{\polylog}{\mathrm{polylog}}
\providecommand{\coloneqq}{:=}
\providecommand{\D}{\mathcal{D}}

\usepackage[normalem]{ulem}

\title{\textbf{Fooling Thresholds of Halfspaces}\vspace{0.3em}}
\date{}
\author[1]{Minglong Qin}
\author[2,3]{Penghui Yao}
\author[2]{Mingnan Zhao}
\author[2]{Haigang Zhou}

\affil[1]{Centre for Quantum Technologies, National University of Singapore, Singapore}
\affil[ ]{\texttt{mlqin6@gmail.com}}
\affil[2]{State Key Laboratory for Novel Software Technology, Nanjing University, Nanjing 210023, China}
\affil[ ]{\texttt{phyao1985@gmail.com, mingnanzh@gmail.com, hgzhou2003@outlook.com}}
\affil[3]{Hefei National Laboratory, Hefei 230088, China}

\begin{document}

\maketitle

\allowdisplaybreaks

\vspace*{-2em}
\begin{abstract}
    We initiate the study of constructing explicit pseudorandom generators for thresholds of halfspaces
with seed length polylogarithmic in the number of halfspaces. This class of functions lies at the frontier of circuit complexity~\cite{ChenTalWang2026}.
We show that the generator designed by O'Donnell, Servedio, and Tan
for polytopes \cite{OST22} also fools this broader class.

To analyze the generator,
we develop a threshold-specific smooth approximation framework
based on a Bentkus-type mollifier.
We prove derivative bounds for this mollifier
and also establish a Boolean anticoncentration theorem
for thresholds of halfspaces
via a random thinning argument.
These ingredients imply that the generator \(\delta\)-fools
every \(k\)-out-of-\(m\) threshold of \(m\) halfspaces over \(\{-1,1\}^n\)
with seed length
\(\tildeo(\kappa^{6+2\varepsilon}\log^{6+2\varepsilon}\!m\cdot
\delta^{-(2+2\varepsilon)}\log n)\), for any arbitrarily small constant
\(\varepsilon>0\), where \(\kappa=\min\{k,m-k+1\}\).
The random thinning argument also yields bounds on the
noise sensitivity and Gaussian surface area for thresholds of halfspaces,
leading to learning algorithms under both the uniform and Gaussian distributions.

\end{abstract}

\vspace{0.2em}
\setcounter{tocdepth}{1}
\hypersetup{bookmarksdepth=2}
\begingroup
\setstretch{0.8}
\tableofcontents
\par
\endgroup

\thispagestyle{empty}
\clearpage

\setcounter{page}{1}
\section{Introduction}
\label{sec:intro}
Derandomization has long been a central theme in computational complexity theory.
It asks whether randomized algorithms can be simulated deterministically
without significantly compromising efficiency.
A fruitful approach is to replace truly random bits
used by randomized algorithms
with pseudorandom bits produced by
a \emph{pseudorandom generator} (\prg) from a much shorter random seed.
If the seed is short enough, enumerating all possible seeds becomes feasible
and yields a deterministic simulation with essentially the same behavior;
see \cite{Vad12} for a survey.
For a class \(\mathcal{C}\) of Boolean functions over \(\pmcube n\),
a generator \(G\) with seed length \(r\ll n\)
is said to \(\delta\)-fool \(\mathcal{C}\) if,
for every \(f\in\mathcal{C}\),
\(f(G(U_r))\) and \(f(U_n)\) differ by at most \(\delta\) in expectation,
where \(U_r\) and \(U_n\) denote the uniform distributions
over \(\pmcube r\) and \(\pmcube n\).
The same notion applies to other distributions over the inputs,
such as the standard Gaussian distribution on \(\R^n\),
by replacing \(U_n\) accordingly.
Seminal works constructed \(\prg\)s
for constant-depth circuits \cite{AW85,Nis91,LVW93}
and space-bounded computation \cite{BNS89,Nis92,NZ93},
and developed general frameworks for pseudorandomness \cite{NW94,INW94,NZ96}.
Since these foundational results,
a substantial body of work has sought
explicit \(\prg\)s for specific, well-structured classes of functions,
including $\CNF$ and $\DNF$ formulas \cite{Baz09,Raz09},
halfspaces \cite{KM15,GKM18},
and polynomial threshold functions \cite{MZ13,OST20,KM22}.

Among these natural classes,
halfspaces have been studied extensively.
A \emph{halfspace}, or \emph{linear threshold function}, is a basic Boolean function
that tests whether a weighted sum of its input bits exceeds a given threshold.
Halfspaces play a fundamental role across many fields,
including machine learning \cite{Ros58,CV95},
theoretical computer science \cite{Hu65,Mur71},
and game theory \cite{TZ92,TZ99}.
Even this simple class has led to a rich line of work in pseudorandomness,
yielding \(\prg\)s that \(\delta\)-fool halfspaces
in both the Boolean \cite{Ser06,DGJ+10,MZ13,GKM18}
and the Gaussian \cite{KM15} settings,
with seed length \(O(\log(n/\delta))\) up to a loglog factor.
One direction generalizes the degree, replacing the linear form
by a low-degree polynomial to obtain \emph{polynomial threshold functions} (\ptf s).
In the Boolean setting,
\cite{MZ13} gave an explicit \(\prg\)
that fools degree-\(d\) \ptf{}s under the uniform distribution
with seed length $(d/\delta)^{O(d)}\cdot\log n$.
In the Gaussian setting, a series of works
\cite{DKN10,Kan11a,Kan11b,Kan12,MZ13,Kan14,Kan15,OST20,KM22}
has progressively improved \(\prg\)s for \ptf{}s,
achieving seed length polynomial in
\(d\), \(1/\delta\), and \(\log n\) \cite{OST20,KM22}.

Beyond single threshold functions,
\emph{functions of halfspaces} have received considerable attention
as another natural generalization \cite{DKN10,GOWZ10,CDS19,OST22,YZ25}.
A particularly important instance is the class of intersections of halfspaces.
An \emph{intersection of halfspaces},
or equivalently a \emph{polytope},
is a Boolean function that outputs $1$ if and only if
the input satisfies all halfspaces in a given collection.
Such functions are fundamental in optimization \cite{GLS88},
computational geometry \cite{Zie95},
and learning theory \cite{KOS02,KOS08}.
Early works initiated this line by constructing \(\prg\)s for restricted settings,
such as intersections of \emph{regular} and \emph{low-weight} halfspaces \cite{HKM13,ST17}.
A significant advance was made by O'Donnell, Servedio, and Tan,
who constructed \(\prg\)s for intersections of $m$ halfspaces over the Boolean cube
with seed length polynomial in \(\log m\), \(1/\delta\), and \(\log n\) \cite{OST22}.
A comparable result on \(\prg\)s for polytopes in the Gaussian setting
was given in \cite{CDS19}.
These results naturally raise an intriguing question:
\begin{quote}
    \itshape
    Beyond intersections, can we fool broader combinations of halfspaces,
    thereby capturing a richer family of geometric concepts?
\end{quote}

One such combination is the \emph{\(k\)-out-of-\(m\) threshold of halfspaces},
a Boolean function that takes \(m\) halfspaces as constraints
and outputs $1$ if and only if the input satisfies at least \(k\) of them.
This class properly generalizes polytopes,
which correspond to the case \(k=m\).
Other choices of \(k\) yield various functions of halfspaces,
including their union when \(k=1\) and their majority when \(k=m/2\).
Apart from derandomization, the importance of this class
also stems from its central role in computational complexity theory.
In circuit complexity, where halfspaces are precisely linear threshold gates,
a threshold of halfspaces is a depth-two threshold circuit \cite{GHR92,HMPST93}.
Such circuits lie at the heart of the class \(\mathsf{TC}^0\),
for which proving strong lower bounds remains a long-standing challenge~\cite{KaneWilliam2016,ChenTalWang2026}.
In learning theory, they have also served as a testbed for
learning geometric concepts beyond a single halfspace \cite{KOS02}.
These connections position thresholds of halfspaces as a natural target for pseudorandomness,
and make constructing explicit \(\prg\)s for them a compelling next step beyond polytopes.

In this work, we study explicit constructions of \(\prg\)s
for thresholds of halfspaces over the Boolean cube.
The result of Gopalan, O'Donnell, Wu, and Zuckerman~\cite{GOWZ10}
for monotone functions of halfspaces applies to this class,
but the resulting seed length has superlinear dependence
on the number \(m\) of halfspaces.
On the other hand, known explicit \(\prg\)s for intersections of halfspaces
achieve polylogarithmic dependence on \(m\) \cite{HKM13,ST17,OST22},
but their analyses do not directly extend to thresholds of halfspaces.
Thus, prior work leaves open whether one can obtain
polylogarithmic dependence on \(m\) for this broader class.
See \cref{sec:discussion} for a more detailed discussion.

\subsection{Main Results}
This work studies the construction of explicit \(\prg\)s
for $k$-out-of-$m$ thresholds of halfspaces (see \cref{def:k-out-of-m-threshold}).
We show that the \(\prg\) designed by O'Donnell, Servedio, and Tan for polytopes
(i.e., intersections of halfspaces) \cite{OST22} also fools this broader class.
Our main result is as follows.
Here and throughout, \(\tildeo(\cdot)\) hides polylogarithmic factors in the relevant parameters.

\begin{theorem}[\cref{thm:main-prg}, informal]
    There is an explicit $\prg$ that \(\delta\)-fools all \(k\)-out-of-\(m\) thresholds of
    halfspaces over \(\pmcube n\)
    with seed length
    \(\tildeo(\kappa^{6+2\varepsilon}\log^{6+2\varepsilon}\!m\cdot
    \delta^{-(2+2\varepsilon)}\log n)\), for any arbitrarily small constant
    \(\varepsilon>0\), where \(\kappa=\min\{k,m-k+1\}\).
\end{theorem}

For simplicity, we henceforth assume \(1\le k\le m/2\)
unless specified otherwise, so that \(\kappa=k\),
since the case \(k>m/2\) is equivalent to
\((m-k+1)\)-out-of-\(m\) thresholds under complementation.
For $k = \polylog(m)$, the seed length of this $\prg$ is
polylogarithmic in both $m$ and $n$.
This extends the $\prg$ for polytopes in \cite{OST22}
from intersections of halfspaces to thresholds of halfspaces,
while preserving the same qualitative bound on seed length.
Indeed, intersections of \(m\) halfspaces correspond to the case \(k=m\),
or equivalently \(k=1\).
A direct application of the result of
Gopalan, O'Donnell, Wu, and Zuckerman for monotone functions of halfspaces \cite{GOWZ10}
to \(k\)-out-of-\(m\) thresholds
gives seed length
\(O((m\log(m/\delta)+\log n)\log(m/\delta))\).
Thus, in the regime \(k=\polylog(m)\),
our result improves the dependence on \(m\) from superlinear to polylogarithmic.

The main difficulty is that a threshold of halfspaces is no longer described by a single orthant,
so the existing analysis for polytopes does not apply directly.
Our contribution is to develop a threshold-specific analysis of the same generator.
This requires two additional ingredients.
First, we construct a suitable Bentkus-type mollifier (see \cref{def:mollifier})
for threshold functions,
which smoothly approximates the discontinuous Boolean threshold function
except within a small neighborhood of the boundary.
We also prove derivative bounds for this mollifier (see \cref{thm:scaled})
that are strong enough for the hybrid argument.
Second, we establish a Boolean anticoncentration bound for thresholds of halfspaces
(see \cref{thm:threshold-ost-anticonc}).
This bound shows that the Boolean cube places little mass near the boundary of a threshold of halfspaces,
and it allows us to reduce fooling the original threshold function to fooling the smooth approximator.
The proof uses a \emph{random thinning argument},
which relates anticoncentration for thresholds of halfspaces
to anticoncentration for an \(\OR\) of a random subset of these halfspaces.
The latter follows from the Boolean anticoncentration theorem for polytopes
\cite{OST22}.
We will elaborate on these ingredients later in the proof overview.

The random thinning argument also applies to other boundary-complexity measures,
allowing us to lift known bounds from intersections of halfspaces to thresholds of halfspaces.
In particular, we obtain the following bounds on noise sensitivity and Gaussian surface area.

\begin{theorem}[\cref{thm:threshold-ns-gsa}, informal]
    Let \(F\) be a \(k\)-out-of-\(m\) threshold of halfspaces, where $k\leq m/2$. Then
    $\ns_{\delta}(F) = O( k \sqrt{\delta\cdot \log(m/k)} ) $ and
    \(\gsa(F) = O(k \sqrt{\log(m/k)})\).
\end{theorem}

These structural bounds,
together with known learning results based on noise sensitivity and Gaussian surface area
\cite{Man94,KOS02,KKMS08,KOS08,PSW26}, yield the following consequences for the learnability of thresholds of halfspaces.

\begin{corollary}[\cref{cor:threshold-learning}, informal]
    The class of \(k\)-out-of-\(m\) thresholds of halfspaces, where $k\leq m/2$, is learnable with
    accuracy \(\varepsilon\): 
    in time \(n^{O(k^2\log(m/k)/\varepsilon^4)}\)
    under the uniform distribution
    and \(n^{\tildeo(k^2\log(m/k)/\varepsilon^2)}\)
    under the Gaussian distribution in the agnostic model;
    in time \(n^{O(k^2\log(m/k)/\varepsilon^2)}\)
    under either the uniform or Gaussian distribution in the PAC model.
\end{corollary}

Furthermore, we show that the upper bounds on noise sensitivity and Gaussian surface area
are nearly tight in their dependence on \(k\) up to a $\log m$ factor.

\begin{theorem}
    [\cref{prop:threshold-ns-lower-bound,prop:threshold-surface-lower-bound}, informal]
    There exist \(k\)-out-of-\(m\) thresholds of halfspaces \(F_1\) and \(F_2\), where $k\leq m/2$, such that
    $\ns_{\delta}(F_1) = \Omega \bigl(\sqrt{\delta}(k+\sqrt{k\cdot\log(m/k)}) \bigr)$ and
    \(\gsa(F_2) = \Omega \bigl(k+\sqrt{k\cdot\log(m/k)} \bigr)\).
\end{theorem}

The same linear dependence on \( k \)
also appears in a similar lower bound for the Boolean anticoncentration probability
of thresholds of halfspaces (see \cref{thm:boolean-threshold-band-lower}).
These lower bounds together indicate a lower bound on the anticoncentration and thus suggest that
new ideas and techniques may be needed to
extend the analysis framework from a very successful line of work
\cite{HKM13,ST17,OST22,AY22} to thresholds of halfspaces
and achieve a \(\prg\) whose seed length depends only polylogarithmically on \(k\).

\subsection{Proof Overview}

Our work uses the generator from~\cite{OST22}, which is originally designed to fool polytopes.
We begin by briefly recalling the analysis framework in~\cite{OST22}, and then introduce new ingredients to handle thresholds of halfspaces.

\subsubsection{Framework for Polytopes from \cite{OST22}}

Define the orthant indicator \(\ort_b^{\wedge}(y)\) to
output \(1\) on input \(y\in\R^m\)
if \(y_i \le b_i\) for all \(i\in[m]\),
and output \(0\) otherwise.
To prove that a generator $G$ fools every polytope,
one needs to show that \(\E_{z\sim G}[\ort_b^{\wedge}(Az)] \approx \E_{u\sim\pmcube{n}} [\ort_b^{\wedge}(Au)]\) for every matrix \(A\in\R^{m\times n}\) and every vector \(b\in\R^m\),
where $z$ is the output of \(G\) on a uniform seed
and $u$ is uniformly distributed over $\pmcube{n}$.
The analysis in~\cite{OST22} begins by reducing the problem to \emph{standardized} matrices.
Each row of such a matrix consists of
a \emph{sparse} head containing the large coordinates and a \emph{regular} tail
in which no individual coordinate is too large.
The remainder of the proof proceeds in two main steps.
\vspace{-1em}
\paragraph{Step 1: Designing and fooling a smooth approximator.}
The discontinuous orthant indicator \(\ort_b^{\wedge}\) is replaced by
a smooth Bentkus mollifier $\mollifier_b^{\wedge}$,
which approximates \(\ort_b^{\wedge}\) outside a narrow neighborhood of its boundary.
It is therefore sufficient to prove two properties:
first, the expectations of \(\mollifier_b^{\wedge}(Az)\)
and \(\mollifier_b^{\wedge}(Au)\) are close;
and second, \(Au\) is unlikely to lie near the boundary of \(\ort_b^{\wedge}\).

The first property is established via a hybrid argument.
One defines a sequence of hybrid random variables $x_0=u, x_1, \ldots, x_L=z$
by replacing the blocks of coordinates of \(u\) with those of \(z\) one at a time.
It then remains to show that
$\E[\mollifier_b^{\wedge}(Ax_{\ell})] \approx\E[\mollifier_b^{\wedge}(Ax_{\ell-1})]$
for each $\ell\in[L]$.

This is done by Taylor expanding the mollifier \(\mollifier_b^{\wedge}\)
to order \(d\) in the coordinates of the bucket being replaced.
The expansion consists of the Taylor terms of degree less than \(d\),
together with a degree-\(d\) remainder.
After fixing all coordinates outside the bucket being replaced,
each low-degree Taylor term becomes a function only of the variables in the bucket.
One key observation is that 
this function can be expressed in terms of small-width \(\CNF\) formulas
due to the sparsity of the head part of each row.
These formulas are fooled by the \(\CNF\) generator used as one component of \(G\).
Therefore, the low-degree terms have nearly the same expectations
under the two adjacent hybrids \(x_{\ell-1}\) and \(x_\ell\).
For the remainder, it suffices to show that
its expectation is small using derivative bounds for the Bentkus mollifier~\cite{Ben90},
together with the regularity of the tail
and moment bounds for the bucket variables.

\vspace{-1em}
\paragraph{Step 2: Establishing Boolean anticoncentration.}
To transfer the closeness from the smooth mollifier back to the original orthant indicator, 
one needs to show that \(Au\) is unlikely to lie
near the boundary of the orthant \(\ort_b^{\wedge}\).
For standardized matrices \(A\), \cite{OST22} proves a high-dimensional Littlewood-Offord-type 
anticoncentration theorem showing that the Boolean cube assigns at most \(O(\Lambda\sqrt{\log m})\) mass
in a width-\(\Lambda\) neighborhood of the boundary of any \(m\)-facet polytope.
It is the Boolean counterpart of the Gaussian anticoncentration bound used in \cite{HKM13},
which follows from Nazarov's bound on the Gaussian surface area of polytopes \cite{Naz04,KOS08}.

\subsubsection{Extending to Thresholds of Halfspaces}
Extending the framework above from polytopes to thresholds of halfspaces
presents two main challenges.
First, 
a \(k\)-out-of-\(m\) threshold of halfspaces can be written as \(\ort_{k,b}(Ax)\),
where the threshold indicator \(\ort_{k,b}(y)\) outputs \(1\) on input \(y\in\R^m\)
if \(y_i \leq b_i\) for at least \(k\) indices \(i\in[m]\),
and \(0\) otherwise.
In other words, \(\ort_{k,b}\) is the \(\OR\) of \(\binom{m}{k}\) orthant indicators
of the form \(\ort_{b_S}^{\wedge}\), each corresponding to a \(k\)-facet polytope.
A direct reduction to the polytope case would therefore have to
handle \(\binom{m}{k}\) such orthant indicators,
which would lead to a seed length with polynomial dependence on
\(\binom{m}{k}\).
To avoid this loss, we develop a direct analysis of the same generator tailored to thresholds of halfspaces.
Second, the analysis requires a Boolean anticoncentration bound
for the boundary of $\ort_{k,b}$,
which is more complicated than the boundary of a polytope.
Thus, such a bound does not follow directly from the corresponding result
for polytopes and therefore requires a separate argument.

We first observe that the reduction to standardized matrices continues
to hold for thresholds of halfspaces.
In fact, this reduction applies more generally to any Boolean function of halfspaces
(see \cref{cor:threshold-outer-reduction}).
After this reduction, the main task is to
show \(\E_{z\sim G}[\ort_{k,b} (Az)] \approx \E_{u\sim\pmcube{n}} [\ort_{k,b} (Au)]\)
for all standardized matrices \(A\).

\vspace{-1em}
\paragraph{A Bentkus-type mollifier for $\ort_{k,b}$}
As in the polytope case, we prove this closeness through a smooth approximation.
Notice that \(\ort_{k,b}(y)\) admits the representation
\[
	\ort_{k,b}(y)
        ~=~
        1-\prod_{S\in \binom{[m]}{k}}
        \left[
        1-\prod_{i\in S} \ind[y_i\le b_i]
        \right]\enspace.
\]
This naturally leads us to define the smooth approximator
\(\mollifier_{k,b}\) (see \cref{def:mollifier})
by replacing each one-dimensional indicator $\ind[y_i\le b_i]$
with its Gaussian-smoothed version $\widetilde{\ind}_{b_i}(y_i)$.
In particular, when \(k=m\), this construction recovers the Bentkus mollifier
\(\mollifier_b^{\wedge}\) used in \cite{OST22}.
With a suitable choice of parameters,
we prove that \(\mollifier_{k,b}\) approximates \(\ort_{k,b}\) except near
the boundary of \(\ort_{k,b}\) (see \cref{lem:approximator}).
Thus, by a sandwiching argument (see \cref{lem:sandwich}),
the desired closeness for \(\ort_{k,b}\) follows once we prove two facts:
the smooth approximator is fooled by \(G\),
and a Boolean anticoncentration bound holds for \(\ort_{k,b}\).

\vspace{-1em}
\paragraph{Fooling the mollifier \(\mollifier_{k,b}\)}
To show that
\(\E_{z\sim G}[\mollifier_{k,b} (Az)] \approx \E_{u\sim\pmcube{n}} [\mollifier_{k,b} (Au)]\),
we follow the same hybrid argument as in the polytope case and
Taylor expand \(\mollifier_{k,b}\).
There are two differences.
First, the low-degree Taylor terms are still expressible by \(\CNF\) formulas
but the width increases by a factor of \(k\) (see \cref{lem:cnf-decomposition}).
Second, the treatment of the remainder term requires derivative bounds for
the new mollifier \(\mollifier_{k,b}\).

The main obstacle lies in proving a good derivative bound for \(\mollifier_{k,b}\),
since \(\mollifier_{k,b}\) is defined through a product over all \(k\)-subsets of \([m]\).
A naive bound would introduce a factor depending on \(\binom{m}{k}\),
which would be too large for our purposes.
We avoid this overhead by exploring the derivative structure of the mollifier more carefully.
After normalizing the parameters,
$\mollifier_{k,b}(y)$ has the form \( \prod_{S\in\binom{[m]}{k}} (1-\prod_{i\in S}\gcdf(y_i))\),
where $\gcdf$ is the Gaussian cumulative distribution function.
By the product rule, derivatives of \(\mollifier_{k,b}\) are governed by sums,
over \(S\in\binom{[m]}{k}\), of derivatives of the clause functions \(\prod_{i\in S}\gcdf(y_i)\).
To handle this, we first prove a refined derivative bound for a single \(\gcdf(y_i)\)
(see \cref{lem:oned}),
then lift this estimate to the product \(\prod_{i\in S}\gcdf(y_i)\)
(see \cref{lem:clause}),
and finally sum these clausewise bounds over all \(S\)
(see \cref{lem:moment}).
A direct summation would introduce a factor \(\binom{m}{k}\),
whereas our estimate incurs only an additional polynomial dependence on \(k\) and \(\log(m/k)\).

\vspace{-1em}
\paragraph{Boolean anticoncentration for \(\ort_{k,b}\)}
A Boolean anticoncentration bound for polytopes was established in~\cite{OST22} through a delicate analysis of
how the Boolean cube interacts with the facets of high-dimensional polytopes.
In our setting, we build on this theorem
and transfer it to thresholds of halfspaces
through a \emph{random thinning argument}.

The goal is to bound the probability that \(Au\) lies
in a width-\(\Lambda\) neighborhood of the boundary of \(\ort_{k,b}\),
namely \(\ort_{k,b+\Lambda\cdot\onevec}\setminus\ort_{k,b-\Lambda\cdot\onevec}\).
Suppose this event holds for some \(u\).
Let \(S_+,S_-\subseteq [m]\) index the halfspaces satisfied by \(u\)
at thresholds \(b+\Lambda\cdot\onevec\) and
\(b-\Lambda\cdot\onevec\), respectively.
Then \(|S_+|\ge k\) and \(|S_-|\leq k-1\).
The thinning argument randomly selects a subset \(R\subseteq[m]\)
by keeping each halfspace independently with probability \(p=1/k\).
Therefore, the random set \(R\) contains at least one index from $S_+$
and no index from \(S_-\) with probability at least
$(1-p)^{|S_- |}\cdot(1-(1-p)^{|S_+\setminus S_- |}) \geq p(1-p)^{k-1} = \Omega(1/k)$.
When this happens, \(Au\) lies in the width-\(\Lambda\) boundary of
the \(\OR\) of the halfspaces indexed by \(R\).
Thus, the probability that \(Au\) lies in the width-\(\Lambda\) boundary of \(\ort_{k,b}\)
is at most \(O(k)\) times the probability that \(Au\) lies in the width-\(\Lambda\) boundary
of the \(\OR\) of a random subset of the halfspaces.
The latter is bounded by the Boolean anticoncentration theorem for polytopes \cite{OST22},
since the \(\OR\) of halfspaces is, after negation,
an intersection of the complementary halfspaces.
The random set \(R\) has size \(O(m/k)\) in expectation,
so averaging over \(R\) yields the \(O(k\cdot\Lambda\sqrt{\log(m/k)})\)
anticoncentration bound for \(\ort_{k,b}\)
(see \cref{thm:threshold-ost-anticonc}).

This random thinning argument may be of independent interest.
We further observe that it can be used to transfer other boundary-complexity bounds
from intersections of halfspaces to thresholds of halfspaces.
Applying it to the known bounds for intersections gives
our bounds on noise sensitivity and Gaussian surface area
of thresholds of halfspaces (see \cref{thm:threshold-ns-gsa}).

\subsection{Related Work}
\label{sec:related-work}
We provide an overview of related prior work.
\Cref{tab:prior-work} summarizes the \(\prg\)s most relevant to our setting.
In the table, \(m\) denotes the number of halfspaces,
\(n\) the input dimension, and \(\delta\) the error of the \(\prg\).

\begin{table}[!b]
    \centering
    \small
    \renewcommand{\arraystretch}{1.1}
    \newlength{\priorworktablewidth}
    \setlength{\priorworktablewidth}{0.96\linewidth}
    \newcommand{\priorworkbodyheight}{5.5ex}
    \newcommand{\priorworkedgeheight}{5.5ex}
    \makeatletter
    \newcommand{\priorworkcellbox}[5]{\raisebox{#4}[#1][0pt]{\vbox to #2{\vfil
                \hbox to #3{\hfil\parbox{#3}{\centering\strut #5\strut}\hfil}\vfil
            }}}
    \newcommand{\priorworkcellwithshift}[4][\priorworkbodyheight]{\priorworkcellbox{#1}{#1}{#2}{#3}{#4}}
    \newcommand{\priorworkcell}[3][\priorworkbodyheight]{\priorworkcellwithshift[#1]{#2}{-.35\dp\@arstrutbox}{#3}}
    \makeatother
    \newcommand{\priorworkref}[2][\priorworkbodyheight]{\priorworkcell[#1]{0.16\priorworktablewidth}{#2}}
    \newcommand{\priorworkclass}[2][\priorworkbodyheight]{\priorworkcell[#1]{\dimexpr0.42\priorworktablewidth-2\tabcolsep\relax}{#2}}
    \newcommand{\priorworkseed}[2][\priorworkbodyheight]{\priorworkcell[#1]{\dimexpr0.42\priorworktablewidth-2\tabcolsep\relax}{#2}}
    \makeatletter
    \newcommand{\priorworkheadref}[1]{\priorworkcellbox{\dimexpr\priorworkedgeheight+.33pt\relax}{\priorworkedgeheight}{0.16\priorworktablewidth}{-.30625\dp\@arstrutbox}{#1}}
    \newcommand{\priorworkheadclass}[1]{\priorworkcellbox{\dimexpr\priorworkedgeheight+.33pt\relax}{\priorworkedgeheight}{\dimexpr0.42\priorworktablewidth-2\tabcolsep\relax}{-.30625\dp\@arstrutbox}{#1}}
    \newcommand{\priorworkheadseed}[1]{\priorworkcellbox{\dimexpr\priorworkedgeheight+.33pt\relax}{\priorworkedgeheight}{\dimexpr0.42\priorworktablewidth-2\tabcolsep\relax}{-.30625\dp\@arstrutbox}{#1}}
    \makeatother
    \caption{Related work on \(\prg\)s for functions of halfspaces.}
    \label{tab:prior-work}
    \vspace{6pt}
    \begin{tabular}{@{}ccc@{}}
        \toprule
        \priorworkheadref{Reference}
        & \priorworkheadclass{Function class}
        & \priorworkheadseed{Seed length}
        \\ \midrule

        \priorworkref{\cite{GOWZ10}}
        & \priorworkclass{Monotone functions of \(m\) halfspaces}
        & \priorworkseed{\(O\!\left(
            (m\log(m/\delta)+\log n)\cdot\log(m/\delta)
          \right)\)}
        \\

        \priorworkref{\cite{HKM13}}
        & \priorworkclass{Intersections of \(m\) \(\tau\)-regular halfspaces}
        & \priorworkseed{
            \(O((\log n\cdot \log m) /\tau)\)\\
            {\footnotesize for \(\tau \le
            \delta^5/(\log^{8.1}m\cdot\log(1/\delta))\)}
          }
               
        \\

        \priorworkref{\cite{ST17}}
        & \priorworkclass{Intersections of \(m\) weight-\(t\) halfspaces}
        & \priorworkseed{\(\operatorname{poly}
            (\log n,\log m,t,1/\delta)\)}
        \\

        \priorworkref{\cite{OST22}}
        & \priorworkclass{Intersections of \(m\) arbitrary halfspaces}
        & \priorworkseed{\(\operatorname{poly}
            (\log m,1/\delta)\cdot\log n\)}
        \\

                \priorworkref{\cite{CDS19}}
        & \priorworkclass{Intersections of \(m\) arbitrary halfspaces}
        & \priorworkseed{\(\operatorname{poly}(\log m,1/\delta)+
        O\!\left(
            \log n
          \right)\)}
        \\

        \priorworkref{\cite{CDS19}}
        & \priorworkclass{Arbitrary functions of \(m\) halfspaces}
        & \priorworkseed{\(\operatorname{poly}(m,1/\delta)+
        O\!\left(
            \log n
          \right)\)}
          \\
        \midrule
        
        \priorworkref[\priorworkedgeheight]{\textbf{Our work}}
        & \priorworkclass[\priorworkedgeheight]{\(k\)-out-of-\(m\) threshold of halfspaces}
        & \priorworkseed[\priorworkedgeheight]{\(\poly{\kappa,\log m,1/\delta}\cdot \log n\)\\where $\kappa=\min\{k,m-k+1\}$}
        \\

        \bottomrule
    \end{tabular}
\end{table}

Gopalan, O'Donnell, Wu, and Zuckerman~\cite{GOWZ10} considered
fooling \emph{monotone functions} of $m$ halfspaces
under \emph{product distributions}, including the uniform distribution over the Boolean cube
and the Gaussian distribution.
Building on the generator of Meka and Zuckerman~\cite{MZ13},
they constructed a \(\prg\) that \(\delta\)-fools this class with seed length
\(O\big((m\log(m/\delta)+\log n)\log(m/\delta)\big)\),
which grows superlinearly with the number of halfspaces.

Remarkably, for intersections of halfspaces, this dependence on $m$ can be improved substantially.
Harsha, Klivans, and Meka~\cite{HKM13} constructed \(\prg\)s for intersections of \emph{regular} halfspaces,
where a halfspace is regular if no single coordinate has large influence.
Their construction works under \emph{proper} and \emph{hypercontractive} distributions,
including both the uniform and Gaussian distributions.
They showed that a slightly modified Meka-Zuckerman generator \cite{MZ13}
fools intersections of \(m\) \(\tau\)-regular halfspaces with seed length
\(O((\log n\cdot \log m)/\tau)\),
provided \(\tau\) is below a certain threshold.
At the heart of their argument is an \emph{invariance principle} for intersections of regular halfspaces
obtained by extending the classical Lindeberg method for proving central limit theorems~\cite{Lin22}.
Servedio and Tan~\cite{ST17} addressed intersections of \emph{weight}-$t$ halfspaces over the Boolean cube,
whose coefficients are integers of absolute value at most \(t\).
They constructed an explicit \(\prg\) that \(\delta\)-fools
any intersection of \(m\) such halfspaces with seed length
\(\poly{\log n,\log m,t,1/\delta}\),
by extending the approach of \cite{HKM13}
and combining it with results on fooling \(\CNF\) formulas
\cite{Baz09,Raz09}.

Both of these results require structural restrictions on the halfspaces
(e.g., regularity or bounded weight).
A major breakthrough came from O'Donnell, Servedio, and Tan \cite{OST22},
who removed these restrictions
and constructed a \(\prg\) for intersections of \emph{arbitrary} halfspaces over the Boolean cube
with seed length \(\operatorname{poly}(\log m,1/\delta)\cdot\log n\).
Their analysis combines a new invariance principle for intersections of general halfspaces
with a Littlewood-Offord-type anticoncentration inequality
for polytopes over the Boolean cube.
A parallel line in the Gaussian setting is due to
Chattopadhyay, De, and Servedio \cite{CDS19}.
Inspired by the \emph{Johnson-Lindenstrauss transform} \cite{JL86,KMN11},
they constructed $\prg$s that \(\delta\)-fool intersections of \(m\) halfspaces with seed length
\(\poly{\log m,1/\delta}+O(\log n)\),
and arbitrary functions of \(m\) halfspaces with seed length
\(\poly{m,1/\delta}+O(\log n)\).

In summary, prior work on \(\prg\)s for compositions of halfspaces exhibits
a clear contrast in the dependence on the number \(m\) of halfspaces.
For intersections of halfspaces,
a sequence of works \cite{HKM13, ST17, CDS19, OST22} achieved seed length
with only polylogarithmic dependence on \(m\).
In contrast, for broader compositions,
such as monotone functions of halfspaces \cite{GOWZ10} or
arbitrary functions of halfspaces \cite{CDS19,YZ25},
the known seed lengths have polynomial dependence on \(m\).

\subsection{Discussion and Future Directions}
\label{sec:discussion}
This work identifies an intermediate regime beyond polytopes
in which such a polylogarithmic dependence on \(m\) can still be obtained:
\(k\)-out-of-\(m\) thresholds of halfspaces for \(k\) up to \(\polylog(m)\).
The generator is inherited from \cite{OST22},
but the analysis is not a black-box consequence of the polytope case:
a direct decomposition of a \(k\)-out-of-\(m\) threshold
into \(\binom{m}{k}\) orthants would lose the desired dependence on \(m\).
Our contribution is to replace this decomposition 
by a direct analysis of the threshold set,
based on a threshold-specific mollifier and a
corresponding Boolean anticoncentration bound for its boundary.
Our results suggest several directions for future work.
\begin{itemize}[topsep=2pt,parsep=0pt,partopsep=0pt,leftmargin=1.5em, itemsep=2pt]
    \item One immediate question is whether the polynomial dependence on \(k\)
    in our seed length can be reduced.
    It would be interesting to determine
    whether this dependence is necessary,
    or whether thresholds of halfspaces admit generators
    whose seed length depends only polylogarithmically on \(m\) and \(1/\delta\)
    for a larger range of \(k\), such as majority of halfspaces.
    Solving this question may require sharper derivative bounds
    for the threshold mollifier
    or a different smooth approximation
    that avoids the current dependence on the threshold parameter.
    \item It is also natural to ask whether polylogarithmic dependence on \(m\)
    can be obtained for other compositions of halfspaces,
    such as $\xor$s or more general symmetric functions of halfspaces.
    One should not expect such a bound for arbitrary outer functions of halfspaces.
    The threshold case suggests a more refined possibility:
    such a bound may still be achievable for outer functions
    for which the set defined by the outer function has boundary complexity
    controlled by a parameter much smaller than its full description size.
    Identifying the right combinatorial or geometric condition
    on the outer function is an interesting problem.
    \item Finally, the random thinning argument may be of independent interest.
    In this work, we use it to transfer Boolean anticoncentration,
    noise sensitivity, and Gaussian surface area
    from intersections of halfspaces to thresholds of halfspaces.
    It would be interesting to see whether similar arguments
    apply to other quantities that are useful
    in learning theory and derandomization,
    or to analogous questions under more general product distributions.
\end{itemize}

\paragraph{Organization.}
The remainder of the paper is organized as follows.
\Cref{sec:preliminaries} introduces the notation and preliminaries used throughout.
In \cref{sec:smooth-approx}, we introduce the Bentkus-type mollifier for thresholds of halfspaces
and prove the sandwiching lemma that reduces fooling the threshold indicator
to fooling the smooth mollifiers.
\Cref{sec:anti} derives the Boolean anticoncentration bound,
as well as the noise sensitivity and Gaussian surface area,
which imply the learning consequences.
In \cref{sec:prg}, we instantiate the \(\prg\) of O'Donnell, Servedio, and Tan
with our parameters and prove the main pseudorandomness theorem for thresholds of halfspaces.
We prove the derivative bounds for the Bentkus-type mollifier in \cref{sec:derivative}.
Some proofs are deferred to the appendix.

\paragraph{Acknowledgment.}
The authors thank Haonan Zhang for helpful discussions on the random thinning argument
and its applications to the Boolean anticoncentration and Gaussian surface area bounds.
MQ was supported by the National Research Foundation, Singapore,
through the National Quantum Office, hosted in A*STAR,
under its Centre for Quantum Technologies Funding Initiative (S24Q2d0009).
PY, MZ and HZ were supported by the National Natural Science Foundation of China (Grant Nos. 62332009 and 12347104), the Quantum Science and Technology-National Science and Technology Major Project (Grant No. 2021ZD0302901), the NSFC/RGC Joint Research Scheme (Grant No. 12461160276), the Natural Science Foundation of Jiangsu Province (No. BK20243060), the Fundamental and Interdisciplinary Disciplines Breakthrough Plan of the Ministry of Education of China (No. JYB2025XDXM118), the ``111 Center'' (No. B26023), and the Fundamental Research Funds for the Central Universities (Grant No. 2026300376).
OpenAI's GPT-5.5 was used for language editing and
to help check the clarity and completeness of parts of the proofs.

\section{Preliminaries}
\label{sec:preliminaries}
\subsection{Basic Notation}

We use the following notation throughout the paper.
Let \(\mathbb{N} = \{0, 1, 2, \ldots\}\) denote the set of natural numbers.
For positive \(n\in\mathbb{N}\), let \([n]=\{1,\ldots,n\}\).
We write \(\ind[\cdot]\) for the indicator function.
For integers $k$ and $m$ with $1\le k\le m$,
let \(\binom{[m]}{k}\) be the family of
all subsets of \([m]\) with size \(k\).
For disjoint sets $S_1$ and $S_2$, we write
\(
    S_1 \sqcup S_2
\)
for their disjoint union. More generally, the notation
\(
    S_1 \sqcup \cdots \sqcup S_k = [n]
\)
means that the sets $S_1,\dots,S_k$ are pairwise disjoint and \(\cup_{i=1}^k S_i = [n]\).
For a distribution \(\mathcal{D}\), the notation \(x\sim\mathcal{D}\)
means that \(x\) is drawn from \(\mathcal{D}\).
For a finite set \(S\), we write \(|S|\) for its size,
and the notation \(x\sim S\) means that \(x\) is drawn uniformly from \(S\).

For a vector \(v \in \R^n\) and an integer \(p \geq 1\),
$\norm{v}_p = \left( \sum_{i=1}^n |v_i|^p \right)^{1/p}$
denotes the \(p\)-norm of $v$,
and \(\|v\|_\infty=\max_i |v_i|\).
For \(B\subseteq[n]\), let \(v_B \in \R^{|B|}\) denote
the restriction of \(v\in\mathbb{R}^n\) to the coordinates in \(B\).
For a matrix \(A\in\R^{m\times n}\), $i\in[m]$ and $B\subseteq[n]$,
let \(A_i \in\R^{n}\) denote the \(i\)-th row of \(A\),
and \(A^B\in\mathbb{R}^{m\times |B|}\) denote the restriction of \(A\)
to the columns indexed by \(B\).
We write \(\onevec\) for the all-ones vector
and \(\zerovec\) for the all-zeros vector.

We use \(\log\) to denote the logarithm with base \(\e\),
use \(\tildeo(\cdot)\) to hide polylogarithmic factors in the relevant parameters,
and write \(O_d(\cdot)\) to hide constants depending only on \(d\).
We write \(\mathcal{N}(0,1)\) for the standard Gaussian distribution.
Its probability density function (PDF) \(\gpdf\)
and cumulative distribution function (CDF) \(\gcdf\)
are given by $\gpdf(x)=\frac{1}{\sqrt{2\pi}}\cdot \e^{-x^2/2}$ and $\gcdf(x)=\int_{-\infty}^x \gpdf(t)\,\mathrm{d}t$, respectively.
The gamma function is defined by
$\Gamma(x)=\int_0^\infty t^{x-1}\cdot\e^{-t}\,\mathrm{d}t$
for $x>0$.

\subsection{Regular and Standardized Vectors and Matrices}
We next define what it means for a vector or matrix to be regular or
standardized, following the terminology of \cite{OST22}.
These notions will be used later to decompose each halfspace
into a sparse head part and a regular tail part.

\begin{definition}[Regularity]
    A vector $w \in \R^n$ is $\tau$-\emph{regular}
    if $\abs{w_i} \leq \tau \norm{w}_2$ for all $i \in [n]$,
    and a matrix $W \in \R^{m\times n}$ is $\tau$-\emph{regular}
    if each of its rows is $\tau$-regular.
\end{definition}

\begin{definition}[Standardization]\label{def:ost-standardized}
A vector $w\in\R^n$ is \emph{$(K,\tau)$-standardized}
if there is a partition $[n]=\newsf{Head}\sqcup \newsf{Tail}$
with $|\newsf{Head}|\le K$ such that
the subvector $w_{\newsf{Tail}}$ is $\tau$-regular and $\sum_{i \in \newsf{Tail}} w_i^2 = 1$.
A matrix is \emph{$(K,\tau)$-standardized} if each of its rows is $(K,\tau)$-standardized.
\end{definition}

\subsection{Boolean Functions}

A function $f:\R^n\to \R$ is called \emph{Boolean} if its range is
contained in $\{0,1\}$.
For a class $\mathcal C$ of Boolean functions, we say that a function $f:\R^n\to\R$ is a
\emph{weight-$W$ combination of functions in $\mathcal C$} if
$f=\sum_{\ell} c_{\ell}\cdot f_{\ell}$,
where each $f_{\ell}\in\mathcal C$ and $\sum_{\ell}|c_{\ell}|\le W$.
In particular, when $\mathcal C$ is the class of all Boolean functions,
we call this a weight-$W$ combination of Boolean functions.
A function is called a \emph{$w$-junta} if it depends on at most $w$ coordinates.

\begin{fact}[\citetext{\citealp[Fact 8.8]{OST22}}]\label[fact]{fact:weight-w-combination}
    The following properties are immediate:
    \begin{itemize}
        \item Every function $f:\pmcube n\to[0,1]$ is a weight-$1$ combination of Boolean functions.
        \item Every function $f:\pmcube n\to[-W,W]$ is a weight-$2W$ combination of Boolean functions.
    \end{itemize}
\end{fact}

We will also use two standard boundary measures for Boolean functions.
\begin{definition}
    For a Boolean function $f:\R^n\to\{0,1\}$ and $\delta\in[0,1]$,
    the \emph{noise sensitivity} of $f$ at noise rate $\delta$ is defined by
    \[
    	\ns_\delta(f) ~\coloneqq~ \Pr_{\substack{x\sim\pmcube n,\ y\sim N_\delta(x)}} [f(x)\neq f(y)] \enspace,
    \]
    where $N_\delta(x)$ is the distribution obtained by flipping each bit of $x$ independently with probability $\delta$.
\end{definition}

\begin{definition}
    Let \(f:\R^n\to\{0,1\}\) be a Boolean function, and write
    $K(f) \coloneqq \{x\in\R^n : f(x)=1\}$.
    The \emph{Gaussian surface area} of \(f\) is defined as
    \[
        \gsa(f)
        ~\coloneqq~
        \lim_{\delta\to0^+}
        \frac{
            \operatorname{vol}_{\mathcal N^n}\!
            \bigl(K(f)^\delta\setminus K(f)\bigr)
        }{\delta} \enspace,
    \]
    where, for \(A\subseteq\R^n\),
    $A^\delta\coloneqq\{x\in\R^n : \operatorname{dist}(x,A)\le \delta\}$
    is the \(\delta\)-neighborhood of \(A\), and
    \(\operatorname{vol}_{\mathcal N^n}\!(\cdot)\) denotes standard Gaussian
    probability mass on \(\R^n\).
\end{definition}

In this work, we concern ourselves with a specific class of Boolean functions, namely,
\emph{$k$-out-of-$m$ thresholds of halfspaces}.
We first give the definition of halfspaces.

\begin{definition}
A function $f:\R^n\to\{0,1\}$ is a \emph{halfspace} or \emph{linear threshold function (\ltf)} if there exist
a vector $w\in\R^n$ and $\theta\in\R$ such that
\[
f(x) = \ind[\langle w,x\rangle\leq \theta] \quad\text{for all }\ x \in \R^n \enspace.
\]
We say that \(f\) is a \(\tau\)-\emph{regular}
(\emph{\((K,\tau)\)-standardized}, resp.) halfspace if it admits such
a representation with \(w\) being \(\tau\)-regular
(\((K,\tau)\)-standardized, resp.).
\end{definition}

For \(m\in\N\), we define
\(\AND_m:\{0,1\}^m\to\{0,1\}\) by
$
    \AND_m(y)
    \coloneqq
    \prod_{i=1}^m y_i
$.
Equivalently, \(\AND_m(y)=1\) if and only if \(y_i=1\) for every
\(i\in[m]\).

\begin{definition}
A function \(F:\R^n\to\{0,1\}\) is an \emph{intersection of \(m\) halfspaces} (or \emph{$m$-facet polytope})
if there exist halfspaces \(f_1,\ldots,f_m:\R^n\to\{0,1\}\) such that
\[
    F(x)
    =
    \AND_m(f_1(x),\ldots,f_m(x)) \quad\text{for all }\ x \in \R^n \enspace.
\]
\end{definition}

We will use the following results for intersections of halfspaces.

\begin{theorem}[\cite{Naz04,KOS08,Kan-ns}]
\label{thm:intersection-halfspaces-ns-gsa}
Let \(f\) be an intersection of \(m\) halfspaces.
Then for a universal constant \(C>0\),
$\ns_\delta(f) \le C\sqrt{\delta\log m}$ for every \(\delta\in(0,1)\) 
and $\gsa(f) \le C\sqrt{\log m}$.
\end{theorem}

For $m\in\N$ and $k\in\N$ with $1\le k\le m$,
we define the \emph{$k$-out-of-$m$ threshold function}
$\thr_{m,k}:\{0,1\}^m \to \{0,1\}$ by
$\thr_{m,k}(y) = \ind\!\left[\sum_{i=1}^m y_i\ge k\right]$.
That is, the threshold function $\thr_{m,k}$ outputs $1$
if and only if at least $k$ of its $m$ input bits are $1$.

\begin{definition}\label{def:k-out-of-m-threshold}
    A function $F:\R^n \to \{0,1\}$ is a \emph{$k$-out-of-$m$ threshold of halfspaces}
    if there exist halfspaces $f_1,\dots,f_m:\R^n\to\{0,1\}$ such that
    \[
        F(x) = \thr_{m,k}(f_1(x),\dots,f_m(x)) \quad\text{for all }\ x \in \R^n \enspace.
    \]
    We call \(F\) a \(k\)-out-of-\(m\) threshold of \(\tau\)-\emph{regular}
    (\emph{\((K,\tau)\)-standardized}, resp.) halfspaces if each of the
    halfspaces \(f_1,\dots,f_m\) is \(\tau\)-regular
    (\((K,\tau)\)-standardized, resp.).
\end{definition}

Throughout the paper, \(m\) denotes the number of halfspaces and
\(k\) denotes the threshold parameter.
Unless otherwise specified,
we assume \(1\le k\le m/2\) without loss of generality,
since for larger $k>m/2$, one may pass to the complementary threshold with parameter $m-k+1$.

\subsection{Pseudorandomness}

We begin with the definition of pseudorandom generators.

\begin{definition}
    Let $\CC$ be a class of Boolean functions.
    We say that a function $G: \pmcube r\to\pmcube n$
    is a \emph{pseudorandom generator (\prg)} for $\CC$ with error $\varepsilon$
    over Boolean space $\pmcube n$
    if for any $f\in\CC$,
    \[
    	\abs{
            \E_{u\sim\pmcube{n}} [f(u)] - \E_{s\sim\pmcube{r}} [f(G(s))]
        }
        ~\leq~ \varepsilon \enspace.
    \]
    We call $r$ the \emph{seed length} of the \prg.
    We also say that $G$ \emph{$\varepsilon$-fools} $\CC$.
\end{definition}

We need the following standard pseudorandom objects in our construction.

\begin{definition}
    For a finite set \(\Omega\) and an integer \(r\ge 1\),
    a random string \(Y=(Y_1,\ldots,Y_n)\in\Omega^n\) is \emph{\(r\)-wise uniform} if,
    for every subset \(I\subseteq[n]\) with \(|I|\le r\),
    the marginal \(Y_I\) is uniform over \(\Omega^{|I|}\).
\end{definition}

\begin{definition}
    For an integer \(r\ge 1\), a family \(\mathcal H\) of functions from
    \([n]\) to \([L]\) is an \emph{\(r\)-wise uniform hash family} if,
    for a uniformly random \(h\in\mathcal H\),
    the sequence \((h(1),\ldots,h(n))\) is \(r\)-wise uniform.
\end{definition}

We use the standard explicit constructions of these objects \cite{Vad12}.
An \(r\)-wise uniform string in \(\Omega^n\) can be sampled
using \(O(r\log n+r\log|\Omega|)\) random bits,
and in particular using \(O(r\log n)\) random bits when \(\Omega=\{-1,1\}\).
An \(r\)-wise uniform hash family from \([n]\) to \([L]\) can be sampled
using \(O(r\log(nL))\) random bits.
We also use the following Gopalan--Meka--Reingold generator
for small-width \CNF{} formulas from \cite{gopalan2013dnf}.

\begin{theorem}\label{thm:ost-cnf-prg}
There is an explicit pseudorandom generator $\gmr=\gmr(w,\delta_{\CNF})$
that \(\delta_{\CNF}\)-fools the class of
all width-\(w\) \CNF{} formulas over $\pmcube n$
with seed length
\[
    O\Bigl(
        w^2\cdot\log^2\bigl(w\cdot\log(1/\delta_{\CNF})\bigr)
        +w\cdot\log w\cdot\log(1/\delta_{\CNF})
        +\log\log(n)
    \Bigr)\enspace.
\]
\end{theorem}

\subsection{Derivatives and Multidimensional Taylor Expansion}
For a $\mathcal{C}^d$ function $\psi:\R^m\to\R$, $x\in\R^m$ and
$i_1,\dots,i_d \in [m]$,
we denote the $d$-th order partial derivative of $\psi$ at $x$
in coordinates $i_1,\dots,i_d$ by
\[
    \partial_{i_1,\dots,i_d}\,\psi(x)
    ~\coloneqq~
    \frac{\partial^d \psi(x)}{\partial x_{i_1}\cdots\partial x_{i_d}} \enspace.
\]
We define the $L_1$-norm of the $d$-th derivative of $\psi$ by
\[
	\norm{ \psi^{(d)}\! (x)}_1 ~\coloneqq~ \sum_{i_1,\dots,i_d = 1}^m \abs{ \partial_{i_1,\dots,i_d}\,\psi(x) } \enspace, \qquad
    \norm{ \psi^{(d)}}_1 ~\coloneqq~ \sup_{x \in \R^m} \norm{ \psi^{(d)}\! (x)}_1  \enspace.
\]
We have the following multidimensional Taylor expansion with an explicit error bound.

\begin{fact}\label[fact]{fact:multidim-taylor}
Let \(d\ge 1\) be an integer and \(\psi:\R^m\to\R\) be a \(\mathcal{C}^d\) function.
Then for all \(x,y\in\R^m\),
\[
    \psi(x+y)
    ~=~
    \psi(x)
    ~+~
    \sum_{c=1}^{d-1}\frac1{c!}
    \sum_{i_1,\dots,i_c=1}^m
    \partial_{i_1,\dots,i_c}\psi(x)\cdot
    y_{i_1}\cdots y_{i_c}
    ~+~
    \newsf{err}_d(x,y)\enspace,
\]
where
\[
    \abs{\newsf{err}_d(x,y)}
    ~\le~
    \frac1{d!}\cdot\norm{\psi^{(d)}}_1\cdot\norm{y}_\infty^d \enspace.
\]
\end{fact}

\subsection{Agnostic Learning and PAC Learning}

We recall the standard learning models.

\begin{definition}[Agnostic learning]
    Let \(\D\) be a distribution over \(\R^n\),
    and let \(\CC\) be a class of Boolean functions \(f:\R^n\to\{0,1\}\).
    We say that an algorithm $\A$ \emph{agnostically learns \(\CC\)
    under \(\D\)
    with accuracy \(\varepsilon\) and confidence \(1-\eta\)} if the following holds.
    For every distribution \(\D'\) over \(\R^n\times\{0,1\}\)
    whose marginal distribution on \(\R^n\) is \(\D\),
    given independent samples $(x,y)\sim\D'$,
    the algorithm outputs a hypothesis $h:\R^n\to\{0,1\}$
    with probability at least $1-\eta$,
    such that
    \[
    	\Pr_{(x,y)\sim\D'} [h(x)\neq y]
        ~\le~
        \inf_{f\in\CC} \Pr_{(x,y)\sim\D'} [f(x)\neq y] + \varepsilon
        \enspace.
    \]
\end{definition}

\begin{definition}[PAC learning]
    Let \(\D\) be a distribution over \(\R^n\),
    and let \(\CC\) be a class of Boolean functions \(f:\R^n\to\{0,1\}\).
    We say that an algorithm \(\A\) \emph{PAC learns \(\CC\)
    under \(\D\) with accuracy \(\varepsilon\) and confidence
    \(1-\eta\)} if the following holds.
    For every target function \(f\in\CC\),
    given independent samples \((x,f(x))\) with \(x\sim\D\),
    the algorithm outputs a hypothesis \(h:\R^n\to\{0,1\}\) such that,
    with probability at least \(1-\eta\),
    \[
        \Pr_{x\sim\D}[h(x)\neq f(x)]
        ~\le~
        \varepsilon
        \enspace.
    \]
\end{definition}

The following results connect noise sensitivity and Gaussian surface area to
the learnability of Boolean functions under the uniform distribution on $\pmcube{n}$
and the Gaussian distribution on $\R^n$, respectively.

 \begin{theorem}\label{thm:ns-gsa-learning}
    Let \(\CC\) be a class of Boolean functions.
    Suppose that there exist a function \(\alpha:[0,1]\to[0,1]\)
    and a parameter \(\beta\ge0\) such that for every \(f\in\CC\)
    and every \(\delta\in[0,1]\),
    \[
        \ns_{\delta}(f) \le \alpha(\delta)
        \qquad\text{and}\qquad
        \gsa(f) \le \beta \enspace.
    \]
    Then, for every \(\varepsilon,\eta\in(0,1)\), there exist algorithms
    that learn \(\CC\) with accuracy \(\varepsilon\) and confidence \(1-\eta\):
    \begin{itemize}[leftmargin=*]
        \item \citetext{\citealp[Theorem 5]{KKMS08}; \citealp[Corollary 17]{KOS02}}
        agnostically under the uniform distribution on \(\pmcube n\) in time
        $n^{O(1/\alpha^{-1}(\varepsilon^2/2.32))}
            \cdot \mathrm{poly}(\frac{1}{\varepsilon},\log(\frac{1}{\eta}))$
        
        \item \citetext{\citealp[Theorem 4.3]{Man94}; \citealp[Corollary 17]{KOS02}}
        in the PAC model under the uniform distribution on \(\pmcube n\) in time
        $n^{O(1/\alpha^{-1}(\varepsilon/4.64))}
            \cdot \mathrm{poly}(\frac{1}{\varepsilon},\log(\frac{1}{\eta}))$;

        \item \cite[Corollary~1.3]{PSW26}
        agnostically under the Gaussian distribution in time
        $n^{\tildeo(\beta^2/\varepsilon^2)}
            \cdot \mathrm{poly}(\frac{1}{\varepsilon},\log(\frac{1}{\eta}))$;        

        \item \cite[Theorems 10 and Theorem 15]{KOS08}
        in the PAC model under the Gaussian distribution in time
        $n^{O(\beta^2/\varepsilon^2)}
            \cdot \mathrm{poly}(\frac{1}{\varepsilon},\log(\frac{1}{\eta}))$.
            \end{itemize}
\end{theorem}

\section{Smooth Approximation of Threshold Functions}
\label{sec:smooth-approx}
This section develops a smooth approximation of threshold functions
used in the pseudorandom generator proof.
In \cref{sec:inner-outer-approximators}, we define the threshold sets
and introduce the inner and outer approximators for these sets.
\cref{sec:bentkus-type-mollifier} constructs Bentkus-type mollifiers
that serve as smooth approximators.
\cref{sec:sandwich} proves the sandwiching lemma,
which converts the problem of fooling the discontinuous threshold indicator
into fooling the smooth mollifiers.

\subsection{Inner and Outer Approximators}
\label{sec:inner-outer-approximators}

We first introduce notation for \emph{$k$-out-of-$m$ threshold sets},
which provides an equivalent representation for $k$-out-of-$m$ thresholds of halfspaces.

\begin{definition}\label{def:threshold-set}
    For an integer $k\geq 1$ and a vector $b\in\R^m$,
    we define $\ort_{k,b}$ to be the \emph{$k$-out-of-$m$ threshold set} associated with $b$:
    \[
        \ort_{k,b}
        :=
        \{x\in\R^m:\text{at least $k$ coordinates of $x$ satisfy }x_i\le b_i\}\enspace.
    \]
    With a slight abuse of notation, we also write $\ort_{k,b}$ for the indicator function of the set $\ort_{k,b}$, i.e., $\ort_{k,b}(x) = 1$ if $x \in \ort_{k,b}$ and $\ort_{k,b}(x) = 0$ otherwise.
    Equivalently, we can write
    \[
        \ort_{k,b}(x)
        =
        1-\prod_{S\in \binom{[m]}{k}}
        \left[
        1-\prod_{i\in S} \ind[x_i\le b_i]
        \right]\enspace.
    \]
\end{definition}

\begin{remark}\label{rmk:equivalent-representation}
    For a $k$-out-of-$m$ threshold of halfspaces $F(x) = \thr_{m,k}(f_1(x),\dots,f_m(x))$, where each $f_i(x) = \ind[\langle a_i,x\rangle \le b_i]$ is a halfspace, we can write $F(x) = \ort_{k,b}(A x)$, where $A$ is the matrix whose $i$-th row is $a_i$ and $b$ is the vector whose $i$-th coordinate is $b_i$.
\end{remark}

To relate smooth functions to the indicator $\ort_{k,b}$,
we allow the approximation to be inaccurate only within a narrow boundary.
We use the following definitions of \emph{inner} and \emph{outer approximators}
adapted from \cite{OST22}.

\begin{definition}\label{def:approximator}
For $k\geq 1$, $b\in\R^m$, $\Lambda>0$ and $\delta\in(0,1)$,
a function $A:\R^m\to[0,1]$ is a $(\Lambda,\delta)$-inner approximator
for the $k$-out-of-$m$ threshold set $\ort_{k,b}$, if
\[
\abs{A(x)-\ort_{k,b}(x)}\le\delta
\qquad
\forall \ x\notin \ort_{k,b}\setminus \ort_{k,b-\Lambda\cdot\onevec}\enspace.
\]
It is a $(\Lambda,\delta)$-outer approximator for $\ort_{k,b}$, if
\[
\abs{A(x)-\ort_{k,b}(x)}\le\delta
\qquad
\forall \ x\notin \ort_{k,b+\Lambda\cdot\onevec}\setminus \ort_{k,b}\enspace.
\]
\end{definition}

We list some elementary properties of these approximators
which follow directly from the definitions and will be used later.
\begin{fact}\label[fact]{fact:property-mollifier}
    Let $k\geq 1$ and $b\in\R^m$. Suppose that $\mollifier_{\newsf{in}}$ and $\mollifier_{\newsf{out}}$ are $(\Lambda,\delta)$-inner and $(\Lambda,\delta)$-outer approximators for the $k$-out-of-$m$ threshold set $\ort_{k,b}$, respectively. Then
    \begin{itemize}
        \item $\mollifier_{\newsf{in}}(x) - \delta \leq \ort_{k,b}(x) \leq \mollifier_{\newsf{out}}(x) + \delta$ for all $x \in \R^m$\,,
        \item $\mollifier_{\newsf{in}}$ is a $(\Lambda,\delta)$-outer approximator for $\ort_{k,b-\Lambda\cdot\onevec}$\,,
        \item $\mollifier_{\newsf{out}}$ is a $(\Lambda,\delta)$-inner approximator for $\ort_{k,b+\Lambda\cdot\onevec}$\,.
    \end{itemize}
\end{fact}

\subsection{Bentkus-Type Mollifier}
\label{sec:bentkus-type-mollifier}
We introduce a Bentkus-type mollifier for the indicator $\ort_{k,b}$
by replacing each halfline indicator in $\ort_{k,b}$ with a smoothed version.
The smoothed indicator for a halfline is defined as follows.

\begin{definition}[Gaussian-Mollified Halfline \citetext{\citealp[Definition 6.1]{OST22}}]
    For $\theta\in\R$ and $\lambda>0$, the $\lambda$-smoothed indicator
    for $\ind[x\le\theta]$ is defined as
    \[
        \widetilde{\ind}_{\theta,\lambda}(x)
        :=
        \Pr_{g\sim\mathcal N(0,1)}[x+\lambda g\le\theta]
        =
        \gcdf\!\left(\frac{\theta-x}{\lambda}\right)\enspace.
    \]  
\end{definition}

\begin{fact}[\citetext{\citealp[Exercise 11.41]{O14}}]\label[fact]{fact:derivative-smoothed-indicator}
    For every $\theta\in\R$, $\lambda>0$, integer $d\geq 1$ and $x\in\R$,
    \[
    	\abs{ \widetilde{\ind}_{\theta,\lambda}^{(d)}(x) }~
        = ~ O_d\!\left(\frac{1}{\lambda^d}\right) \enspace.
    \]
\end{fact}

The product representation of $\ort_{k,b}$ in \cref{def:threshold-set} expresses
the threshold indicator in terms of halfline indicators. We define its mollified
version by keeping the same product structure and replacing each indicator
$\ind[x_i\le b_i]$ with its $\lambda$-smoothed version
$\widetilde{\ind}_{b_i,\lambda}(x_i)$.

\begin{definition}\label{def:mollifier}
For an integer $k\geq 1$, a vector $b\in\R^m$, and $\lambda>0$,
the \emph{Bentkus-type $\lambda$-mollifier} for the $k$-out-of-$m$
threshold set $\ort_{k,b}$ is the function $\mollifier_{k,b,\lambda}:\R^m\to\R$
defined by
\[
    \mollifier_{k,b,\lambda}(x)
    :=
    1-\prod_{S\in \binom{[m]}{k}}
    \left[
    1-\prod_{i\in S}
    \widetilde{\ind}_{b_i,\lambda}(x_i)
    \right]\enspace.
\]
\end{definition}

When $k=m$, the threshold set $\ort_{k,b}$ is an orthant,
and $\mollifier_{k,b,\lambda}$ coincides with the Bentkus' mollifier for polytopes
\cite{Ben90,OST22}. 

The next theorem gives a derivative bound for this mollifier,
which will be used to control the Taylor remainder in the hybrid argument.
The proof is deferred to \cref{sec:derivative},
since it requires a separate technical analysis of the derivative structure
of the Bentkus-type mollifier.

\begin{restatable}{theorem}{derivatives}\label[theorem]{thm:scaled}
Let $m,k,d$ be integers with $m\ge 2$, $1\le k\le m/2$, and $d\ge 1$.
For every $b\in\R^m$ and every $\lambda>0$, the Bentkus-type $\lambda$-mollifier
$\mollifier_{k,b,\lambda}$ for $k$-out-of-$m$ threshold set $\ort_{k,b}$ satisfies
\[
    \norm{\mollifier_{k,b,\lambda}^{(d)}}_1
    =
    O_d\!\left(
    \frac{ k\sqrt{\log (m/k)} }{\lambda}
    \right)^d\enspace.
\]
\end{restatable}

With an appropriate choice of parameters, the Bentkus-type mollifier gives
both smooth inner and outer approximators for the threshold set.

\begin{restatable}{lemma}{inner}
\label[lemma]{lem:approximator}
    For every $1\leq k\leq m$, $b\in\R^m$, $\lambda>0$, and $\delta\in(0,1)$, let
    \[
    \beta:= \Theta\!\br{\lambda\sqrt{\log\frac{m}{\delta}}}\enspace,
    \qquad
    b^{\newsf{in}}:=b-\beta\cdot\onevec\enspace,
    \qquad
    b^{\newsf{out}}:=b+\beta\cdot\onevec\enspace,
    \qquad
    \Lambda:=2\beta\enspace.
    \]
    Then $\mollifier_{k,b^{\newsf{in}},\lambda}$ and
    $\mollifier_{k,b^{\newsf{out}},\lambda}$ are
    $(\Lambda,\delta)$-inner and $(\Lambda,\delta)$-outer approximators,
    respectively, for the $k$-out-of-$m$ threshold set $\ort_{k,b}$.
\end{restatable}

We use the following elementary estimate in the proof of \cref{lem:approximator}.
It says that the probability of seeing a bad $k$-subset is small
when every such subset contains at least one sufficiently unlikely element.
The proof is deferred to \appref{app:prob-estimate}.

\begin{restatable}{lemma}{simpleest}\label[lemma]{lem:prob-estimate}
    Let $1\leq k\leq m$, $\delta\in(0,1)$ and $p_i\in[0,1]$ for all $i\in[m]$.
    Suppose that there exists a set $T\subseteq[m]$ of size $m-k+1$
    such that for every $i\in T$, $p_i\le \frac{\delta}{4m^2}$.
    Then we have
\[
1-\prod_{S\in \binom{[m]}{k}}\left(1-\prod_{i\in S}p_i\right)
\le
\delta\enspace.
\]
\end{restatable}

\begin{proof}[Proof of \cref{lem:approximator}]
    We prove the inner approximator, and the outer approximator statement follows from a similar argument. Let $x\notin \ort_{k,b}\setminus \ort_{k,b-\Lambda\cdot\onevec}$. We consider two cases.
    \begin{itemize}
        \item $x\in \ort_{k,b-\Lambda\cdot\onevec}$. In this case,
        we know that there exists a set $S_0\in\binom{[m]}{k}$ such that $x_i\le b_i-\Lambda$ for every $i\in S_0$. Since
        $b_i^{\newsf{in}}=b_i-\beta$ and $\Lambda=2\beta$,
        we have that for every $i\in S_0$,
        \[
            \widetilde{\ind}_{b_i^{\newsf{in}},\lambda}(x_i)
            =
            \gcdf\!\left(\frac{b_i^{\newsf{in}}-x_i}{\lambda}\right)
            =
            \gcdf\!\left(\frac{b_i-\beta-x_i}{\lambda}\right)
            \ge
            \gcdf\!\left(\frac{\beta}{\lambda}\right)
            \ge
            1-e^{-\beta^2/(2\lambda^2)}\enspace,
        \]
        where the first inequality is from the fact that $b_i - \beta - x_i \geq \Lambda - \beta = \beta$ and $\gcdf$ is increasing, and 
        the second inequality is from the tail bound of the Gaussian distribution.
        Choosing $C$ large enough and letting $\beta = C\lambda\sqrt{\log\frac{m}{\delta}}$,
        we have
        $e^{-\beta^2/(2\lambda^2)}\le \frac{\delta}{m}$, and hence
        \[
            \prod_{i\in S_0}
            \widetilde{\ind}_{b_i^{\newsf{in}},\lambda}(x_i)
            \ge
            \left(1-\frac{\delta}{m}\right)^k
            \ge
            1 - \frac{k}{m}\cdot\delta
            \ge
            1-\delta\enspace.
        \]
        Thus, $1-\delta\le \mollifier_{k,b^{\newsf{in}},\lambda}(x)\le 1$.
        Since $\ort_{k,b}(x)=1$, we have that $\abs{\mollifier_{k,b^{\newsf{in}},\lambda}(x)-\ort_{k,b}(x)}\le\delta$.
        \item $x\notin \ort_{k,b}$. Then there are at least $m-k+1$ coordinates satisfying $x_i>b_i$.
        Let $T$ be the set of these coordinates. For every $i\in T$,
        \[
            \widetilde{\ind}_{b_i^{\newsf{in}},\lambda}(x_i)
            =
            \gcdf\!\left(\frac{b_i^{\newsf{in}}-x_i}{\lambda}\right)
            =
            \gcdf\!\left(\frac{b_i-\beta-x_i}{\lambda}\right)
            \le
            \gcdf\!\left(-\frac{\beta}{\lambda}\right)
            \le
            e^{-\beta^2/(2\lambda^2)}\enspace.
        \]
        Choosing $C$ large enough and letting $\beta = C\lambda\sqrt{\log\frac{m}{\delta}}$, this is at most $\frac{\delta}{4m^2}$. 
        By \cref{lem:prob-estimate}, we have that $\mollifier_{k,b^{\newsf{in}},\lambda}(x)\le\delta$.
        In this case, $\ort_{k,b}(x)=0$, so $\abs{\mollifier_{k,b^{\newsf{in}},\lambda}(x)-\ort_{k,b}(x)}\le\delta$. \qedhere
    \end{itemize}
\end{proof}

\subsection{A Sandwiching Argument}
\label{sec:sandwich}
We now prove the sandwiching lemma used in \cref{sec:prg}.
Intuitively, it says that if two random variables $x$ and $y$
have close expectations for the inner and outer approximators,
and if $x$ has an anticoncentration property that
$x$ is unlikely to be in the boundary of $\ort_{k,b}$,
then $x$ and $y$ also have close expectations for the threshold indicator $\ort_{k,b}$.

\begin{lemma}\label[lemma]{lem:sandwich}
    Let $1\leq k \leq m$, $b\in\R^m$, $\Lambda>0$ and $\delta,\varepsilon,\xi\in(0,1)$.
    Suppose that $\mollifier_{\newsf{in}}$ and $\mollifier_{\newsf{out}}$ are $(\Lambda,\delta)$-inner and $(\Lambda,\delta)$-outer approximators for
    the $k$-out-of-$m$ threshold set $\ort_{k,b}$, respectively.
    Let $x$ and $y$ be two random variables in $\R^m$ satisfying:
    \begin{itemize}
        \item For both $\mollifier\in\{\mollifier_{\newsf{in}}, \mollifier_{\newsf{out}}\}$,
        \[
        	\abs{ \expect{x}{\mollifier(x)} - \expect{y}{\mollifier(y)}  } ~\leq~ \varepsilon \enspace.
        \]
    \item For random variable $x$,
    \[
    	\Pr_x\!\Br{x \in \ort_{k,b+\Lambda\cdot\onevec}   \setminus \ort_{k,b - \Lambda\cdot\onevec}} ~\leq~ \xi \enspace.
    \]
    \end{itemize}
    Then, we have
    \[
    	\abs{ \expect{x}{\ort_{k,b}(x)} - \expect{y}{\ort_{k,b}(y)} } ~\leq~ \xi + \varepsilon + 2 \delta  \enspace.
    \]
\end{lemma}

\begin{proof}
    From \cref{fact:property-mollifier}, we have
    \begin{align*}
        \expect{y}{\ort_{k,b}(y)}
        ~\leq~ \expect{y}{ \mollifier_{\newsf{out}}( y ) } + \delta
        ~\leq~ \expect{x}{ \mollifier_{\newsf{out}}( x ) } + \varepsilon + \delta
        ~\leq~
        \expect{x}{\ort_{k,b + \Lambda\cdot\onevec}(x)} + \varepsilon + 2\delta \enspace.
    \end{align*}
    Moreover, we have
    \[
    	\expect{x}{\ort_{k,b + \Lambda\cdot\onevec}(x)}
        =
        \Pr_x\!\Br{x \in \ort_{k,b + \Lambda\cdot\onevec}}
        =
        \Pr_x\!\Br{x \in \ort_{k,b}} + \Pr_x\!\Br{x \in \ort_{k,b + \Lambda\cdot\onevec} \setminus \ort_{k,b} }
        \leq
        \expect{x}{\ort_{k,b}(x)} + \xi \enspace.
    \]
    Combining these inequalities, we obtain
    \[
    	\expect{y}{\ort_{k,b}(y)} ~\leq~ \expect{x}{\ort_{k,b}(x)} + \xi + \varepsilon + 2 \delta \enspace.
    \]
    By a similar argument, we have
    \[
    	\expect{y}{\ort_{k,b}(y)} ~\geq~ \expect{x}{\ort_{k,b}(x)} - \xi - \varepsilon - 2 \delta \enspace.
    \]
    This completes the proof.
\end{proof}

\section{Boolean Anticoncentration, Noise Sensitivity and Gaussian Surface Area}
\label{sec:anti}
This section proves a Boolean anticoncentration estimate
for $k$-out-of-$m$ thresholds of standardized halfspaces, as needed for the
sandwiching argument.
Our main technique is a \emph{random thinning argument},
which reduces the boundary of a $k$-out-of-$m$ threshold to
the boundary of an $\OR$ of a random subset of the halfspaces.
We begin with the Boolean anticoncentration theorem
for intersections (i.e., $\AND$) of halfspaces.

\begin{theorem}[\citetext{\citealp[Theorem 7.1]{OST22}}]\label{thm:ost-boolean-anticonc}
Let matrix $A\in\R^{m\times n}$ satisfy the condition that
each row of $A$ has a $\tau$-regular subvector of $2$-norm one.
For $b\in\R^m$, define the orthant $\ort_b^{\wedge}$ as
\[
    \ort_b^{\wedge} \coloneqq \{x\in\R^m:x_i\le b_i\ \text{ for all }\ i\in[m]\} \enspace.
\]
Then, for all $b\in \R^m$ and $\Lambda\ge\tau$,
\[
    \Pr_{u\sim \pmcube n}
    \left[
        Au \in \ort_{b+\Lambda\cdot\onevec}^{\wedge}
        \setminus
        \ort_{b-\Lambda\cdot\onevec}^{\wedge}
    \right]
    = O\!\br{\Lambda\sqrt{\log m}} \enspace.
\]
\end{theorem}

By the standard duality between $\AND$ and $\OR$, obtained by negating the halfspaces,
we also get the following anticoncentration bound for an $\OR$ of halfspaces.

\begin{corollary}\label{cor:ost-boolean-anticonc-or}
Let matrix $A\in\R^{m\times n}$ satisfy the condition that
each row of $A$ has a $\tau$-regular subvector of $2$-norm one.
For $b\in\R^m$, define
\[
    \ort_b^{\vee}
    \coloneqq
    \{x\in\R^m: x_i\le b_i\ \text{ for some } i\in[m]\}\enspace.
\]
Then, for all $b\in\R^m$ and $\Lambda\ge\tau$,
\[
    \Pr_{u\sim\pmcube n}
    \left[
        Au \in \ort_{b+\Lambda\cdot\onevec}^{\vee}
        \setminus
        \ort_{b-\Lambda\cdot\onevec}^{\vee}
    \right]
    =
    O\!\br{\Lambda\sqrt{\log m}}\enspace.
\]
\end{corollary}

\begin{proof} 
If $Au\in \ort_{b+\Lambda\cdot\onevec}^{\vee}\setminus\ort_{b-\Lambda\cdot\onevec}^{\vee}$,
then $(Au)_i>b_i-\Lambda$ for all $i$, and $(Au)_i\le b_i+\Lambda$ for some $i$.
Hence, $-Au\in\ort_{-b+2\Lambda\cdot\onevec}^{\wedge}\setminus\ort_{-b-2\Lambda\cdot\onevec}^{\wedge}$.
Therefore, by \cref{thm:ost-boolean-anticonc} with matrix $-A$, threshold $-b$, and width parameter $2\Lambda$,
the desired probability is
$O(2\Lambda\sqrt{\log m})=O(\Lambda\sqrt{\log m})$.
\end{proof}

We now prove an anticoncentration bound for $k$-out-of-$m$ thresholds.

\begin{theorem}\label{thm:threshold-ost-anticonc}
Let matrix $A\in\R^{m\times n}$ satisfy that
each row of $A$ has a $\tau$-regular subvector of $2$-norm one.
Then for every $1\leq k \leq m/2$, $b\in\R^m$ and $\Lambda\ge\tau$,
\[
    \Pr_{u\sim \pmcube n }
    \Br{
        Au\in \ort_{k,b+\Lambda\cdot\onevec} 
        \setminus
        \ort_{k,b-\Lambda\cdot\onevec} 
    }
    =
    O\!\br{ k \cdot \Lambda \sqrt{\log(m/k)}} \enspace,
\]
where $ \ort_{k,b} $ is the $k$-out-of-$m$ threshold set defined in \cref{def:threshold-set}.
\end{theorem}

\begin{proof}
For a fixed $u\in\pmcube n$, define
\[
    S_-(u) \coloneqq \{i:A_i u\le b_i-\Lambda\}
    \qquad\text{and}\qquad
    S_+(u) \coloneqq \{i:A_i u\le b_i+\Lambda\}
\]
to be the sets of satisfied halfspaces at the inner and outer thresholds, respectively.
Then,
\[
	\ind\!\Br{Au\in \ort_{k,b+\Lambda\cdot\onevec} \setminus \ort_{k,b-\Lambda\cdot\onevec}}
    =
    \ind\bigl[|S_+(u)|\ge k \ \wedge\ |S_-(u)|<k\bigr] \eqqcolon \EE(u) \enspace.
\]

Now, we randomly pick a subset $R\subseteq[m]$ by adding each row index $i\in [m]$ into $R$
independently with probability $p = 1/k$.
Define
\[
	\EE_R(u) \coloneqq \ind\bigl[R\cap S_+(u)\ne\emptyset ~\wedge~ R\cap S_-(u)=\emptyset\bigr]
\]
to indicate the event that $R$ contains at least one index from $S_+(u)$ but no index from $S_-(u)$.
Note that for a fixed $u\in\pmcube n$, if $\EE(u) = 1$,
then $|S_+(u)|\ge k$, $|S_-(u)|\le k-1$, and hence $|S_+(u)\setminus S_-(u)| \ge 1$.
In this case, we further have
\begin{align*}
    \Pr_R[\EE_R(u) = 1] = (1-p)^{|S_-(u)|} \cdot \left(1-(1-p)^{|S_+(u)\setminus S_-(u)|}\right)
    \ge p(1-p)^{k-1} \enspace.
\end{align*}
This means that for every $u\in\pmcube n$,
\[
	p(1-p)^{k-1}\cdot\EE(u) ~\le~ \Pr_R[\EE_R(u) = 1] \enspace.
\]
Taking expectation over $u\sim \pmcube n $ gives
\begin{equation}\label{eq:key-equation}
    p(1-p)^{k-1}\E_{u\sim \pmcube n}[\EE(u)]
    \le
    \E_{u\sim \pmcube n}\Br{\Pr_R[\EE_R(u) = 1]}
    =
    \Pr_{u\sim \pmcube n, R}\Br{\EE_R(u) = 1}\enspace.
\end{equation}

We next upper bound the right-hand side of \cref{eq:key-equation}.
To this end, we fix $R$ and upper bound $\Pr_u[\EE_R(u) = 1]$.
If $R$ is empty, then $\EE_R(u) = 0$ for all $u$, and thus $\Pr_u[\EE_R(u) = 1] = 0$.
Below, we assume $R$ is nonempty.
The key observation is that $\EE_R(u) = 1$ implies that
$u$ belongs to the $\pm\Lambda$-boundary of
the $\OR$ of the halfspaces $\{A_i u\le b_i\}_{i\in R}$.
Indeed, on this event,
there exists $i\in R$ such that $A_{i} u\le b_{i}+\Lambda$,
and there is no $i\in R$ such that $A_{i} u\le b_{i}-\Lambda$.
Let $A'\in\R^{|R|\times n}$ be the submatrix of $A$ with rows indexed by $R$,
and $b'\in\R^{|R|}$ be the corresponding subvector of $b$.
Then, $\EE_R(u) = 1$ implies
$A' u \in \ort_{b'+\Lambda\cdot\onevec}^{\vee}\setminus\ort_{b'-\Lambda\cdot\onevec}^{\vee}$.
Applying \cref{cor:ost-boolean-anticonc-or}, we get
\begin{align*}
    \Pr_u[\EE_R(u) = 1]
    &\le
    \Pr_u\Br{A' u \in \ort_{b'+\Lambda\cdot\onevec}^{\vee}\setminus\ort_{b'-\Lambda\cdot\onevec}^{\vee}}
    =
    O\!\br{\Lambda\sqrt{\log(|R|)}} \enspace.
\end{align*}

Combining the above bound with \cref{eq:key-equation}, we have
\begin{align*}
    p(1-p)^{k-1}\E_u[\EE(u)]
    &\le
    \E_R\Br{\Pr_u[\EE_R(u) = 1]}
    \le
    O\!\br{\Lambda \cdot \E_R\Br{\sqrt{\log(|R|)}}} 
    \le O\!\br{\Lambda \sqrt{\log(mp)}} \enspace,
\end{align*}
where the last inequality follows from Jensen's inequality and the fact that $\E_R[|R|] = mp$.
Since $p = 1/k$, we have $p(1-p)^{k-1} = \Omega(1/k)$, and thus
\[
    \Pr_{u\sim \pmcube n }
    \left[
        Au\in \ort_{k,b+\Lambda\cdot\onevec} 
        \setminus
        \ort_{k,b-\Lambda\cdot\onevec} 
    \right]
    =\E_u[\EE(u)]
    \le
    O\!\br{ k \cdot \Lambda \sqrt{\log(m/k)}}\enspace. \qedhere
\]
\end{proof}

The random thinning argument can also be used to
give bounds on the noise sensitivity and Gaussian surface area
of $k$-out-of-$m$ thresholds of halfspaces.
A similar argument yields the following theorem.

\begin{restatable}{theorem}{nsgsa}\label{thm:threshold-ns-gsa}
    Let \(m,k\) be integers with \(1 \le k \le m/2\),
    and let $F:\R^n \to \{0,1\}$ be a $k$-out-of-$m$ threshold of halfspaces.
    Then,
    \begin{itemize}
    \item \(\ns_\delta(F) = O( k \sqrt{\delta\cdot \log(m/k)} )\) for every
    \(\delta\in(0,1)\), and
    \item \(\gsa(F) = O(k \sqrt{\log(m/k)})\).
    \end{itemize}
\end{restatable}

We also provide lower bounds on the noise sensitivity
and Gaussian surface area of thresholds of halfspaces.
The proofs of both the upper and lower bounds
are given in \appref{app:threshold-ns-gsa}.

\begin{restatable}{theorem}{nslb}\label{prop:threshold-ns-lower-bound}
    There exists a universal constant \(c_0>0\) such that for all integers
    \(m,k\) with \(1 \le k \le m/2\) and
    \(0<\delta\le c_0/(k^2\log(m/k))\),
    there exists \(n_0=n_0(m,k,\delta)\)
    such that for every \(n\ge n_0\),
    some \(k\)-out-of-\(m\) threshold of halfspaces \(F:\R^n\to\{0,1\}\) satisfies
    \(\ns_\delta(F) = \Omega \bigl(\sqrt{\delta}(k+\sqrt{k\cdot\log(m/k)}) \bigr)\).
\end{restatable}

\begin{restatable}{theorem}{gsalb}\label{prop:threshold-surface-lower-bound}
    For all integers \(n,m,k\) with \(1 \le k \le m/2\) and \(m \le n\),
    there exists a \(k\)-out-of-\(m\) threshold of halfspaces
    \(F:\R^n\to\{0,1\}\) such that
    $\gsa(F)=\Omega (k + \sqrt{k \cdot \log (m/k)})$.
\end{restatable}

Furthermore, the following theorem shows that
in a small-\(\Lambda\) regime,
the Boolean anticoncentration probability
in \cref{thm:threshold-ost-anticonc} cannot in general
have only polylogarithmic dependence on \(k\).
Thus, obtaining a \(\prg\) with seed length polylogarithmic in \(k\)
would require new ideas to overcome this obstruction.
We defer the proof to \appref{app:boolean-threshold-band-lower}.

\begin{restatable}{theorem}{booleanlb}\label{thm:boolean-threshold-band-lower}
There is a universal constant \(c_0>0\) with the following property.
Let \(1\le k\le m/2\) and
\(0<\Lambda\le c_0/(k\sqrt{\log(m/k)})\).
Then there exists an integer \(n_0=n_0(m,k,\Lambda)\) such that,
for every \(n\ge n_0\),
there exist a matrix \(A\in\R^{m\times n}\) and a threshold \(b\in\R^m\)
such that every row of \(A\) has $2$-norm \(1\) and is \(\tau\)-regular
for some \(\tau\le\Lambda\), and
\[
    \Pr_{u\sim\pmcube n}
    \left[
        Au\in \ort_{k,b+\Lambda\cdot\onevec}
        \setminus
        \ort_{k,b-\Lambda\cdot\onevec}
    \right]
    =
    \Omega\!\left(\Lambda\!\left(k+\sqrt{k\cdot\log(m/k)}\right)\right).
\]
\end{restatable}

Combining \cref{thm:threshold-ns-gsa} with the \cref{thm:ns-gsa-learning},
we immediately obtain the following results on the learnability of $k$-out-of-$m$ thresholds of halfspaces.

\begin{corollary}\label{cor:threshold-learning}
    Let \(n,m,k\in\N\) satisfy \(1\le k\le m/2\).
    For every \(\varepsilon,\eta\in(0,1)\),
    the class of $k$-out-of-$m$ thresholds of halfspaces over \(\R^n\)
    is learnable with accuracy \(\varepsilon\) and confidence \(1-\eta\)
    as follows:
    \begin{itemize}[leftmargin=*]
        \item
        agnostically under the uniform distribution on \(\pmcube n\) 
        in time
        $n^{O( k^2\log(m/k)/\varepsilon^4 )}
            \cdot \mathrm{poly}(\frac{1}{\varepsilon},\log(\frac{1}{\eta}))$;

        \item 
        agnostically under the Gaussian distribution in time
        $n^{\tildeo(k^2\log(m/k)/\varepsilon^2)}
            \cdot \mathrm{poly}(\frac{1}{\varepsilon},\log(\frac{1}{\eta}))$;

        \item 
        in the PAC model, under either the uniform distribution on \(\pmcube n\) or
        the Gaussian distribution in time
        $n^{O(k^2\log(m/k)/\varepsilon^2)}
            \cdot \mathrm{poly}(\frac{1}{\varepsilon},\log(\frac{1}{\eta}))$.
    \end{itemize}
\end{corollary}

\section{A Pseudorandom Generator for Thresholds of Halfspaces}
\label{sec:prg}

This section proves our main pseudorandom generator result
for thresholds of halfspaces.
The core ingredient is a careful analysis showing that
the generator from \cite{OST22} fools the mollifier
from \cref{def:mollifier} for the $k$-out-of-$m$ threshold
of \emph{standardized} halfspaces.
Together with the anticoncentration bound from
\cref{sec:anti} and the sandwiching argument from
\cref{sec:sandwich}, this implies fooling of the actual $k$-out-of-$m$
threshold of standardized halfspaces.
The result for general halfspaces then follows from a reduction to standardized halfspaces.

In \cref{sec:prg-construction}, we describe the generator,
list the parameter choices, and give the resulting seed length.
\cref{sec:standardization} states the standardization reduction,
which allows us to restrict attention to thresholds of standardized halfspaces.
In \cref{sec:threshold-mollifier-fool}, we prove via a hybrid argument
that the generator fools the mollifier
for the $k$-out-of-$m$ threshold of standardized halfspaces.
\cref{sec:cnf-derivative} proves a technical lemma used in the hybrid argument.
Finally, in \cref{sec:threshold-prg}, we put these ingredients together
to prove the main theorem.

\subsection{Construction of the Generator}
\label{sec:prg-construction}

We adopt the O'Donnell-Servedio-Tan $\prg$ construction
in \cite[Definition 4.2]{OST22},
which is a combination of Meka–Zuckerman generator for halfspaces
\cite{MZ13} and
the Gopalan-Meka-Reingold generator for small-width \CNF{} formulas
\cite{gopalan2013dnf}.

\begin{figure}[!t]
\centering
\setlength{\fboxsep}{0.8em}
\fbox{\begin{minipage}{0.92\textwidth}
    \centering
    \vspace{0.1em}
    \(
        m ~\text{: number of halfspaces}~
        \quad k ~\text{: threshold parameter}~
        \quad \delta\in(0,1) ~\text{: error of the \prg}
    \)
\end{minipage}
}

\vspace{0.6em}

\fbox{\begin{minipage}{0.92\textwidth}
    \vspace{0.4em}
\[ 
\begin{array}{rcl@{\hspace{2.4em}}l}
    \varepsilon &\in& (0,1) 
    &\text{arbitrarily small constant}\,,\\[1.4em]

    \dhyb &\in& \N 
    &\text{some constant depends on $\varepsilon$}\,,\\[1.4em]

    \lambda &=& 
    \dfrac{\delta}{k\sqrt{\log(m/k) \log (m/\delta)}}
    &\text{dictated by \cref{eq:lambda}}\,,\\[1.4em]

    \tau &=&
    \dfrac{\delta^{1+\varepsilon}}{k^{2+\varepsilon}\cdot \log^{2.5+\varepsilon}\!m}
    &\text{\text{dictated by \cref{eq:error}}\,}\,,\\[1.4em]

    L &=&  
        \dfrac{k^{4+2\varepsilon} \cdot  \log^{4+2\varepsilon}\!m}{\delta^{2 + 2\varepsilon}}
    &\text{chosen according to \cref{eq:error}}\,,\\[1.4em]

    r_{\newsf{hash}} &=&
    C_{\newsf{hash}}\cdot \log(mL/\delta)
    &\text{dictated by \cref{prop:ost-head-good-hash}}\,,\\[1.4em]

    r_{\newsf{bucket}} &=&
    \log(m/\delta)
    &\text{dictated by \cref{lem:ost-tail-bound}}\,,\\[1.4em]

    \Khead &=&
    C_{\newsf{std}} \cdot
    \dfrac{\log (m/\delta)\log \log (m/\delta)}{\tau^2}
    &\text{dictated by \cref{lem:rowwise-standardization}}\,,\\[1.4em]

    \whd &=&
    \dfrac{2\Khead}{L}
    &\text{dictated by \cref{prop:ost-head-good-hash}}\,,\\[1.4em]

    \Wcnf &=&
    k\cdot\whd
    &\text{dictated by \cref{lem:one-bucket-replacement}}\,,\\[1.4em]

    \deltaCNF &=&
    \dfrac{\delta}{L}\cdot \br{\dfrac{\lambda}{m^{k+1}\sqrt{n}}}^{\dhyb-1}
    &\text{dictated by \cref{eq:error}}\,.\\[1.4em]
\end{array}
\]
\end{minipage}
}
\caption{Parameter choices}
\label{fig:threshold-parameters}
\end{figure}

\begin{definition}\label{def:threshold-ost-generator}
Fix parameters
$L,r_{\newsf{hash}},r_{\newsf{bucket}},\Khead,\Wcnf\in\N$
and $\deltaCNF\in(0,1)$ as in \cref{fig:threshold-parameters}.
The generator
$G=G(L,r_{\newsf{hash}},r_{\newsf{bucket}},\Khead,\Wcnf,\deltaCNF)$
is defined as follows:
\begin{enumerate}
    \item Draw a hash function \(h\) uniformly at random from an
    \(r_{\newsf{hash}}\)-wise uniform hash family.
    \item Draw independent $r_{\newsf{bucket}}$-wise uniform strings
    $y^1,\ldots,y^L\in\pmcube n$.
    \item Draw independent strings
    $\widetilde y^1,\ldots,\widetilde y^L\in\pmcube n$
    from $\gmr(\Wcnf,\deltaCNF)$ in \cref{thm:ost-cnf-prg}.
    \item Draw an independent $2\Khead$-wise uniform string
    $y^\star\in\pmcube n$.
\end{enumerate}
Define the random variables $\widehat{y}^1,\ldots,\widehat{y}^L$ by
$\widehat{y}^\ell \coloneqq y^\ell \oplus \widetilde y^\ell$ for every $\ell\in[L]$,
where $\oplus$ is the coordinate-wise $\xor$ operation on $\pmcube{}$.
Define the random variable $\breve y \in \pmcube n$ by
\[
	\breve{y}_{h^{-1}(\ell)} \coloneqq \widehat{y}^\ell_{h^{-1}(\ell)}\enspace, \quad \forall\ \ell\in[L] \enspace.
\]
In other words, on the \(\ell\)-th bucket \(h^{-1}(\ell)\), \(\breve y\) is the
coordinate-wise $\xor$ of \(y^\ell\) and \(\widetilde y^\ell\).
The generator outputs a string $z\in\pmcube n$ where
$z \coloneqq \breve y \oplus y^\star$.
\end{definition}

The seed length of the generator is given by the following lemma, which is a direct consequence of the definition and the seed lengths of each component.

\begin{lemma}[\citetext{\citealp[Fact 4.3]{OST22}}]\label{lem:ost-generator-seed}
The generator in \cref{def:threshold-ost-generator} has seed length
\[
	\begin{aligned}
        &
        O\bigl(r_{\newsf{hash}}\cdot\log (nL) + L\cdot r_{\newsf{bucket}}\cdot \log n\bigr)
        ~ + ~ O(\Khead\cdot\log n) \\
        &\quad + ~
        L\cdot
        O\Bigl(
            \Wcnf^2 \cdot \log^2(\Wcnf\cdot\log(1/\deltaCNF))
            +\Wcnf \cdot \log \Wcnf \cdot \log(1/\deltaCNF)
            +\log\log n
        \Bigr) \enspace.
    \end{aligned}
\]
In particular, with the parameter choices in \cref{fig:threshold-parameters}, the seed length is
\[
	\tildeo\Bigl(
        \frac{k^{6+2\varepsilon} \cdot \log^{6+2\varepsilon}\!m}{\delta^{2 + 2\varepsilon}}\cdot \log n
    \Bigr) \enspace.
\]

\end{lemma}

\subsection{Reduction to Thresholds of Standardized Halfspaces}
\label{sec:standardization}

We show that thresholds of \emph{general} halfspaces can be reduced to
thresholds of \emph{standardized} halfspaces,
which have a \emph{sparse} head and a \emph{regular} tail.
This reduction is the same as in \cite[Section 5]{OST22},
but we restate it here for completeness.

\begin{lemma}[\citetext{\citealp[Lemma 5.1]{OST22}}]\label{lem:rowwise-standardization}
\footnote{Lemma~5.1 of \cite{OST22} is stated for intersections of halfspaces. The formulation used here follows the same proof: each halfspace is standardized separately, and the final error bound is obtained by a union bound over the \(m\) halfspaces.}
There is a universal constant $C_\newsf{std}>0$ such that the following holds:
fix $m\geq 1$, $0<\eta,\tau<1/2$ such that
\[
    K
    \coloneq
    C_\newsf{std}\cdot\frac{\log(m/\eta)\log\log(m/\eta)}{\tau^2} \leq \frac{n}{2} \enspace.
\]
For every matrix $A\in\R^{m\times n}$ and vector $b\in\R^m$, there is a
$(K,\tau)$-standardized matrix $A'\in\R^{m\times n}$ and a vector
$b'\in\R^m$ such that for every $2K$-wise uniform
$y\in\pmcube n$,
\[
    \Pr_y
    \Br{\exists\ i \in[m] \text{ such that }
        \ind[\ip{ A_i, y }\le b_i]\ne \ind[\ip{A'_i, y}\le b'_i]}
    \le \eta \enspace.
\]
\end{lemma}

As a consequence, we can reduce fooling arbitrary thresholds of halfspaces to
fooling thresholds of standardized halfspaces.

\begin{corollary}\label{cor:threshold-outer-reduction}
Let \(0<\eta,\tau<1/2\) and \(K\)
be as in \cref{lem:rowwise-standardization}.
Let \(f:\{0,1\}^m\to\{0,1\}\) be any Boolean function.
For every \(A\in\R^{m\times n}\) and \(b\in\R^m\), let \(A',b'\) be the
\((K,\tau)\)-standardized instance given by \cref{lem:rowwise-standardization}.
Then for every $2K$-wise uniform
$y\in\pmcube n$,
\[
    \Pr_y\Bigl[
        f(\ind[\ip{A_1,y}\le b_1],\ldots,\ind[\ip{A_m,y}\le b_m])
        \ne
        f(\ind[\ip{A'_1,y}\le b'_1],\ldots,\ind[\ip{A'_m,y}\le b'_m])
    \Bigr]
    \le \eta \enspace.
\]
In particular, for $f=\thr_{m,k}$, i.e., the $k$-out-of-$m$ threshold function, we have
\[
	\Pr_y\Bigl[
        \ort_{k,b}(Ay)
        \ne
        \ort_{k,b'}(A'y)
    \Bigr]
    \le \eta \enspace,
\]
where $\ort_{k,b}$ and $\ort_{k,b'}$ are indicators of the $k$-out-of-$m$ threshold sets
defined in \cref{def:threshold-set}.
\end{corollary}

\begin{proof}
The function $f$ can disagree only if at least one of the $m$ row bits disagrees.
Hence the disagreement event is contained in the rowwise disagreement event
from \cref{lem:rowwise-standardization}, whose probability is at most \(\eta\).
\end{proof}

The final \(2\Khead\)-wise $\xor$ in \cref{def:threshold-ost-generator} makes
the generator output \(2\Khead\)-wise uniform.  Hence
\cref{cor:threshold-outer-reduction} applies to the pseudorandom distribution
as well as to the uniform distribution, and it remains to fool thresholds of
standardized halfspaces.
We now describe the structural consequence of
standardization that will be used in the bucket-by-bucket hybrid argument.

\vspace{-0.6em}
\paragraph{Sparse heads and regular tails.}
Let $A\in\R^{m\times n}$ be a $(K,\tau)$-standardized matrix.
By definition, we can decompose $A$ as
\begin{equation} \label{eq:head-tail-decomposition}
    A = H + T \enspace,
\end{equation}
where \(H\in\R^{m\times n}\) is the \emph{head} matrix and \(T\in\R^{m\times n}\) is the \emph{tail} matrix.
In particular, each row of $H$ is supported on at most $K$ coordinates,
and each row of $T$ is $\tau$-regular with norm one.
This decomposition is useful because a random hash spreads
the heavy coordinates across buckets, so each bucket sees only a small part of the head.

To be more precise, for a bucket $B\subseteq[n]$,
we say that $B$ is \(w\)\emph{-sparse} with respect to \(H\)
if
\begin{equation} \label{eq:def-sparse-bucket}
     |\operatorname{supp}(H_i)\cap B|\le w
    \qquad\text{for every }i\in[m] \enspace.
\end{equation}
The following proposition says that this sparsity condition holds
simultaneously for all rows and all buckets with high probability.

\begin{proposition}[\citetext{\citealp[Proposition 8.11]{OST22}}]\label{prop:ost-head-good-hash}
Let \(H\in\R^{m\times n}\) be a matrix such that each row is supported on
at most \(K\) coordinates.
Let \(h:[n]\to[L]\) be sampled uniformly from an \(r\)-wise uniform hash family,
where \(r\ge C_{\newsf{hash}}\log(mL/\delta)\) for some constant \(C_{\newsf{hash}}\), and let \(w = 2K/L\).
Then, except with probability at most $\delta$ over $h$,
for every \(\ell\in[L]\), the bucket \(h^{-1}(\ell)\) is \(w\)-sparse with respect to \(H\).
\end{proposition}

\subsection{Fooling the Bentkus-Type Mollifier for Thresholds of Standardized Halfspaces}
\label{sec:threshold-mollifier-fool}
The main ingredient for the generator analysis is the following theorem,
which shows that the generator from \cref{def:threshold-ost-generator} fools
the Bentkus-type mollifier for thresholds of standardized halfspaces.

\begin{theorem}\label{thm:threshold-mollifier-fool}
    Let all parameters be chosen as in \cref{fig:threshold-parameters},
    and let $G$ be the generator from \cref{def:threshold-ost-generator}.
    Then for every $(\Khead,\tau)$-standardized matrix $A\in\R^{m\times n}$ and
    all $b\in\R^m$,
    \[
    	\abs{
            \E_{u\sim\pmcube n} \Br{\mollifier_{k,b,\lambda}(Au)}
            -
            \E_{z\sim G} \Br{\mollifier_{k,b,\lambda}(Az)}  
        }
        ~=~
        O(\delta)\enspace,
    \]
    where $\mollifier_{k,b,\lambda}$ is the Bentkus-type mollifier defined in \cref{def:mollifier} for the $k$-out-of-$m$ threshold set $\ort_{k,b}$.
\end{theorem}

Recall from \cref{def:threshold-ost-generator} that the random variable
$z$ is the coordinate-wise $\xor$ of $\breve y$ and $y^\star$.
We can also view a uniform $u$ as the coordinate-wise $\xor$ of $u$ and $y^\star$.
Since $y^\star$ can be absorbed in the coefficient matrix $A$,
it suffices to show that $\breve y$ fools the mollifier, i.e., that
\[
    \abs{
        \E_{u\sim\pmcube n} \Br{\mollifier_{k,b,\lambda}(Au)}
        -
        \E_{\breve y} \Br{\mollifier_{k,b,\lambda}(A\breve y)}  
    }
    ~=~
    O(\delta)\enspace.
\]
We prove this by a standard hybrid argument,
replacing the random variables one bucket at a time.
For each hash function $h$ and each bucket index $\ell=0,1,\ldots,L$, define
the hybrid random variable $x^{h,\ell}$ by
\begin{equation} \label{eq:hybrid-x}
    x^{h, \ell}_j
    \coloneqq
    \begin{cases}
        ~\breve y_j & \text{if } h(j) \leq \ell\,,\\[0.1em]
        ~u_j & \text{otherwise}\,.
    \end{cases}
\end{equation}
Then $x^{h,0}=u$, $x^{h,L}=\breve y$, and for each $\ell\in[L]$,
$x^{h, \ell}$ is obtained from $x^{h, \ell-1}$ by replacing the coordinates
in the $\ell$-th bucket with the corresponding coordinates of $\breve y$.

Note that the $(\Khead,\tau)$-standardized $A$ can be written as
$A = H + T$ as in \cref{eq:head-tail-decomposition}.
As long as the hash function $h$ is such that every bucket is $w$-sparse with respect to $H$,
which happens with high probability by \cref{prop:ost-head-good-hash}, 
we can show that
\begin{equation}\label{eq:single-hybrid}
    \E_{u\sim\pmcube n} \Br{\mollifier_{k,b,\lambda}(Ax^{h, \ell-1})} \approx \E_{u\sim\pmcube n} \Br{\mollifier_{k,b,\lambda}(Ax^{h, \ell})} \enspace.
\end{equation}
This is done by a Taylor expansion of the mollifier.
To be more precise, let $s \in \pmcube{[n]\setminus h^{-1}(\ell)}$
be the variables that $x^{h, \ell}$ and $x^{h, \ell-1}$ have in common.
We may write
\begin{align}\label{eq:single-hybrid-explain}
    \mollifier_{k,b,\lambda}(Ax^{h, \ell-1}) 
     = \mollifier_{k,b,\lambda}(A^{[n]\setminus h^{-1}(\ell)}\cdot s + A^{h^{-1}(\ell)}\cdot u_{h^{-1}(\ell)}) 
     = \mollifier_{k,b-A^{[n]\setminus h^{-1}(\ell)}\cdot s,\lambda}(A^{h^{-1}(\ell)}\cdot u_{h^{-1}(\ell)})
\end{align}
where the second equality follows from the definition of the mollifier.
Using the decomposition of $A$, we have
$ A^{h^{-1}(\ell)}\cdot u_{h^{-1}(\ell)} = H^{h^{-1}(\ell)}\cdot u_{h^{-1}(\ell)} + T^{h^{-1}(\ell)}\cdot u_{h^{-1}(\ell)}$.
Letting $b' = b-A^{[n]\setminus h^{-1}(\ell)}\cdot s$,
expanding the last term in \cref{eq:single-hybrid-explain} around the head $H^{h^{-1}(\ell)}\cdot u_{h^{-1}(\ell)}$ according to \cref{fact:multidim-taylor} gives
\begin{align}\label{eq:taylor-expansion-mollifier}
    &\mollifier_{k,b', \lambda}(A^{h^{-1}(\ell)} u_{h^{-1}(\ell)}) \nonumber \\
    &\qquad =  \sum_{c = 0}^{d - 1} \frac{1}{c!} \cdot \sum_{i_1, \ldots, i_c=1}^m
    \partial_{i_1, \ldots, i_c} \mollifier_{k,b',\lambda}(H^{h^{-1}(\ell)} u_{h^{-1}(\ell)}) \cdot (T^{h^{-1}(\ell)} u_{h^{-1}(\ell)})_{i_1} \cdots (T^{h^{-1}(\ell)} u_{h^{-1}(\ell)})_{i_c}  \nonumber \\
    &\qquad\qquad + \newsf{err}_d(H^{h^{-1}(\ell)} u_{h^{-1}(\ell)}, T^{h^{-1}(\ell)} u_{h^{-1}(\ell)})\enspace,
\end{align}
where
\begin{align*}
    \abs{\newsf{err}_d(H^{h^{-1}(\ell)} u_{h^{-1}(\ell)}, T^{h^{-1}(\ell)} u_{h^{-1}(\ell)}) }
    \leq \frac{1}{d!} \cdot \norm{ \mollifier_{k,b',\lambda}^{(d)} }_1 \cdot \norm{T^{h^{-1}(\ell)} u_{h^{-1}(\ell)}}_\infty^d \enspace.
\end{align*}
Therefore, to show \cref{eq:single-hybrid},
it suffices to show that the generator fools each term in the above expansion,
and that the error term is small.

To fool each term, 
we prove the following lemma in \cref{sec:cnf-derivative},
which shows that a random variable that fools width-$kw$ \CNF{} formulas
also fools the terms in the Taylor expansion,
as long as the head is $w$-sparse.

\begin{restatable}{lemma}{bucketreplacement}\label{lem:one-bucket-replacement}
    Let $B \subseteq [n]$ and $H^B, T^B \in \R^{m \times |B|}$
    such that each row of $H^B$ is supported on at most $w$ coordinates
    and each row of $T^B$ has norm at most $1$.
    Let $u$ be a uniform random variable over $\pmcube{B}$ and $y$ be a random variable over $\pmcube{B}$ that fools all width-$kw$ \CNF{} formulas with error at most $\deltaCNF$.
    Then for every $b\in\R^m$, $\lambda>0$, $c\in\N$ and $i_1, \ldots, i_c \in [m]$, we have
    \begin{align*}
        &\abs{
            \E_{u} \Br{\partial_{i_1, \ldots, i_c} \mollifier_{k,b,\lambda}(H^B u) \cdot (T^B u)_{i_1} \cdots (T^B u)_{i_c}}
            -
            \E_{y} \Br{\partial_{i_1, \ldots, i_c} \mollifier_{k,b,\lambda}(H^B y) \cdot (T^B y)_{i_1} \cdots (T^B y)_{i_c}}
        } \\
        & \qquad\qquad\qquad\qquad\qquad\qquad\qquad\qquad\qquad\qquad\qquad\qquad\qquad\qquad\qquad =~ \deltaCNF \cdot O_c\!\br{ \frac{m^k \sqrt{n}}{\lambda} }^c \enspace,
    \end{align*}
    where $\mollifier_{k,b,\lambda}$ is the Bentkus-type mollifier defined in \cref{def:mollifier} for the $k$-out-of-$m$ threshold set $\ort_{k,b}$.
\end{restatable}

To bound the error term, we use the bounds on the derivatives of the mollifier in \cref{thm:scaled}
and the following moment bound on the tail from \cite{OST22}.

\begin{lemma}[\citetext{\citealp[Lemma 8.10]{OST22}}]\label{lem:ost-tail-bound}
    Let \(h:[n]\to[L]\) be sampled uniformly from an
    \(r_{\newsf{hash}}\)-wise uniform hash family,
    and $y$ be an $r_{\newsf{bucket}}$-wise uniform random variable over $\pmcube n$,
    where $r_{\newsf{hash}}$ and $r_{\newsf{bucket}}$ are as in \cref{fig:threshold-parameters}.
    Let $T\in\R^{m \times n}$ be a matrix whose rows are $\tau$-regular with norm $1$.
    Then, for all $\ell\in[L]$ and integer $d\ge 1$,
    \[
    	\E_{h, y} \Br{\norm{T^{h^{-1}(\ell)} y_{h^{-1}(\ell)}}_\infty^d}
        ~=~
        O_d\!\br{ \tau \log (m/\delta) + \sqrt{ \frac{\log (m/\delta) }{L} }}^d\enspace.
    \]
\end{lemma}

Now we have all the ingredients to complete the proof of \cref{thm:threshold-mollifier-fool}.
\begin{proof}[Proof of \cref{thm:threshold-mollifier-fool}]
    As explained above, it suffices to show that
    $\breve y$ defined in \cref{def:threshold-ost-generator} fools the mollifier alone.
    Using the hybrid random variables $\{ x^{h, \ell} \}_{\ell=0}^L$ defined in \cref{eq:hybrid-x},
    \begin{align*}
        (\triangle)~&\coloneqq~ \abs{
            \E_{u\sim\pmcube n} \Br{\mollifier_{k,b,\lambda}(Au)}
            -
            \E_{\breve y} \Br{\mollifier_{k,b,\lambda}(A\breve y)}  
        }\\
        &~=~ \abs{
            \E_{h} \Br{\E_{u,\breve y} \Br{\mollifier_{k,b,\lambda}(Ax^{h,0})}
            -
            \E_{u,\breve y} \Br{\mollifier_{k,b,\lambda}(Ax^{h,L})}}
        } \\
        &~\leq~ \E_{h} \Br{\abs{
            \E_{u,\breve y} \Br{\mollifier_{k,b,\lambda}(Ax^{h, 0})}
            -
            \E_{u,\breve y} \Br{\mollifier_{k,b,\lambda}(Ax^{h, L})}
        }} \enspace.
    \end{align*}
    Recall that the $(\Khead,\tau)$-standardized $A$ can be written as
    $A = H + T$ as in \cref{eq:head-tail-decomposition}.
    By \cref{prop:ost-head-good-hash}, except with probability at most $\delta$ over the choice of $h$, every bucket $h^{-1}(\ell)$ is $\whd$-sparse with respect to $H$.
    Let $\mathcal{E}$ denote the event that all buckets are $\whd$-sparse,
    and $\ind_{\mathcal{E}}$ be the indicator random variable for $\mathcal{E}$.
    Then, we have
    \begin{align}\label{eq:intro-hybrid}
        (\triangle)
        \leq ~&~ \E_{h} \Br{\abs{
            \E_{u,\breve y} \Br{\mollifier_{k,b,\lambda}(Ax^{h, 0})}
            -
            \E_{u,\breve y} \Br{\mollifier_{k,b,\lambda}(Ax^{h, L})}
        } \cdot \ind_{\mathcal{E}}} ~+~ \delta \enspace.
    \end{align}
    Now, for every fixed $h$ such that $\mathcal{E}$ holds,
    we claim that for all integers $d\ge1$,
    \begin{align}\label{eq:hybrid-goal}
        &\abs{
            \E_{u,\breve y} \Br{\mollifier_{k,b,\lambda}(Ax^{h, 0})}
            -
            \E_{u,\breve y} \Br{\mollifier_{k,b,\lambda}(Ax^{h, L})}
        }
        ~\leq~
        L \cdot \deltaCNF \cdot O\!\br{ \frac{m^{k+1} \sqrt{n}}{\lambda} }^{d-1} \nonumber\\
        &\quad+~
        O\!\br{ \frac{k\sqrt{\log(m/k)}}{\lambda} }^d \cdot
        \sum_{\ell = 1}^L 
        \left(
            \E_{u}\Br{\norm{T^{h^{-1}(\ell)} u_{h^{-1}(\ell)}}_\infty^d}
            +
            \E_{\widehat{y}^\ell}\Br{\norm{T^{h^{-1}(\ell)} \widehat{y}^\ell_{h^{-1}(\ell)}}_\infty^d}
        \right) \enspace,
    \end{align}
    where $\widehat{y}^\ell$ is defined in \cref{def:threshold-ost-generator}.
    Combining \cref{eq:intro-hybrid} and \cref{eq:hybrid-goal}, we have
    \begin{align} \label{eq:error}
        (\triangle)
        \leq ~& L \cdot \deltaCNF \cdot O\!\br{ \frac{m^{k+1} \sqrt{n}}{\lambda} }^{d-1} \nonumber\\
        &~+
        O\!\br{ \frac{k\sqrt{\log(m/k)}}{\lambda} }^d \cdot
        \sum_{\ell = 1}^L 
        \left(
            \E_{h,u}\Br{\norm{T^{h^{-1}(\ell)} u_{h^{-1}(\ell)}}_\infty^d}
            +
            \E_{h,\widehat{y}^\ell}\Br{\norm{T^{h^{-1}(\ell)} \widehat{y}^\ell_{h^{-1}(\ell)}}_\infty^d}
        \right)  + \delta \nonumber\\
        = L & \cdot \deltaCNF \cdot O\!\br{ \frac{m^{k+1} \sqrt{n}}{\lambda} }^{d-1}\!\!\!\!\!+ O\!\br{ \frac{k\sqrt{\log(m/k)}}{\lambda} }^d\!\! \cdot
        L\cdot O\!\br{ \tau \log (m/\delta) + \sqrt{ \frac{\log (m/\delta) }{L} }}^d\!\! + \delta \nonumber\\
        = L & \cdot \deltaCNF \cdot O\!\br{ \frac{m^{k+1} \sqrt{n}}{\lambda} }^{d-1}\!\!\!\!\!+ L \cdot O\!\br{
            \frac{\tau k\sqrt{\log(m/k)}\log (m/\delta)}{\lambda} + \frac{k\sqrt{\log(m/k)\log (m/\delta) }}{\lambda \sqrt{L}}
        }^d \,,
    \end{align}
    where the first equality follows from \cref{lem:ost-tail-bound}, since both $u$ and $\widehat{y}^\ell$ are $r_{\newsf{bucket}}$-wise uniform, and each row of $T$ has norm $1$.
    Taking $d = \dhyb$ and using the parameters in \cref{fig:threshold-parameters}, we have
    \begin{align*}
        (\triangle) &= O(\delta) + \frac{k^5\cdot \log^5\! m}{\delta^{2+\varepsilon}}\cdot
        O\!\br{
            \delta^{\varepsilon}\cdot \frac{ \log (m/k) \cdot \log^{1.5}\!(m/\delta) }{k^{\varepsilon}\cdot \log^{2.5+\varepsilon}\! m}
            +
            \delta^{\varepsilon}\cdot \frac{ \log (m/k) \cdot \log (m/\delta) }{k^{\varepsilon}\cdot \log^{2+\varepsilon}\! m}
        }^{\dhyb}\,.
    \end{align*}
    Choosing $\dhyb$ such that $\dhyb\cdot\varepsilon$ is sufficiently large, we have $(\triangle) = O(\delta)$ as desired.

    Below, we prove \cref{eq:hybrid-goal} for every fixed $h$
    such that for every $\ell\in[L]$, the bucket $h^{-1}(\ell)$ is $\whd$-sparse with respect to $H$.
    Fix such an $h$.
    By the triangle inequality, we have
    \begin{align}\label{eq:hybrid-triangle}
        \abs{
            \E_{u,\breve y} \Br{\mollifier_{k,b,\lambda}(Ax^{h, 0})}
            -
            \E_{u,\breve y} \Br{\mollifier_{k,b,\lambda}(Ax^{h, L})}
        } ~\leq~
        \sum_{\ell=1}^L \abs{
            \E_{u,\breve y} \Br{\mollifier_{k,b,\lambda}(Ax^{h, \ell-1})}
            -
            \E_{u,\breve y} \Br{\mollifier_{k,b,\lambda}(Ax^{h, \ell})}
        }\enspace.
    \end{align}
    Fix $\ell\in[L]$,
    let $B \coloneqq h^{-1}(\ell)$ be the $\ell$-th $\whd$-sparse bucket replaced in the $\ell$-th hybrid step,
    and let $s \in \pmcube{[n]\setminus B}$
    be the variables that $x^{h, \ell}$ and $x^{h, \ell-1}$ have in common.
    As in \cref{eq:single-hybrid-explain}, let $b' = b - A^{[n]\setminus B}\cdot s$
    and we have
    \begin{align*}
        \mollifier_{k,b,\lambda}(Ax^{h, \ell-1}) 
        = \mollifier_{k,b',\lambda}(A^B\cdot u_B)
        \quad\text{and}\quad
        \mollifier_{k,b,\lambda}(Ax^{h, \ell}) 
        = \mollifier_{k,b',\lambda}(A^B\cdot \breve{y}_B) \enspace.
    \end{align*}
    Thus, since $u_B$ and $\breve{y}_B$ are independent of $s$, we have
    \begin{align}\label{eq:one-step-in-hybrids}
        \abs{
            \E_{u,\breve y} \Br{\mollifier_{k,b,\lambda}(Ax^{h, \ell-1})}
            -
            \E_{u,\breve y} \Br{\mollifier_{k,b,\lambda}(Ax^{h, \ell})}
        } = &~ \abs{
            \E_{u,\breve y} \Br{\mollifier_{k,b',\lambda}(A^B u_B)}
            -
            \E_{u,\breve y} \Br{\mollifier_{k,b',\lambda}(A^B \breve{y}_B)}
        } \nonumber\\ \leq &~
        \E_{s}\biggl[
            \underbrace{\abs{
                \E_{u_B} \Br{\mollifier_{k,b',\lambda}(A^B u_B)}
                -
                \E_{\breve{y}_B} \Br{\mollifier_{k,b',\lambda}(A^B \breve{y}_B)}
            }}_{(\Diamond)}
        \biggr] \enspace.
    \end{align}
    Expanding $\mollifier_{k,b',\lambda}(A^B u_B)$ and $\mollifier_{k,b',\lambda}(A^B \breve y_B)$
    around $H^B u_B$ and $H^B \breve y_B$ up to order $d-1$, respectively,
    as in \cref{eq:taylor-expansion-mollifier}, we obtain
    \begin{equation}\label{eq:taylor-expansion-mollifier-hybrid}
    \resizebox{0.93\textwidth}{!}{$
    \begin{aligned}
    (\Diamond)
    \leq
    &\sum_{c = 0}^{d - 1} \frac{1}{c!}
    \sum_{i_1, \ldots, i_c=1}^m
    \Bigg|
        \E_{u_B}\Br{
            \partial_{i_1, \ldots, i_c}
            \mollifier_{k,b',\lambda}(H^B u_B)
            \prod_{t=1}^c (T^B u_B)_{i_t}
        }
    -
        \E_{\breve y_B}\Br{
            \partial_{i_1, \ldots, i_c}
            \mollifier_{k,b',\lambda}(H^B \breve y_B)
            \prod_{t=1}^c (T^B \breve y_B)_{i_t}
        }
    \Bigg| \\
    &\quad
    + \frac{1}{d!}\cdot\norm{\mollifier_{k,b',\lambda}^{(d)}}_1\cdot
    \left(
        \E_{u_B}\Br{\norm{T^B u_B}_\infty^d}
        +
        \E_{\breve y_B}\Br{\norm{T^B \breve y_B}_\infty^d}
    \right)\enspace.
    \end{aligned}
    $}
    \end{equation}
    \par\noindent Since $B$ is $\whd$-sparse with respect to $H$,
    each row of $H^B$ has at most $\whd$ non-zero entries by
    definition in \cref{eq:def-sparse-bucket}.
    Moreover, each row of $T$ is $\tau$-regular with norm $1$ (see \cref{eq:head-tail-decomposition}),
    and hence each row of $T^B$ has norm at most $1$.
    By the definition of $\breve y$ in \cref{def:threshold-ost-generator},
    $\breve y_B$ is the coordinate-wise $\xor$ of $y_B^\ell$ and $\widetilde{y}_B^\ell$, where $\widetilde{y}_B^\ell$ fools all width-$k\whd$ \CNF{} formulas with error at most $\deltaCNF$.
    Thus, $\breve y_B$ also $\deltaCNF$-fools all width-$k\whd$
    \CNF{} formulas.
    Applying \cref{lem:one-bucket-replacement} gives
    {\small\begin{align*}
        \Bigg|
            \E_{u_B}\Br{
                \partial_{i_1, \ldots, i_c}
                \mollifier_{k,b',\lambda}(H^B u_B)
                \prod_{t=1}^c (T^B u_B)_{i_t}
            } 
        -
            \E_{\breve y_B}\Br{
                \partial_{i_1, \ldots, i_c}
                \mollifier_{k,b',\lambda}(H^B \breve y_B)
                \prod_{t=1}^c (T^B \breve y_B)_{i_t}
            }
        \Bigg|
        = \deltaCNF \cdot O_c\!\br{ \frac{m^k \sqrt{n}}{\lambda} }^c .
    \end{align*}}
    \par\noindent Inserting this bound into
    \cref{eq:taylor-expansion-mollifier-hybrid}
    and using the bounds on the derivatives in \cref{thm:scaled} gives
    \begin{align*}
        (\Diamond)
        ~\leq~   \deltaCNF \cdot O \!\br{ \frac{m^{k+1} \sqrt{n}}{\lambda} }^{d-1}
        + ~
        O \!\br{ \frac{k\sqrt{\log(m/k)}}{\lambda} }^{d} 
        \cdot
        \left(
            \E_{u_B}\Br{\norm{T^B u_B}_\infty^d}
            +
            \E_{\breve y_B}\Br{\norm{T^B \breve y_B}_\infty^d}
        \right) \,.
    \end{align*}
    Combining this with \cref{eq:hybrid-triangle,eq:one-step-in-hybrids},
    we have
    \begin{align*}
        &\abs{
            \E_{u,\breve y} \Br{\mollifier_{k,b,\lambda}(Ax^{h, 0})}
            -
            \E_{u,\breve y} \Br{\mollifier_{k,b,\lambda}(Ax^{h, L})}
        }
        ~\leq~
        L \cdot \deltaCNF \cdot O\!\br{ \frac{m^{k+1} \sqrt{n}}{\lambda} }^{d-1} \\
        &\quad+~
        O\!\br{ \frac{k\sqrt{\log(m/k)}}{\lambda} }^d \cdot~
        \sum_{\ell = 1}^L 
        \left(
            \E_{u_{h^{-1}(\ell)}}\Br{\norm{T^{h^{-1}(\ell)} u_{h^{-1}(\ell)}}_\infty^d}
            +
            \E_{\breve y_{h^{-1}(\ell)}}\Br{\norm{T^{h^{-1}(\ell)} \breve y_{h^{-1}(\ell)}}_\infty^d}
        \right) \enspace.   
    \end{align*}
    \cref{eq:hybrid-goal} then follows from the fact $\breve y_{h^{-1}(\ell)}=\widehat y^\ell_{h^{-1}(\ell)}$. 
\end{proof}

\subsection{Fooling Taylor Terms}
\label{sec:cnf-derivative}

This section proves \cref{lem:one-bucket-replacement},
which shows that any random variable fooling width-$kw$ \CNF{} formulas also fools the Taylor terms,
provided that the head is $w$-sparse.
Before that, we first show that the Taylor terms can be expressed as
linear combinations of width-$kw$ \CNF{} formulas with bounded total weight,
which is the key structural consequence of sparse heads.

\begin{lemma}\label{lem:cnf-decomposition}
    Let $B$, $H^B$ and $T^B$ be as in \cref{lem:one-bucket-replacement}.
    For every $b\in\R^m$, $\lambda>0$, $c\in\N$ and $i_1, \ldots, i_c \in [m]$, define
    $f:\pmcube B \to \R$ by
    \[
    	f(x) \coloneqq \partial_{i_1, \ldots, i_c} \mollifier_{k,b,\lambda}(H^B x) \cdot (T^B x)_{i_1} \cdots (T^B x)_{i_c} \enspace,
    \]
    where $\mollifier_{k,b,\lambda}$ is the Bentkus-type mollifier defined in \cref{def:mollifier}
    for the $k$-out-of-$m$ threshold set $\ort_{k,b}$.
    Then $f$ is a combination of width-$kw$ \CNF{} formulas with
    weight at most $O_c\!\br{m^k \sqrt{n}/\lambda }^c$.
\end{lemma}

\begin{proof}
    Fix $b\in\R^m$, $\lambda>0$, $c\in\N$ and $I = (i_1, \ldots, i_c) \in [m]^c$.
    By the product rule for multivariate derivatives, we have~\footnote{For $c=0$,
    the constant $1$ remains, but the same argument applies. We omit this term for simplicity.}
    \begin{align*}
        \partial_I \mollifier_{k,b,\lambda}(H^B x)
        \coloneqq \partial_{i_1, \ldots, i_c} \mollifier_{k,b,\lambda}(H^B x) 
        &= -~\partial_{i_1, \ldots, i_c}\!\! \br{ \prod_{S\in \binom{[m]}{k}}
        \br{1- \prod_{\ell\in S} \widetilde{\ind}_{b_\ell,\lambda} (H^B_\ell x) }} \\
        &= -~\sum_{g: [c] \to \binom{[m]}{k}} \prod_{S\in \binom{[m]}{k}} \partial_{g^{-1}(S)}\!\!
        \br{1- \prod_{\ell\in S} \widetilde{\ind}_{b_\ell,\lambda} (H^B_\ell x) }\enspace,
    \end{align*}
    where $\partial_{g^{-1}(S)}$ denotes the partial derivative with respect to
    all $i_j$ such that $j\in g^{-1}(S)$.
    For each $S\in \binom{[m]}{k}$, the function
    $\partial_{g^{-1}(S)}\!\!
        \br{1- \prod_{\ell\in S} \widetilde{\ind}_{b_\ell,\lambda} (H^B_\ell x) }$
    is zero if $g^{-1}(S)$ contains an index $j$ such that $i_j \notin S$.
    Thus, we assume that $i_j\in S$ for every $j\in g^{-1}(S)$.
    For each $\ell\in S$, let $\alpha_{\ell,S}$ be the number of indices $j\in g^{-1}(S)$ such that $i_j = \ell$, so $\sum_{\ell\in S} \alpha_{\ell,S} = |g^{-1}(S)|$ and $\sum_{S}\sum_{\ell\in S} \alpha_{\ell,S} = c$.
    Then, we have
    {\small\begin{align*}
        &\partial_I \mollifier_{k,b,\lambda}(H^B x)\\
        = &-\!\! \sum_{\substack{
                g: [c] \to \binom{[m]}{k}\\
                i_j\in g(j)\ \forall j\in[c]
            }}
            \prod_{S\in \binom{[m]}{k}} 
        \left(
            -\prod_{\ell\in S}
            \widetilde{\ind}_{b_\ell,\lambda}^{(\alpha_{\ell,S})}(H^B_\ell x)
            \cdot \ind[g^{-1}(S)\neq\emptyset]
            +
            \br{1-\prod_{\ell\in S}
            \widetilde{\ind}_{b_\ell,\lambda}(H^B_\ell x)}
            \cdot \ind[g^{-1}(S)=\emptyset] 
        \right)\enspace.
    \end{align*}}
    \par\noindent
    Since each row of $H^B$ is supported on at most $w$ coordinates and $\abs{S} = k$,
    $1-\prod_{\ell\in S} \widetilde{\ind}_{b_\ell,\lambda}(H^B_\ell x)$
    is $[0,1]$-valued and depends on at most $kw$ coordinates.
    Hence, by \cref{fact:weight-w-combination}, it is a
    weight-$1$ combination of Boolean $kw$-juntas.
    By \cref{fact:weight-w-combination} and
    \cref{fact:derivative-smoothed-indicator},
    each $\widetilde{\ind}_{b_\ell,\lambda}^{(\alpha_{\ell,S})}(H^B_\ell x)$
    is a weight-$O(\lambda^{-\alpha_{\ell,S}})$ combination of Boolean $w$-juntas,
    with weight $1$ when $\alpha_{\ell,S}=0$.
    Therefore,
    $\prod_{\ell\in S}\widetilde{\ind}_{b_\ell,\lambda}^{(\alpha_{\ell,S})}(H^B_\ell x)$
    is a weight
    $O_c\!\br{\lambda^{-\sum_{\ell\in S}\alpha_{\ell,S}}}$
    combination of products of Boolean $w$-juntas.
    Thus, $\partial_I \mollifier_{k,b,\lambda}(H^B x)$ is a combination of
    products of Boolean $kw$-juntas, and the weight is bounded by
    \[
    	\binom{m}{k}^c \cdot O_c\!\br{ \frac{1}{\lambda} }^c
        ~=~
        O_c\!\br{ \frac{m^k}{\lambda} }^c \enspace.
    \]
    In other words, we can write
    $\partial_I \mollifier_{k,b,\lambda}(H^B x) = \sum_{j} c_j \cdot u_j(x)$
    where each $u_j$ is a product of Boolean $kw$-juntas
    and $\sum_j |c_j| \le O_c\!\br{ m^k/\lambda }^c$.

    It remains to handle the tail monomial.
    For each $t\in[c]$,
    since each row of $T^B$ has norm at most $1$,
    the Cauchy--Schwarz inequality gives
    $ \lVert T^B_{i_t} \rVert_1 \le \sqrt n $.
    Hence, $(T^B x)_{i_t}$ is a weight-$2\sqrt n$ combination of Boolean $1$-juntas.
    Therefore, $(T^B x)_{i_1}\cdots (T^B x)_{i_c}$
    is a weight-$O_c(n^{c/2})$ combination of products of Boolean $1$-juntas.

    Combining the above decompositions, we can write $f(x)=\sum_j a_j \cdot F_j(x)$,
    where each $F_j$ is a product of Boolean $kw$-juntas and Boolean $1$-juntas, and
    \[
        \sum_j |a_j|
        ~\le~
        O_c\!\br{\frac{m^k}{\lambda}}^c \cdot~ O_c(n^{c/2})
        ~=~
        O_c\!\br{\frac{m^k\sqrt n}{\lambda}}^c .
    \]
    Since products correspond to conjunctions,
    each $F_j$ is computable by a width-$kw$ \CNF{} formula.
    Hence, $f$ is a combination of width-$kw$ \CNF{}
    formulas with the claimed weight.
\end{proof}

Now it is straightforward to prove \cref{lem:one-bucket-replacement}.
We restate the lemma here for convenience.

\bucketreplacement*

\begin{proof}
    Let $f$ be as in \cref{lem:cnf-decomposition}.
    By \cref{lem:cnf-decomposition}, we can write $f(x)=\sum_j a_j \cdot F_j(x)$, where each $F_j$ is a width-$kw$ \CNF{} formula and $\sum_j |a_j| \le O_c\!\br{ \frac{m^k \sqrt{n}}{\lambda} }^c$.
    Therefore,
    \begin{align*}
        &\abs{
            \E_{u} \Br{\partial_{i_1, \ldots, i_c} \mollifier_{k,b,\lambda}(H^B u) \cdot (T^B u)_{i_1} \cdots (T^B u)_{i_c}}
            -
            \E_{y} \Br{\partial_{i_1, \ldots, i_c} \mollifier_{k,b,\lambda}(H^B y) \cdot (T^B y)_{i_1} \cdots (T^B y)_{i_c}}
        } \\
        \le ~&\sum_j |a_j| \cdot \abs{\E_u[F_j(u)] - \E_y[F_j(y)]} \\
        \le ~&\deltaCNF \cdot O_c\!\br{ \frac{m^k \sqrt{n}}{\lambda} }^c \,,
    \end{align*}
    where the last step uses the assumption that $y$ $\deltaCNF$-fools width-$kw$ \CNF{} formulas.
\end{proof}

\subsection{Fooling Thresholds of General Halfspaces}
\label{sec:threshold-prg}

We first show that the generator from \cref{def:threshold-ost-generator} fools
standardized thresholds of standardized halfspaces,
and then combine this with \cref{cor:threshold-outer-reduction}
to prove the main $\prg$ theorem for thresholds of general halfspaces.

\begin{theorem}\label{thm:prg-standardized}
    Let all parameters be chosen as in \cref{fig:threshold-parameters},
    and let $G$ be the generator from \cref{def:threshold-ost-generator}.
    Then, for every $(\Khead,\tau)$-standardized matrix $A\in\R^{m\times n}$ and every $b\in\R^m$,
    \[
        \abs{
        \E_{u\sim \pmcube n}\Br{\ort_{k,b}(Au)}
        -
        \E_{z\sim G}\Br{\ort_{k,b}(Az)}
        } ~=~ O(\delta)\enspace,
    \]
    where $\ort_{k,b}$ is the indicator of the $k$-out-of-$m$ threshold set defined in \cref{def:threshold-set}.
    \end{theorem}

\begin{proof}
    Let $\lambda$ be as in \cref{fig:threshold-parameters}.
    By \cref{lem:approximator}, $\mollifier_{k,b^{\newsf{in}},\lambda}$ and
    $\mollifier_{k,b^{\newsf{out}},\lambda}$ are
    $(\Lambda,\delta)$-inner and $(\Lambda,\delta)$-outer approximators for
    the $k$-out-of-$m$ threshold set $\ort_{k,b}$, respectively,
    where $\Lambda = \Theta(\lambda\sqrt{\log(m/\delta)})$.
    By \cref{thm:threshold-mollifier-fool}, we have that for $\mollifier\in \{\mollifier_{k,b^{\newsf{in}},\lambda}, \mollifier_{k,b^{\newsf{out}},\lambda}\}$
    \[
    	\abs{
            \E_{u\sim \pmcube n}\Br{\mollifier(Au)}
            -
            \E_{z\sim G}\Br{\mollifier(Az)}
        }
        ~=~ O(\delta) \enspace.
    \]
    Since $A$ is $(\Khead,\tau)$-standardized and $\Lambda\geq \tau$, \cref{thm:threshold-ost-anticonc} gives
    \[
    	\Pr_{u\sim \pmcube n}\Br{Au\in \ort_{k,b +\Lambda\cdot\onevec} \setminus \ort_{k,b -\Lambda\cdot\onevec}} ~=~ O\!\br{k\cdot\Lambda\sqrt{\log(m/k)}} \enspace.
    \]
    Applying \cref{lem:sandwich}, we conclude that
    \begin{align}\label{eq:lambda}
        \abs{
        \E_{u\sim \pmcube n}\Br{\ort_{k,b}(Au)}
        -
        \E_{z\sim G}\Br{\ort_{k,b}(Az)}
        } ~=~ O(\delta) ~+~ O\!\br{k\cdot\Lambda\sqrt{\log(m/k)}} ~=~ O(\delta) \enspace.
    \end{align}
    This completes the proof.
\end{proof}

We are now ready to prove the main $\prg$ theorem for thresholds of general halfspaces.

\begin{theorem}\label{thm:main-prg}
    Let \(n,m\in\N\), \(1\le k\le m\), and \(\delta\in(0,1)\).
    Set \(\kappa=\min\{k,m-k+1\}\).
    Then there is an explicit $\prg$ with seed length $\tildeo\Bigl(\dfrac{\kappa^{6+2\varepsilon} \cdot \log^{6+2\varepsilon}\!m}{\delta^{2 + 2\varepsilon}}\cdot \log n\Bigr)$
    for any arbitrarily small constant $\varepsilon > 0$ 
    that $\delta$-fools the class of all $k$-out-of-$m$ thresholds of halfspaces over $\pmcube n$.
\end{theorem}

\begin{proof}
It suffices to prove the theorem for \(k\le m/2\), in which case
\(\kappa=k\). The case \(k>m/2\) follows by applying the same argument
to the complementary \((m-k+1)\)-out-of-\(m\) threshold of the
complemented halfspaces.

    Fix parameters as in \cref{fig:threshold-parameters}.
    Let $G$ be the generator from \cref{def:threshold-ost-generator}.
    If $\Khead > n/2$, then the generator has seed length $n$, so it trivially fools every function on $\pmcube n$.
    Therefore, we assume $\Khead \leq n/2$.
    Given a $k$-out-of-$m$ threshold of halfspaces $F(x)$,
    we can write $F(x)$ as $\ort_{k,b}(Ax)$ for some $A\in\R^{m\times n}$ and $b\in\R^m$
    as in \cref{rmk:equivalent-representation}.
    By \cref{cor:threshold-outer-reduction},
    there is a $(\Khead,\tau)$-standardized matrix $A'$ and
    a vector $b'\in\R^m$ such that for uniform $u\sim \pmcube n$,
    \begin{align}\label{eq:approx-1}
        \abs{\E_{u\sim \pmcube n}\Br{\ort_{k,b}(Au)}
        -
        \E_{u\sim \pmcube n}\Br{\ort_{k,b'}(A'u)}
        } ~=~ O(\delta) \enspace.
    \end{align}
    Similarly, since $z\sim G$ is $2\Khead$-wise uniform, we have
    \begin{align}\label{eq:approx-2}
        \abs{\E_{z\sim G}\Br{\ort_{k,b}(Az)}
        -
        \E_{z\sim G}\Br{\ort_{k,b'}(A'z)}
        } ~=~ O(\delta) \enspace.
    \end{align}
    By \cref{thm:prg-standardized}, we have
    \begin{align}\label{eq:approx-3}
        \abs{\E_{u\sim \pmcube n}\Br{\ort_{k,b'}(A'u)}
        -
        \E_{z\sim G}\Br{\ort_{k,b'}(A'z)}
        } ~=~ O(\delta) \enspace.
    \end{align}
    Combining \cref{eq:approx-1,eq:approx-2,eq:approx-3} gives
    \begin{align*}
        \abs{\E_{u\sim \pmcube n}\Br{\ort_{k,b}(Au)}
        -
        \E_{z\sim G}\Br{\ort_{k,b}(Az)}
        } ~=~ O(\delta) \enspace,
    \end{align*}
    as claimed.
    To get error $\delta$, we can simply rescale the parameters by a constant factor.
    The seed length of $G$ is $\tildeo\Bigl(\dfrac{k^{6+2\varepsilon} \cdot \log^{6+2\varepsilon}\!m}{\delta^{2 + 2\varepsilon}}\cdot \log n\Bigr)$ by \cref{lem:ost-generator-seed}.
\end{proof}

\section{Derivative Bounds for the Bentkus-Type Mollifier}
\label{sec:derivative}
This section proves a derivative bound for the Bentkus-type mollifier
used in the analysis of the Taylor remainder.
To this end, we first introduce a normalized form for the mollifier.
Recall that \(\gcdf\) is the CDF of the standard Gaussian distribution.

\begin{definition}
    For integers $m\ge 2$ and $1\le k\le m/2$,
    we define a function $\mollifierkernel_{m,k}:\R^m\to\R$ by 
    \[
    \mollifierkernel_{m,k}(x)
    \coloneqq
    \prod_{S \in \binom{[m]}{k}}
    \left(
    1 - \prod_{i \in S} \gcdf(x_i)
    \right)\enspace.
    \]
    When the parameters $m$ and $k$ are clear from the context, we simply write $\mollifierkernel$ for $\mollifierkernel_{m,k}$.
\end{definition}

\begin{remark}\label{rmk:mollifier-rewrite}
    The Bentkus-type $\lambda$-mollifier $\mollifier_{k,b,\lambda}$ defined in \cref{def:mollifier} for the $k$-out-of-$m$ threshold set
    can be written as
    \[
        \mollifier_{k,b,\lambda}(x)
        =
            1-\mollifierkernel_{m,k}\!\left(\frac{b-x}{\lambda}\right) 
        \quad\text{where}\quad
        \frac{b-x}{\lambda}
        \coloneqq
        \left(
            \frac{b_1-x_1}{\lambda}\enspace,
            \dots\enspace,
            \frac{b_m-x_m}{\lambda}
        \right)\enspace.
    \]
\end{remark}

The main result of this section is the following uniform derivative bound for $\mollifierkernel_{m,k}$.

\begin{restatable}{theorem}{fullderivative}\label{thm:derivative-bound}
For all integers $m,k,d$ with $m\ge 2$, $1\le k\le m/2$, and $d\ge 1$,
there exists a constant $C_d>0$ depending only on $d$ such that
\[
    \norm{\mollifierkernel_{m,k}^{(d)}}_1
    ~\leq~
    C_d \cdot k^d \cdot \br{\log\frac{m}{k}}^{d/2}\enspace.
\]
\end{restatable}

This theorem immediately gives the corresponding bound for derivatives of $\mollifier_{k,b,\lambda}$
as stated in \cref{thm:scaled}.
We restate the theorem here for the reader's convenience.

\derivatives*

\begin{proof}
By \cref{rmk:mollifier-rewrite} and the chain rule,
for each ordered derivative tuple $(j_1,\dots,j_d)$ we have
\[
    \partial_{j_1\cdots j_d}\mollifier_{k,b,\lambda}(x)
    =
    (-1)^{d+1} \cdot \lambda^{-d} \cdot
    \left(
    \partial_{j_1\cdots j_d}\mollifierkernel
    \right)\!\!\left(\frac{b-x}{\lambda}\right)\enspace.
\]
Therefore
\[
    \sum_{j_1,\dots,j_d=1}^m
    \left|
    \partial_{j_1\cdots j_d}\mollifier_{k,b,\lambda}(x)
    \right|
    =
    \frac1{\lambda^d}\cdot
    \sum_{j_1,\dots,j_d=1}^m
    \left|
    \left(
    \partial_{j_1\cdots j_d}\mollifierkernel
    \right)\!\!\left(\frac{b-x}{\lambda}\right)
    \right|
    \le
    \frac{1}{\lambda^d}\cdot\norm{\mollifierkernel^{(d)}}_1\enspace.
\]
Taking the supremum over $x$ and applying \cref{thm:derivative-bound} yields the claimed bound.
\end{proof}

The rest of this section is devoted to proving \cref{thm:derivative-bound}.
The key step is to carefully control the derivatives of $\prod_{i\in S}\gcdf(x_i)$
for every subset $S\subseteq [m]$ with size $k$.
Indeed, every factor in $\mollifierkernel_{m,k}$ has
the form $1-\prod_{i\in S}\gcdf(x_i)$,
and the product rule expresses $\Vert \mollifierkernel_{m,k}^{(d)} \Vert_1$
in terms of sums of derivatives of these clause functions.
\cref{app:gaussian-derivative} proves the derivative estimate for a single $\gcdf(x_i)$.
This estimate is then lifted to one clause $\prod_{i\in S}\gcdf(x_i)$ in \cref{sec:clause-derivative}.
In \cref{sec:sum-clause-bound}, we show that
the sum of these clausewise bounds over all $S$ remains controlled,
avoiding the trivial factor $\binom{m}{k}$.
Finally, \cref{sec:proof-derivative-bound} proves \cref{thm:derivative-bound}.

\subsection{Derivative Bounds for the Gaussian CDF}
\label{app:gaussian-derivative}

This section bounds the derivatives of the Gaussian CDF $\gcdf$
using $\gcdf$ itself.
We start with Mills' ratio inequality,
which gives a lower bound on the Gaussian tail probability $\gcdf(-z)$
in terms of the Gaussian density $\gpdf(z)$.

\begin{lemma}[Mills' Ratio Inequality \cite{Gor41}]\label[lemma]{lem:mills-ratio}
    For every $z\geq 0$,
    \(
    \dfrac{z}{1+z^2}\cdot \gpdf(z)
    \leq
    \gcdf(-z)
    \).
\end{lemma}

We will use the following consequence,
which remains useful when $z$ is close to zero.

\begin{corollary}\label[corollary]{cor:mills}
For every $z\ge0$,
\(
    \dfrac{\gpdf(z)}{2(1+z)}
    \le
    \gcdf(-z).
\)
\end{corollary}

\begin{proof}
Let
\(
    f(z) \coloneqq \dfrac{\gpdf(z)}{\gcdf(-z)} 
\)
and therefore
\(
    f'(z)
    =
    \dfrac{\gpdf(z)\cdot (\gpdf(z)-z\cdot\gcdf(-z))}{\gcdf(-z)^2}
\),
where we use the fact that $\gpdf'(z)=-z\cdot\gpdf(z)$.
Note that for $z > 0$,
\[
    \gcdf(-z)
    ~=~
    \int_z^\infty \gpdf(t)\,\d t
    ~\le~
    \frac1z \int_z^\infty t\cdot \gpdf(t)\,\d t
    ~=~
    \frac{\gpdf(z)}{z}\enspace,
\]
and hence
\(
    f'(z) \ge 0
\).
Thus, $f$ is increasing on $(0,\infty)$.
If $0\le z\le1$, then
\[
    \frac{\gpdf(z)}{\gcdf(-z)}
    ~=~
    f(z)
    ~\le~
    f(1)
    ~=~
    \frac{\gpdf(1)}{\gcdf(-1)} \le 2 \le 2(1+z) \enspace,
\]
where the second inequality follows from \cref{lem:mills-ratio} at $z=1$.
As for $z\ge1$, \cref{lem:mills-ratio} gives
\[
    \frac{\gpdf(z)}{\gcdf(-z)}
    ~\le~
    \frac{1+z^2}{z}
    ~=~
    z+\frac1z
    ~\le~
    z+1
    ~\le~
    2(1+z)\enspace. \qedhere
\]
\end{proof}

We now prove the Gaussian CDF derivative estimate needed for the clausewise
bound.

\begin{lemma}\label[lemma]{lem:oned}
For every integer $r\in\N$,
there exists a constant $C_{1,r}>0$ depending only on $r$
such that for every $x\in\R$,
\[
    \left|\gcdf^{(r)}(x)\right|
    ~\le~
    C_{1,r}\cdot\gcdf(x)\cdot
    \left(
    1+\log\frac1{\gcdf(x)}
    \right)^{r/2}\enspace.
\]
\end{lemma}

\begin{proof}
The case $r=0$ is immediate by taking $C_{1,0}=1$, and hence we assume $r\ge1$.
Let $\He_j$ denote the \emph{$j$-th probabilists' Hermite polynomial} \cite[Section 11]{O14}
defined through the Gaussian density $\gpdf$ by $\gpdf^{(j)}\!(x)=(-1)^j\cdot \He_j(x)\cdot \gpdf(x)$.
Therefore, we have
\[
\gcdf^{(r)}\!(x) ~=~ \gpdf^{(r-1)}\!(x) ~=~ (-1)^{r-1}\cdot \He_{r-1}(x)\cdot\gpdf(x)\enspace.
\]
Since $\He_{r-1}$ has degree $r-1$, there is a constant $A_r>0$ depending only on $r$ such that
\[
    |\He_{r-1}(x)|
    \le
    A_r\cdot (1+|x|^{r-1})
    \qquad\text{for all }x\in\R\enspace.
\]
Therefore,
\[
    |\gcdf^{(r)}(x)|
    \le
    A_r\cdot (1+|x|^{r-1})\cdot \gpdf(x)
    \qquad\text{for all }x\in\R\enspace.
\]

We first assume that $x\ge 0$.
Since the function $x\mapsto (1+|x|^{r-1})\cdot\gpdf(x)$ is bounded on $[0,\infty)$,
there exists a constant $B_r>0$ depending only on $r$ such that
\(
    |\gcdf^{(r)}(x)|
    \le
    B_r.
\)
Notice that $\gcdf(x)\ge \gcdf(0)=\frac12$ for $x\ge0$. Hence
\[
    |\gcdf^{(r)}(x)|
    ~\le~
    B_r
    ~\le~
    2B_r \cdot \gcdf(x)
    ~\le~
    2B_r \cdot \gcdf(x) \cdot \left(1+\log\frac1{\gcdf(x)}\right)^{r/2}\enspace.
\]
Now suppose $x<0$, and let $z=-x$ with $z>0$. By \cref{cor:mills},
\(
\dfrac{\gpdf(z)}{\gcdf(-z)}
\le
2(1+z).
\)
Therefore,
there exists a constant $D_r>0$ depending only on $r$, such that
\begin{equation}\label{eq:single-derivative-bound-1}
    \frac{\abs{\gcdf^{(r)}(x)}}{\gcdf(x)}
    ~\le~
    2A_r\cdot (1+z^{r-1})(1+z)
    ~\le~
    D_r \cdot (1+z^r) \enspace.
\end{equation}
By the standard tail bound
\(
    \gcdf(-z)\le e^{-z^2/2}
\),
we have
\(
z^2 \le 2\log\dfrac{1}{\gcdf(x)}
\)
and therefore
\[
    1+z^r
    ~\le~
    (1+2^{r/2})\cdot\left(1+\log\frac1{\gcdf(x)}\right)^{r/2} \enspace.
\]
Substituting this into \cref{eq:single-derivative-bound-1} gives  
\[
    |\gcdf^{(r)}(x)|
    ~\le~
    D_r\cdot(1+2^{r/2})\cdot  \gcdf(x) \cdot \left(1+\log\frac1{\gcdf(x)}\right)^{r/2}\enspace.
\]
Choosing $C_{1,r} = \max\{ 2 B_r, D_r\cdot(1+2^{r/2})\}$ completes the proof.
\end{proof}

\subsection{Derivative Bounds for a Clause}
\label{sec:clause-derivative}

This section lifts the previous bound for derivatives of a single Gaussian CDF
to products of multiple Gaussian CDFs.

\begin{lemma}\label[lemma]{lem:clause}
For every integer $r\ge1$, there exists a constant $C_{2,r}>0$ depending only on $r$
such that the following holds:
Let $m$ and $k$ be integers with $m\ge2$ and
$1\le k\le m/2$, and let $S\subseteq[m]$ satisfy $|S|=k$. Define
$\clauseprob_S(x)\coloneqq\prod_{i\in S}\gcdf(x_i)$. Then, for every
$x\in\R^m$,
\[
    \norm{\clauseprob_S^{(r)}\!(x)}_1
    ~\le~
    C_{2,r}\cdot k^r\cdot \clauseprob_S(x)\cdot
    \left(
    1+\frac1k \cdot \log\frac1{\clauseprob_S(x)}
    \right)^{r/2}\enspace.
\]
\end{lemma}

\begin{proof}
    Fix $r\geq 1$, $S\subseteq[m]$ with $|S|=k$ and $x\in\R^m$.
    If at least one of the indices $i_1,\dots,i_r$ is not in $S$, then
    $\partial_{i_1\cdots i_r} \clauseprob_S(x)=0$, since
    $\clauseprob_S$ depends only on the coordinates $x_i$ with $i\in S$.
    Thus, fix an $r$-tuple $(i_1,\dots,i_r)\in S^r$.
    For each $j\in S$, let $\alpha_j$ be the multiplicity of $j$ in the tuple:
    $\alpha_j \coloneqq \#\{t\in[r] : i_t=j\}$.
    Thus, $\alpha_j\ge0$ for all $j\in S$ and
    \(
        \sum_{j\in S}\alpha_j = r
    \).
    Using this notation, we can write the derivative as
    \[
        \partial_{i_1\cdots i_r} \clauseprob_S(x)
        ~=~
        \prod_{j\in S} \gcdf^{(\alpha_j)}(x_j)
        ~=~
        \clauseprob_S(x)\cdot
        \prod_{\substack{j\in S\\ \alpha_j\ge1}}
        \frac{\gcdf^{(\alpha_j)}(x_j)}{\gcdf(x_j)}\enspace.
    \]
    By \cref{lem:oned}, we have
    \begin{equation}\label{eq:single-derivative-bound}
        \abs{
        \partial_{i_1\cdots i_r} \clauseprob_S(x)
        }
        ~\le~ 
        A_{\alpha}\cdot \clauseprob_S(x)\cdot
        \prod_{j\in S} \br{1+\log\frac1{\gcdf(x_j)}}^{\alpha_j/2}
        \enspace.
    \end{equation}
    where
    \(
        A_{\alpha} \coloneqq
        \prod_{j\in S,\ \alpha_j\ge1} C_{1,\alpha_j}
    \).
    Since at most $r$ of the $\alpha_j$'s are nonzero and each nonzero
    $\alpha_j$ lies in $\{1,\dots,r\}$,
    $A_{\alpha}$ is bounded by a
    constant $C_{2,r}>0$ depending only on $r$.
    Since $\alpha_j$ is the multiplicity of $j$ in $(i_1,\dots,i_r)$, we have
    \begin{equation}\label{eq:equal-product}
        \prod_{j\in S}
        \left(1+\log\frac1{\gcdf(x_j)}\right)^{\alpha_j/2}
        =~~~
        \prod_{t=1}^r
        \left(1+\log\frac1{\gcdf(x_{i_t})}\right)^{1/2}\enspace.
    \end{equation}
    Summing over all $(i_1,\dots,i_r)\in S^r$ on both sides of \cref{eq:single-derivative-bound}
    and applying \cref{eq:equal-product}, we obtain
    \begin{align*}
        \norm{\clauseprob_S^{(r)}(x)}_1
        &~\le~
        C_{2,r}\cdot \clauseprob_S(x)\cdot
        \sum_{(i_1,\dots,i_r)\in S^r}
        \prod_{t=1}^r \br{1+\log\frac1{\gcdf(x_{i_t})}}^{1/2}\\
        &~=~
        C_{2,r}\cdot\clauseprob_S(x)\cdot
        \left(
        \sum_{j\in S} \br{1+\log\frac1{\gcdf(x_j)}}^{1/2}
        \right)^r\enspace.
    \end{align*}
    Applying the Cauchy--Schwarz inequality, we have
    \begin{align*}
        \norm{\clauseprob_S^{(r)}(x)}_1
        &\leq
        C_{2,r}\cdot \clauseprob_S(x)\cdot
        k^{r/2}\cdot \left(\sum_{j\in S} \br{1+\log\frac1{\gcdf(x_j)}}\right)^{r/2} \\
        &=
        C_{2,r}\cdot \clauseprob_S(x)\cdot
        k^{r/2}\cdot \left(k+\log\frac1{\clauseprob_S(x)}\right)^{r/2}\\
        &=
        C_{2,r}\cdot \clauseprob_S(x)\cdot
        k^{r}\cdot \left(1+\frac1k\log\frac1{\clauseprob_S(x)}\right)^{r/2}\enspace.
    \end{align*}
    This completes the proof.
\end{proof}

\subsection{Summing the Clause Bounds}
\label{sec:sum-clause-bound}
In this section, we show that the derivative bounds for individual clauses
can be summed over all $k$-subsets with only polynomial overhead in $k$ and $\log m$,
rather than the trivial factor $\binom{m}{k}$.
The key intuition comes from the form of the clause bound in \cref{lem:clause}:
the sum to be controlled can be viewed as a weighted sum of
$(1+\frac1k\log\frac1{\clauseprob_S(x)})^{r/2}$
with weights $\clauseprob_S(x)$.
Large values of the logarithmic factor can occur only when
$\clauseprob_S(x)$ is small, and then the corresponding weight is exponentially
small in that logarithmic parameter.
Thus, large terms contribute little to the sum,
which prevents the full $\binom{m}{k}$ factor from appearing.

\begin{lemma}\label[lemma]{lem:moment}
For every integer $r\ge1$,
there exists a constant $C_{3,r}>0$ depending only on $r$
such that the following holds:
Let $m$ and $k$ be integers with $m\ge2$ and $1\le k\le m/2$,
and let $N = \binom{m}{k}$.
For each $S\in\binom{[m]}{k}$, define
$\clauseprob_S(x)\coloneqq\prod_{i\in S}\gcdf(x_i)$, and define
$\clausemass(x)\coloneqq \sum_{S\in\binom{[m]}{k}} \clauseprob_S(x)$.
Then, for every $x\in\R^m$,
\[
    \sum_{S\in\binom{[m]}{k}}
    \norm{ \clauseprob_S^{(r)}\!(x) }_1
    ~\leq~
    C_{3,r}\cdot k^r \cdot \clausemass(x)\cdot
    \left[
    1+
    \left(
    \frac{1}{k}\cdot \log\frac{N}{\clausemass(x)}
    \right)^{r/2}
    \right]\enspace.
\]
\end{lemma}

\begin{proof}
By \cref{lem:clause}, it suffices to show that for $\tau = r/2$,
there exists a constant $C_{4,\tau}>0$ depending only on $\tau$ such that
for every $x\in\R^m$,
\[
    \clausemoment_\tau(x) ~\coloneqq~
    \sum_{S\in\binom{[m]}{k}}
    \clauseprob_S(x)
    \left(
    1+\frac1k\cdot\log\frac1{\clauseprob_S(x)}
    \right)^\tau
    ~\le~
    C_{4,\tau}\cdot \clausemass(x)\cdot
    \left[
    1+
    \left(
    \frac{1}{k}\cdot \log\frac{N}{\clausemass(x)}
    \right)^\tau
    \right]\enspace.
\]
Fix $x\in\R^m$.
For brevity, write $\clauseprob_S \coloneqq \clauseprob_S(x)$
and $\clausemass \coloneqq \clausemass(x)$.
Since each $\clauseprob_S$ lies in $(0,1)$, we have $0<\clausemass<N$.
Define a probability measure $\mu$ on $\binom{[m]}{k}$ by $\mu(S) \coloneqq \frac{\clauseprob_S}{\clausemass}$,
and define $Z_S \coloneqq \frac1k  \cdot \log\frac1{\clauseprob_S}$.
With this notation, $\clausemoment_\tau$ can be written as
\[
\clausemoment_\tau(x) ~=~ \clausemass\cdot\expect{S\sim\mu}{(1+Z_S)^\tau}\enspace.
\]
For every $\tau>0$, there is a constant $A_\tau>0$ such that
\(
    (1+z)^\tau \le A_\tau(1+z^\tau)
\)
for all $z\ge0$.
Therefore, we have
\begin{align} \label{eq:moment-tau}
    \clausemoment_\tau(x)
    ~\le~
    A_\tau \cdot \clausemass\cdot
    \expect{S\sim\mu}{1+Z_S^\tau}
    ~=~
    A_\tau \cdot \clausemass\cdot
    \left(1+\expect{S\sim\mu}{Z_S^\tau}\right)\enspace.
\end{align}
Thus, it suffices to give a bound for
$\expect{S\sim\mu}{Z_S^\tau}$.

Consider the moment-generating function of $Z_S$.
For $\beta\in(0,1)$,
\[
    \expect{S\sim\mu}{\e^{\beta k Z_S}}
    ~=~
    \sum_{S\in\binom{[m]}{k}} \frac{\clauseprob_S}{\clausemass} \cdot \e^{\beta k Z_S}
    ~=~
    \frac1\clausemass\cdot \sum_{S\in\binom{[m]}{k}} \clauseprob_S^{1-\beta}\enspace.
\]
Then, by Jensen's inequality, we have
\[
    \expect{S\sim\mu}{\e^{\beta k Z_S}}
    ~=~
    \frac1\clausemass\cdot \sum_{S\in\binom{[m]}{k}} \clauseprob_S^{1-\beta}
    ~\leq~
    \frac1\clausemass\cdot N^\beta \cdot
    \left(
    \sum_{S\in\binom{[m]}{k}} \clauseprob_S
    \right)^{1-\beta}
    ~=~
    \left(\frac{N}{\clausemass}\right)^\beta\enspace.
\]
Taking $\beta=\frac12$ and applying Markov's inequality,
for every $t\ge0$,
\[
	\Pr_{S\sim\mu}[Z_S\ge t]
    \le
    \frac{ \expect{S\sim\mu}{\e^{k Z_S/2}} }{\e^{k t/2}}
    \leq
    \exp\!\left(\frac12\br{\log\frac{N}{\clausemass}-kt}\right)\enspace.
\]
Now set a threshold $T \coloneqq \frac{1}{k} \cdot \log\frac{N}{\clausemass}$.
We have that for every $0 \le t \le T$, $\Pr_{S\sim\mu}[Z_S\ge t] \le 1$,
and for every $t\ge T$, $\Pr_{S\sim\mu}[Z_S\ge t] \le \e^{-k(t-T)/2}$.
Hence, 
\begin{align*}
    \expect{S\sim\mu}{Z_S^\tau}
    ~=~
    \int_0^\infty \Pr_{S\sim\mu}[Z_S^\tau \ge t]\,\d t
    ~&=~
    \int_0^\infty \Pr_{S\sim\mu}[Z_S \ge t^{1/\tau}]\,\d t \\
    ~&=~
    \tau \cdot \int_0^\infty t^{\tau-1} \cdot \Pr_{S\sim\mu}[Z_S \ge t]\,\d t \\
    ~&\le~
    \tau \cdot \int_0^T t^{\tau-1}\,\d t
    ~+~
    \tau \cdot \int_T^\infty t^{\tau-1}\cdot \e^{-k(t-T)/2}\,\d t\enspace.
\end{align*}
The first integral equals $T^\tau$.
As for the second integral, we consider two cases:
\begin{itemize}
    \item If $0<\tau<1$, then $t^{\tau-1}$ is decreasing on $(0,\infty)$, and hence
    \begin{align*}
        \int_T^\infty t^{\tau-1}\cdot \e^{-k(t-T)/2}\, \d t
        ~&=~
        \int_0^\infty (T+v)^{\tau-1} \cdot \e^{-kv/2}\, \d v\\
        ~&\le~
        \int_0^\infty v^{\tau-1} \cdot \e^{-kv/2}\, \d v
        ~=~
        \left(\frac{2}{k}\right)^\tau \cdot \Gamma(\tau)
        ~\le~
        2^\tau \cdot \Gamma(\tau)
        \enspace.
    \end{align*}
    \item If $\tau\ge1$,
    we have that for every $T,v\ge0$,
    $(T+v)^{\tau-1} \le 2^{\tau-1}\cdot\bigl(T^{\tau-1}+v^{\tau-1}\bigr)$.
    Therefore, using $k\ge1$, we have
    \[
        \begin{aligned}
        \int_T^\infty t^{\tau-1}\cdot \e^{-k(t-T)/2}\,\d t
        ~&=~
        \int_0^\infty (T+v)^{\tau-1}\cdot \e^{-kv/2}\,\d v \\
        ~&\le~
        2^{\tau-1}\cdot
        \left(
        T^{\tau-1} \cdot \int_0^\infty \e^{-kv/2}\,\d v
        ~+~
        \int_0^\infty v^{\tau-1} \cdot \e^{-kv/2}\,\d v
        \right) \\
        ~&=~
        2^{\tau-1}
        \left(
        \frac{2}{k} \cdot T^{\tau-1}
        ~+~
        \left(\frac{2}{k}\right)^\tau \cdot\Gamma(\tau)
        \right) \\
        ~&\le~
        2^{2\tau}\cdot \Gamma(\tau) \cdot (1+T^\tau)\enspace.
        \end{aligned}
    \]
\end{itemize}
Thus, for some constant $B_{\tau}>0$ depending only on $\tau$,
$\expect{S\sim\mu}{Z_S^\tau} \le B_{\tau}\cdot (1 + T^\tau)$.
Inserting this into \cref{eq:moment-tau} gives
\[
    \clausemoment_\tau(x)
    ~\le~
    A_\tau \cdot \clausemass\cdot
    \left(1+B_{\tau}\cdot (T^\tau+1)\right)
    ~\le~
    A_\tau \cdot (B_{\tau}+1) \cdot \clausemass\cdot
    \left[
    1 +
    \left(\frac{1}{k}\cdot \log\frac{N}{\clausemass}\right)^\tau
    \right]\enspace.
\]
Letting $C_{4,\tau} = A_\tau \cdot (B_{\tau}+1)$ completes the proof.
\end{proof}

\subsection{\texorpdfstring{Proof of \cref{thm:derivative-bound}}{Proof of Theorem 6.3}}
\label{sec:proof-derivative-bound}

We are ready to prove the main derivative bound of this section,
which is restated below.

\fullderivative*

\begin{proof}
Fix $m,k,d$ with $m\ge2$, $1\le k\le m/2$ and $d\ge1$,
and fix $x\in\R^m$.
We prove a uniform bound for $\Vert\mollifierkernel^{(d)}(x)\Vert_1$,
and the theorem then follows by taking the supremum over all $x\in\R^m$.
We first introduce some notation.
Let $\S \coloneqq \{S \subseteq [m] : |S| = k\}$.
For each $S\in\S$, define
\[
    \clauseprob_S~\coloneqq~\clauseprob_S(x) ~=~  \prod_{i\in S}\gcdf(x_i)\enspace,
    \qquad
    Q_S ~\coloneqq~ 1-\clauseprob_S\enspace,
    \qquad
    \clausemass ~\coloneqq~ \sum_{S\in\S}\clauseprob_S\enspace.
\]
Thus, $\mollifierkernel(x)=\prod_{S\in\S}Q_S$ and $0<\clauseprob_S<1$ for every $S\in\S$.

Fix an ordered derivative index
$J=(j_1,\dots,j_d)\in[m]^d$.
By the product rule for multivariate derivatives,
$\partial_J\mollifierkernel(x) \coloneqq \partial_{j_1}\cdots\partial_{j_d}\mollifierkernel(x)$
can be expressed as
\[
    \partial_J\mollifierkernel(x) ~= 
    \sum_{ h : [d] \to \S }
    \prod_{S\in\S} \partial_{J_{h^{-1}(S)}} Q_S\enspace,
\]
where
\begin{itemize}[topsep=5pt,itemsep=3pt]
    \item $h$ ranges over all functions from $[d]$ to $\S$,
            assigning each derivative index in $J$ to a factor $Q_S$ on which it acts,
    \item $h^{-1}(S) \coloneqq \{t \in [d] : h(t) = S\}$ denotes the set of derivative indices that hit $Q_S$,
    \item $J_{h^{-1}(S)}$ denotes the subtuple of $J$ consisting of the indices $j_t$ with $t\in h^{-1}(S)$, and
    \item $\partial_{J_{h^{-1}(S)}} Q_S$ denotes the derivative of $Q_S$ with respect to the indices in $J_{h^{-1}(S)}$.
\end{itemize}
Let $\S_h \coloneqq \{S\in\S : h^{-1}(S)\neq \emptyset\}$
denote the set of factors in $\mollifierkernel(x)$ that receive at least one
derivative under the assignment $h$.
Then, we have
\[
    \partial_J\mollifierkernel(x) ~= 
    \sum_{ h : [d] \to \S }
    \prod_{S\in\S_h} \partial_{J_{h^{-1}(S)}} Q_S \cdot
    \prod_{R\notin\S_h}Q_R
    \enspace.
\]
Summing over all $J\in[m]^d$, we obtain
\begin{align*}
    \norm{\mollifierkernel^{(d)}(x)}_1
    ~=~
    \sum_{j_1,\dots,j_d=1}^m
    \left|
    \partial_{j_1\cdots j_d}\mollifierkernel(x)
    \right|
    ~&\le~
    \sum_{j_1,\dots,j_d=1}^m
    \sum_{ h : [d] \to \S }
    \prod_{S\in\S_h} \abs{\partial_{J_{h^{-1}(S)}} Q_S} \cdot
    \prod_{R\notin\S_h}Q_R \nonumber\\
    &=~
    \sum_{ h : [d] \to \S }
    \left(
    \sum_{j_1,\dots,j_d=1}^m
    \prod_{S\in\S_h} \abs{\partial_{J_{h^{-1}(S)}} Q_S} \cdot
    \prod_{R\notin\S_h}Q_R
    \right) \nonumber\\
    &=~
    \sum_{ h : [d] \to \S }
    \left(\prod_{S\in\S_h} \norm{ \clauseprob_S^{(|h^{-1}(S)|)}(x) }_1 \right)\cdot
    \prod_{R\notin\S_h}Q_R\enspace.
\end{align*}
Moreover, since $0<\clauseprob_R<1$ for every $R\in\S$ and $\abs{\S_h}\leq d$,
we have
\[
    \prod_{R\notin \S_h}Q_R
    ~=~
    \prod_{R\notin \S_h}(1-\clauseprob_R)
    ~\le~
    \exp\!\left(-\sum_{R\notin \S_h}\clauseprob_R\right)
    ~=~
    \exp\!\left(-\clausemass+\sum_{R\in \S_h}\clauseprob_R\right)
    ~\le~
    \e^{d-\clausemass}\enspace.
\]
Therefore, we have
\begin{align*}
    \norm{\mollifierkernel^{(d)}(x)}_1
    ~\le~
    \e^{d-\clausemass} \cdot
    \sum_{ h : [d] \to \S }
    \left(\prod_{S\in\S_h} \norm{ \clauseprob_S^{(|h^{-1}(S)|)}(x) }_1 \right) \enspace.
\end{align*}
For each function $h:[d]\to\S$, the nonempty sets $\{h^{-1}(S)\}_{S\in \S_h}$
form a partition of $[d]$.
Conversely, such a function is obtained by choosing a partition of $[d]$
and then assigning a distinct $S\in\S$ to each block.
Let $\mathcal{P}_d$ denote the set of partitions of $[d]$.
For each partition $P\in\mathcal P_d$, fix an arbitrary ordering of its blocks
and write $P=\{B_1,\dots,B_l\}$.
Then
\begin{align*}
    \sum_{ h : [d] \to \S }
    \left(\prod_{S\in\S_h} \norm{ \clauseprob_S^{(|h^{-1}(S)|)}(x) }_1 \right)
    ~&=~
    \sum_{ \substack{P\in\mathcal{P}_d : P = \{B_1,\dots,B_l\} \\
    B_1,\dots,B_l\neq\emptyset \\ B_1\sqcup \cdots \sqcup B_l = [d] } }
    \sum_{\substack{S_1,\dots,S_l\in\S\\ \text{are distinct}}}
    ~\prod_{t=1}^l \norm{ \clauseprob_{S_t}^{(|B_t|)}(x) }_1 \\
    ~&\leq~
    \sum_{ \substack{P\in\mathcal{P}_d : P = \{B_1,\dots,B_l\} \\
    B_1,\dots,B_l\neq\emptyset \\ B_1\sqcup \cdots \sqcup B_l = [d] } }
    ~\prod_{t=1}^l \left( \sum_{S\in\S} \norm{ \clauseprob_S^{(|B_t|)}(x) }_1 \right)\enspace.
\end{align*}
Thus, we have
\[
    \norm{\mollifierkernel^{(d)}(x)}_1
    ~\le~
    \e^{d-\clausemass}\cdot
    \sum_{ \substack{P\in\mathcal{P}_d : P = \{B_1,\dots,B_l\} \\
    B_1,\dots,B_l\neq\emptyset \\ B_1\sqcup \cdots \sqcup B_l = [d] } }
    ~\prod_{t=1}^l \left( \sum_{S\in\S} \norm{ \clauseprob_S^{(|B_t|)}(x) }_1 \right)
    \enspace.
\]
Now set $T = \frac{1}{k} \cdot \log\frac{N}{\clausemass}$ where $N = \binom{m}{k}$.
Applying \cref{lem:moment} to each block $B_t$, we obtain
\begin{align*}
    \norm{\mollifierkernel^{(d)}(x)}_1
    ~\leq~
    \e^{d-\clausemass}\cdot
    \sum_{ \substack{P\in\mathcal{P}_d : P = \{B_1,\dots,B_l\} \\
    B_1,\dots,B_l\neq\emptyset \\ B_1\sqcup \cdots \sqcup B_l = [d] } }
    ~\prod_{t=1}^l~
    C_{3,|B_t|} \cdot k^{|B_t|}\cdot \clausemass \cdot (1+T^{|B_t|/2}) \enspace.
\end{align*}
Since each $\abs{B_t} \leq d$ and $l\le d$,
there exists a constant $C_{5,d}>0$ depending only on $d$ such that
\begin{align}
    \norm{\mollifierkernel^{(d)}(x)}_1
    &~\leq~
    C_{5,d} \cdot k^d \cdot 
    \sum_{ \substack{P\in\mathcal{P}_d : P = \{B_1,\dots,B_l\} \\
    B_1,\dots,B_l\neq\emptyset \\ B_1\sqcup \cdots \sqcup B_l = [d] } }
    \e^{-\clausemass} \cdot \clausemass^{l}\cdot
    \prod_{t=1}^l ~(1+T^{|B_t|/2}) \nonumber\\
    &~\leq~
    C_{5,d} \cdot 2^d \cdot k^d \cdot 
    \sum_{ \substack{P\in\mathcal{P}_d : P = \{B_1,\dots,B_l\} \\
    B_1,\dots,B_l\neq\emptyset \\ B_1\sqcup \cdots \sqcup B_l = [d] } }
    \e^{-\clausemass} \cdot \clausemass^{l}\cdot
    (1+T^{d/2}) \enspace, \label{eq:final-bound-on-derivative}
\end{align}
where the last step uses
$\prod_{t=1}^l(1+T^{|B_t|/2})\le 2^d(1+T^{d/2})$ for every
$P=\{B_1,\dots,B_l\}\in\mathcal P_d$.
Since the number of partitions of $[d]$ depends only on $d$,
it suffices to bound
$\e^{-\clausemass} \cdot \clausemass^{l}\cdot (1+T^{d/2})$
for all $1\le l\le d$.
We split into two cases according to whether $\clausemass\ge1$ or $0<\clausemass<1$.
\begin{itemize}
\item If $\clausemass\ge1$, then
$T = \frac{1}{k} \cdot \log\frac{N}{\clausemass} \leq \frac{1}{k} \cdot \log N = \frac{1}{k} \cdot \log \binom{m}{k} \leq \log \frac{\e m}{k}$,
and hence we have
\begin{align*}
    \e^{-\clausemass} \cdot \clausemass^{l} \cdot (1+T^{d/2})
    ~\le~
    \left(\sup_{u\ge1}~ \e^{-u}\cdot u^l\right) \cdot \br{1+\br{\log \frac{\e m}{k}}^{d/2}}\enspace.
\end{align*}
Since $\e^{-u}\cdot u^l$ is bounded on $[1,\infty)$,
for some constant $C_{6,l}>0$ depending only on $l$,
\[
    \e^{-\clausemass} \cdot \clausemass^{l} \cdot (1+T^{d/2})
    ~\le~
    C_{6,l}\cdot\br{1+\br{\log \frac{\e m}{k}}^{d/2}}\enspace.
\]

\item If $0<\clausemass<1$, then
$T = \frac{1}{k} \cdot \log\frac{N}{\clausemass} = \frac{1}{k} \cdot \log \binom{m}{k} + \frac{1}{k} \cdot \log \frac{1}{\clausemass} \leq \log \frac{\e m}{k} + \log \frac{1}{\clausemass}$, and hence
\begin{align*}
    \e^{-\clausemass} \cdot \clausemass^{l} \cdot (1+T^{d/2})
    ~&\le~
     \clausemass^l \cdot
    \br{1 + \br{\log\frac{\e m}{k} + \log\frac1{\clausemass}}^{d/2}}\\
    ~&\leq~
     \clausemass^l \cdot
    \br{1 + 2^{d/2}\cdot\br{\br{\log\frac{\e m}{k}}^{d/2} + \br{\log\frac1{\clausemass}}^{d/2}}} \\
    ~&\leq~
     2^{d/2} \cdot
    \br{1 +  \br{\log\frac{\e m}{k}}^{d/2} + \clausemass^l \cdot \br{\log\frac1{\clausemass}}^{d/2}} \enspace.
\end{align*}
Since $\clausemass^l \cdot \br{\log\frac1{\clausemass}}^{d/2}$ is bounded on $(0,1)$ for every $l\in[d]$, there exists a constant $C_{7,l,d}>0$ depending only on $l,d$ such that
\[
    \e^{-\clausemass} \cdot \clausemass^{l} \cdot (1+T^{d/2})
    ~\le~
    C_{7,l,d}\cdot\br{1+\br{\log\frac{\e m}{k}}^{d/2}}\enspace.
\]
\end{itemize}
Since $l\in[d]$, we can take $C_{8,d} = \max_{l\in[d]}\max\{C_{6,l}, C_{7,l,d}\}$
and obtain, for every $1\le l\le d$,
\[
    \e^{-\clausemass} \cdot \clausemass^{l} \cdot (1+T^{d/2})
    ~\le~
    C_{8,d}\cdot\br{1+\br{\log\frac{\e m}{k}}^{d/2}}\enspace.
\]
Substituting this into \cref{eq:final-bound-on-derivative} gives
\[
    \norm{\mollifierkernel^{(d)}(x)}_1 ~\leq~
    C_{5,d} \cdot C_{8,d} \cdot 2^d \cdot 
    \abs{\mathcal{P}_d}
    \cdot k^d \cdot \br{1+\br{\log\frac{\e m}{k}}^{d/2}}\enspace.
\]
Since $m\geq2k$, we have
\[
	\log\frac{\e m}{k} ~\leq~ \br{1+\frac{1}{\log2}} \cdot \log\frac{m}{k} \enspace.
\]
Letting $C_d = C_{5,d} \cdot C_{8,d} \cdot 2^d \cdot \br{1+\frac{1}{\log2}}^{d/2} \cdot \abs{\mathcal{P}_d}$ completes the proof.
\end{proof}

\newpage
\appendix

\phantomsection
\addcontentsline{toc}{section}{Appendix}
\section*{Appendix}

\section{\texorpdfstring{Proof of \cref{lem:prob-estimate}}{Proof of Lemma 3.10}}
\label{app:prob-estimate}

\simpleest*

\begin{proof}
Since $|[m]\setminus T|=k-1$, every $S\in \binom{[m]}{k}$ intersects $T$.
Hence, by the union bound,
\[
    1-\prod_{S\in\binom{[m]}{k}}\left(1-\prod_{i\in S}p_i\right)
    ~\le~
    \sum_{S\in\binom{[m]}{k}}\prod_{i\in S}p_i
    ~\le~
    \sum_{r=1}^k
    \binom{m-k+1}{r}\binom{k-1}{k-r}\br{\frac{\delta}{4m^2}}^r \enspace.
\]
Now, let 
\[
    t_r~\coloneqq~\binom{m-k+1}{r}\binom{k-1}{k-r}\br{\frac{\delta}{4m^2}}^r \enspace.
\]
Then, we have that $t_1 = (m-k+1)\cdot{\frac{\delta}{4m^2}} \leq \frac{\delta}{4m}$, and 
for $1\le r<k$,
\[
    \frac{t_{r+1}}{t_r}
    ~=~
    \frac{(m-k+1-r)(k-r)}{r(r+1)} \cdot{\frac{\delta}{4m^2}}
    ~\le~
    m^2\cdot{\frac{\delta}{4m^2}}
    ~\le~
    \frac14\enspace.
\]
Thus, we have that $t_{r+1} \le t_r/4$ for all $1\le r<k$, and hence
\[
    \sum_{r=1}^k t_r
    ~\le~
    \sum_{r=1}^k t_1 \cdot 4^{-(r-1)}
    ~\le~
    \frac43 \cdot t_1
    ~\le~
    \frac43\cdot \frac{\delta}{4m}
    ~\le~
    \delta \enspace.
\]
This completes the proof.
\end{proof}

\section[Noise Sensitivity and Gaussian Surface Area of Thresholds of Halfspaces]{Noise Sensitivity and Gaussian Surface Area of Thresholds of\\Halfspaces}
\label{app:threshold-ns-gsa}

This section proves the upper and lower bounds
on the noise sensitivity and Gaussian surface area of thresholds of halfspaces
stated in \cref{sec:anti}.
We first prove \cref{thm:threshold-ns-gsa},
which gives upper bounds on the noise sensitivity and
Gaussian surface area of thresholds of halfspaces.
We restate the theorem here for convenience.

\nsgsa*

\begin{proof}
    As discussed in \cref{rmk:equivalent-representation},
    we can write $F$ as $F(x) = \ort_{k,b}(Ax)$
    for some $A\in\R^{m\times n}$ and $b\in\R^m$.
    Let \(R\subseteq[m]\) be a random subset 
    obtained by picking each index independently with probability \(p=1/k\),
    and define
    \[
        F_R(x) ~\coloneqq~ \ind[\exists i\in R:A_i x\le b_i] \enspace.
    \]
    That is, \(F_R\) is the \(\OR\) of the halfspaces \(\{A_i x\le b_i\}_{i\in R}\).

    We first prove the bound for noise sensitivity.
    Let \(x\sim\pmcube n\) and \(y \sim N_\delta(x)\).
    For \(z\in\pmcube n\), let $S(z)\coloneqq \{i\in[m]:A_i z\le b_i\}$
    denote the set of satisfied halfspaces at \(z\).
    On the event \(F(x)=1\) and \(F(y)=0\), we have \(|S(x)|\ge k\) and
    \(|S(y)|\leq k-1\). Hence
    \[
        \Pr_R[F_R(x)=1,\ F_R(y)=0]
        ~=~
        (1-p)^{|S(y)|} \cdot
        \left(1-(1-p)^{|S(x)\setminus S(y)|}\right)
        ~\ge~
        p(1-p)^{k-1}.
    \]
    Similarly, if \(F(x)=0\) and \(F(y)=1\), then
    $\Pr_R[F_R(x)=0,\ F_R(y)=1]~\ge~p(1-p)^{k-1}$.
    Therefore, 
    \[
    	p(1-p)^{k-1} \cdot \ind[F(x)\ne F(y)]
        ~\le~
        \Pr_{R} [ F_R(x) \neq  F_R(y)] \enspace.
    \]
    Taking expectation over \(x\sim\pmcube n\) and \(y\sim N_\delta(x)\) gives
    \[
        p(1-p)^{k-1} \cdot \ns_\delta(F)
        ~\le~
        \E_{R}[\ns_\delta(F_R)] \enspace.
    \]
    Since \(F_R\) is an \(\OR\) of halfspaces,
    \(1-F_R\) is an intersection of the complementary halfspaces.
    Therefore, applying \cref{thm:intersection-halfspaces-ns-gsa} to \(1-F_R\),
    we have $\ns_\delta(F_R)=\ns_\delta(1-F_R)\le C\sqrt{\delta\log(|R|)}$.
    Thus, by Jensen's inequality and \(\E_R[|R|]=m/k\), we have
    \[
        \ns_\delta(F)
        ~\le~
        \frac{C\cdot \E_R[\sqrt{\delta\log(|R|)}]}{p(1-p)^{k-1}}
        ~\le~
        \frac{C\sqrt{\delta\cdot \log(\E_R[|R|])}}{p(1-p)^{k-1}}
        ~=~
        O\!\br{k \sqrt{\delta\cdot \log(m/k)}} \enspace.
    \]

    We next prove the bound for Gaussian surface area.
    Without loss of generality, assume that the boundary hyperplanes
    \(\{H_i\}_{i=1}^m\), where \(H_i \coloneqq \{x\in\R^n:A_i x=b_i\}\),
    satisfy that every pairwise intersection \(H_i\cap H_j\), \(i\ne j\),
    has \((n-1)\)-dimensional surface measure zero.
    The general case follows by an arbitrarily small perturbation of the halfspaces
    and then letting the perturbation tend to zero.
    
    For each $i\in[m]$ and \(x\in H_i\), let
    $ s_i(x)\coloneqq \#\{j\in[m]\setminus\{i\}: A_jx\le b_j\} $
    denote the number of satisfied halfspaces of $x$
    other than the $i$-th one.
    Except on intersections of two hyperplanes,
    crossing \(H_i\) changes only the \(i\)-th halfspace indicator.
    Hence the part of the boundary of \(K(F)\) lying on \(H_i\)
    consists exactly of those \(x\in H_i\) for which \(s_i(x)=k-1\).
    Therefore,
    \begin{equation}\label{eq:gsa-threshold}
        \gsa(F)
        ~=~
        \sum_{i=1}^m
        \int_{H_i}  \ind[s_i(x)=k-1] \cdot \phi_n(x)\,\mathrm{d}\sigma_i(x) \enspace,
    \end{equation}
    where \(\phi_n\) is the standard Gaussian density on \(\R^n\), and
    \(\mathrm{d}\sigma_i\) is surface measure on \(H_i\).
    For each \(i\in [m]\), \(x\in H_i\) and random subset \(R\) as above,
    define
    \[
        \EE_{R,i}(x)
        ~\coloneqq~
        \ind[i\in R]\cdot \ind\bigl[\{j\in R\setminus\{i\}:A_jx\le b_j\}=\emptyset\bigr] \enspace.
    \]
    Then, we have
    \begin{equation} \label{eq:gsa-threshold-R}
        \gsa(F_R)
        ~=~
        \sum_{i\in [m]}
        \int_{H_i} \EE_{R,i}(x)\cdot \phi_n(x)\,\mathrm{d}\sigma_i(x) \enspace.
    \end{equation}
    For fixed \(i\in[m]\) and \(x\in H_i\), we have
    \[
        \E_R[ \EE_{R,i}(x)] ~=~ p(1-p)^{s_i(x)} \enspace.
    \]
    In particular, this expectation equals \(p(1-p)^{k-1}\) for \(s_i(x)=k-1\).
    Consequently,
    \begin{equation} \label{eq:gsa-threshold-R-compare}
        \ind[s_i(x)=k-1]
        ~\le~
        \frac{1}{p(1-p)^{k-1}}\cdot \E_R[\EE_{R,i}(x)] \enspace.
    \end{equation}
    Combining \cref{eq:gsa-threshold,eq:gsa-threshold-R,eq:gsa-threshold-R-compare} gives
    \[
        \gsa(F)
        ~\le~
        \frac{1}{p(1-p)^{k-1}}\cdot\E_R[\gsa(F_R)] \enspace.
    \]
    Since \(F_R\) is an \(\OR\) of halfspaces, \(1-F_R\) is an intersection of
    the complementary halfspaces.  Gaussian surface area is unchanged under
    complementing the function, so by \cref{thm:intersection-halfspaces-ns-gsa},
    \[
        \gsa(F_R)~=~\gsa(1-F_R)~\le~ C\sqrt{\log(|R|)} \enspace.
    \]
    Therefore, by Jensen's inequality and \(\E_R[|R|]=mp=m/k\),
    \[
        \gsa(F)
        ~\le~
        \frac{C\cdot \E_R[\sqrt{\log(|R|)}]}{p(1-p)^{k-1}}
        ~\le~
        \frac{C\sqrt{\log(\E_R[|R|])}}{p(1-p)^{k-1}}
        ~=~
        O\!\br{k\sqrt{\log(m/k)}} \enspace. \qedhere
    \]
\end{proof}

Now we prove the lower bounds in \cref{prop:threshold-ns-lower-bound}
and \cref{prop:threshold-surface-lower-bound}.
We first prove the lower bound for Gaussian surface area,
which is restated below.

\gsalb*

\begin{proof}

Fix $n \geq m$. We will give two constructions of
\(k\)-out-of-\(m\) thresholds of halfspaces $F_1, F_2:\R^n\to\{0,1\}$ 
such that
\[
	\gsa(F_1) = \Omega(k) \qquad\text{and}\qquad
    \gsa(F_2) = \Omega\!\br{\sqrt{k\log(m/k)}} \enspace.
\]
The desired lower bound then follows since
\[
    \max\left\{k,
    \sqrt{k\log\frac{m}{k}}\right\}
    ~\ge~
    \frac12
    \left(k+\sqrt{k\log\frac{m}{k}}\right) \enspace.
\]

\vspace{-1em}
\paragraph{Constructing \(F_1\).}
Choose $2k-1$ distinct thresholds $t_1<t_2<\cdots<t_{2k-1}$
in a fixed interval around the origin.
For \(j=1,\ldots,2k-1\), define the halfspace $h_j:\R^n\to\{0,1\}$ by
\[
    h_j(x)
    =
    \begin{cases}
        ~\ind[x_1\ge t_j]\,, & j \text{ is odd} \enspace,\\[0.2em]
        ~\ind[x_1\le t_j]\,, & j \text{ is even}\enspace.
    \end{cases}
\]
For \(j=2k,\ldots,m\), set $h_j(x)=\ind[\langle \mathbf{0},x\rangle\le -1]$,
which always outputs 0.
Finally, define
\[
    F_1(x)
    ~=~
    \thr_{m,k}(h_1(x),\ldots,h_m(x)) \enspace.
\]

We claim that \(\gsa(F_1)=\Omega(k)\).
Note that the function \(F_1\) depends only on the first coordinate.
For \(x_1<t_1\), exactly \(k-1\) of the halfspaces \(h_1,\ldots,h_{2k-1}\) are satisfied.
As \(x_1\) crosses each threshold \(t_j\),
exactly one of these halfspace values changes,
and hence the number of satisfied halfspaces alternates between \(k-1\) and \(k\).
Therefore each hyperplane $\{x\in\R^n:x_1=t_j\}$
lies on the boundary of \(F_1^{-1}(1)\).
Hence, we have
\[
    \gsa(F_1)
    ~\ge~
    \sum_{j=1}^{2k-1}\gpdf(t_j)
    ~=~
    \Omega(k)\enspace,
\]
where \(\gpdf\) is the PDF of the standard Gaussian distribution
and each \(\gpdf(t_j)\) is a constant.

\vspace{-1em}
\paragraph{Constructing \(F_2\).}
Let \(p=k/m\), and choose \(t\in\R\) such that
$\Pr_{g\sim \NN(0,1)}[g\ge t]=p$.
For \(j=1,\ldots,m\), define the halfspace \(h_j:\R^n\to\{0,1\}\) by
$h_j(x) = \ind[x_j\ge t]$.
Finally, define $F_2(x) = \thr_{m,k}(h_1(x),\ldots,h_m(x))$.

We now prove that \(\gsa(F_2)=\Omega(\sqrt{k\log(m/k)})\).
For each \(j\in[m]\), the hyperplane \(\{x\in\R^n:x_j=t\}\)
contributes to the boundary of \(F_2^{-1}(1)\)
exactly when precisely \(k-1\) of the other \(m-1\) halfspaces are satisfied.
Hence,
\[
    \gsa(F_2)
    ~=~
    m \cdot \gpdf(t)\cdot
    \binom{m-1}{k-1} \cdot p^{k-1}(1-p)^{m-k} \enspace.
\]
We will use the following bound on $\gpdf(t)$.

\begin{claim}\label[claim]{clm:gaussian-quantile-density}
There is a constant \(C>0\) such that for every \(0<p\le1/2\), if
\(t\ge0\) is chosen so that \(\Pr_{g\sim \NN(0,1)}[g\ge t]=p\), then
\(\gpdf(t)\ge C\cdot p\sqrt{\log(1/p)}\).
\end{claim}

\begin{proof}
If \(0\le t\le1\), then \(\gpdf(t)\ge \gpdf(1)\).
Since \(p\in(0,1/2]\), \(p\sqrt{\log(1/p)}\) is bounded above by a universal constant.
Hence, by choosing the constant \(C>0\) sufficiently small,
we have \(C\cdot p\sqrt{\log(1/p)}\le \gpdf(1)\le \gpdf(t)\).

For \(t>1\), we have
$\Pr_{g\sim \NN(0,1)}[g\ge t]
    =
    \int_t^\infty \gpdf(s)\,\d s
    \le
    \frac1t\int_t^\infty s\cdot \gpdf(s)\,\d s
    =
    \frac{\gpdf(t)}{t}$.
This gives \(\gpdf(t)\ge pt\).
On the other hand, by \cref{lem:mills-ratio},
$p=\gcdf(-t) \ge \frac{t}{1+t^2} \cdot \gpdf(t)$.
Since \(t>1\), we have \( \frac{t}{1+t^2} \geq \frac{1}{2t} \).
Hence, $ p \geq  \frac{1}{2\sqrt{2\pi}t}\cdot \e^{-t^2/2}$.
Taking logarithms gives
$\log(1/p) \le\frac{t^2}{2}+\log(2\sqrt{2\pi}t)$.
Thus for sufficiently small $C>0$,
$ C\sqrt{\log(1/p)} \leq t $.
Combining this with
\(\gpdf(t)\ge pt\), we get
$\gpdf(t)\ge C\cdot p\sqrt{\log(1/p)}$.
\end{proof}
By \cref{clm:gaussian-quantile-density},
we have
$ \gpdf(t) = \Omega\bigl((k/m)\sqrt{\log(m/k)}\bigr) $.
Therefore, it suffices to show that
\begin{equation}\label{eq:binomial-lower-bound}
    \binom{m-1}{k-1} \cdot p^{k-1}(1-p)^{m-k} ~=~ \Omega\!\br{\frac{1}{\sqrt{k}}} \enspace.
\end{equation}
The case \(k=1\) is immediate, so assume \(k\ge2\).
Since \(p=k/m\), by Stirling's formula, we have
\[
\begin{aligned}
\binom{m-1}{k-1}\cdot p^{k-1}(1-p)^{m-k}
~=~
\Theta\br{\sqrt{\frac{m-1}{(k-1)(m-k)}} \cdot
\left(\frac{m-1}{m}\right)^{m-1} \cdot 
\left(\frac{k}{k-1}\right)^{k-1}} \enspace.
\end{aligned}
\]
The required lower bound in \cref{eq:binomial-lower-bound}
then follows from the fact that
the last two factors are bounded below by absolute constants,
and
\[
    \sqrt{\frac{m-1}{(k-1)(m-k)}} ~\ge~ \sqrt{\frac{m-1}{k(m-1)}} ~=~ \frac{1}{\sqrt{k}} \enspace. \qedhere
\]
\end{proof}

We then prove the lower bound for noise sensitivity in \cref{prop:threshold-ns-lower-bound}.

\nslb*

\begin{proof}

We will give two $k$-out-of-$m$ thresholds of halfspaces $F_1$ and $F_2$ such that
\[
    \ns_\delta(F_1)=\Omega(k\sqrt{\delta})
    \qquad\text{and}\qquad
    \ns_\delta(F_2)
    =
    \Omega\!\left(\sqrt{\delta\cdot k\log\frac{m}{k}}\right).
\]
The desired lower bound then follows since
\[
    \max\left\{k\sqrt{\delta},
    \sqrt{k\delta\log\frac{m}{k}}\right\}
    \ge
    \frac12\sqrt{\delta}
    \left(k+\sqrt{k\log\frac{m}{k}}\right).
\]
The idea is to approximate Gaussian variables
by normalized sums of independent Boolean variables
and use the constructions from the proof of \cref{prop:threshold-surface-lower-bound}.

Let \(q\) be a sufficiently large integer, and let $n_0 = mq$.
For \(x\sim\pmcube{n_0}\), we define random variables $Z^x \in \R^m$ by
\[
	Z^x_i = \frac{\sum_{j=(i-1)q+1}^{iq} x_j}{\sqrt{q}}\,, \qquad \forall i\in [m]\enspace.
\]
Let \(y\sim N_\delta(x)\), and define \(Z^y\) analogously.
As \(q\to\infty\),
the joint distribution of \((Z^x,Z^y)\) converges to
that of \((G,G')\), where \(G=(G_1,\ldots,G_m)\) and \(G'=(G'_1,\ldots,G'_m)\)
are standard Gaussian vectors
and each pair \((G_i,G'_i)\) has correlation \(1-2\delta\).
We choose \(q\) sufficiently large
so that the Boolean sums approximate the Gaussian variables well enough
for all probability estimates used below.
The construction uses only the first \(n_0=mq\) coordinates.
If \(n>n_0\), we simply let the function be independent of
the remaining coordinates.

\vspace{-1em}
\paragraph{Constructing \(F_1\).}
Choose \(2k-1\) thresholds \(t_1<t_2<\cdots<t_{2k-1}\)
in a fixed interval around the origin
so that adjacent thresholds are separated by
\(\Theta(\sqrt{\delta})\).
We can do this because \(\delta\le c_0/(k^2\log(m/k))\)
and \(c_0\) is a sufficiently small constant.
For \(j=1,\ldots,2k-1\), define
\[
    h_j(x)
    =
    \begin{cases}
        \ind[Z^x_1\ge t_j]\ , & j \text{ is odd}\enspace,\\[0.2em]
        \ind[Z^x_1\le t_j]\ , & j \text{ is even}\enspace.
    \end{cases}
\]
For \(j=2k,\ldots,m\), set
\(h_j(x)=\ind[\langle \mathbf{0},x\rangle\le -1]\), which always outputs \(0\).
Finally, define
\[
    F_1(x)=\thr_{m,k}(h_1(x),\ldots,h_m(x)) \enspace.
\]

We show that \(\ns_\delta(F_1)=\Omega(k\sqrt{\delta})\).
Let $x\sim\pmcube{n_0}$ and $y\sim N_\delta(x)$.
By our choice of \(q\),
we view \((Z^x_1,Z^y_1)\) as a pair of correlated Gaussian variables
with correlation \(1-2\delta\).
For \(j=1,\ldots,2k-1\),
let \(\EE_j\) be the event that
\(Z^x_1\) and \(Z^y_1\) lie on different sides of \(t_j\),
but on the same side of every other threshold.
Since the number of satisfied halfspaces alternates between \(k-1\) and \(k\)
across each threshold \(t_j\), we have \(F_1(x)\ne F_1(y)\) on \(\EE_j\).
We now lower bound \(\Pr[\EE_j]\).
Let \(\eta>0\) be a constant such that the event
\[
    t_j-\eta\sqrt{\delta}
    ~\le~ Z^x_1
    ~\le~ t_j-\frac{\eta}{2}\sqrt{\delta}
    \qquad\text{and}\qquad
    t_j+\frac{\eta}{2}\sqrt{\delta}
    ~\le~ Z^y_1
    ~\le~ t_j+\eta\sqrt{\delta}
\]
is contained in \(\EE_j\).
It remains to estimate the probability of this event.
Conditioned on \(Z^x_1=s\),
the random variable \(Z^y_1\) is a Gaussian variable
with mean \((1-2\delta)s\) and variance \(1-(1-2\delta)^2=\Theta(\delta)\).
For \(s\in[t_j-\eta\sqrt{\delta},\,t_j-\eta\sqrt{\delta}/2]\),
the interval \([t_j+\eta\sqrt{\delta}/2,\,t_j+\eta\sqrt{\delta}]\)
has length \(\Theta(\sqrt{\delta})\)
and lies within \(O(\sqrt{\delta})\) of this conditional mean.
Therefore the conditional probability that \(Z^y_1\) lies
in this interval is bounded below by an absolute constant.
Since all thresholds lie in a fixed interval around the origin,
the Gaussian measure of \(Z^x_1\) is bounded below by $\Omega(\sqrt{\delta})$
on \([t_j-\eta\sqrt{\delta},\,t_j-\eta\sqrt{\delta}/2]\).
Hence, we have
\[
    \Pr[\EE_j]
    ~\ge~
    \Omega(\sqrt{\delta})\cdot \Omega(1)
    ~=~
    \Omega(\sqrt{\delta}) .
\]
The events \(\EE_1,\ldots,\EE_{2k-1}\) are disjoint, and each implies
\(F_1(x)\ne F_1(y)\).  Thus, we obtain
\[
    \ns_\delta(F_1)
    ~\ge~
    \sum_{j=1}^{2k-1}\Pr[\EE_j]
    ~=~
    \Omega(k\sqrt{\delta}) \enspace.
\]

\vspace{-2em}
\paragraph{Constructing \(F_2\).}
Let \(p=k/m\leq 1/2\), and choose \(t\geq 0\) such that
\(\Pr_{g\sim\NN(0,1)}[g\ge t]=p\).
Then \(t=O(\sqrt{\log(m/k)})\) by the standard Gaussian tail bound.
For \(i=1,\ldots,m\), define $h_i(x)=\ind[Z^x_i\ge t]$.
Finally, define
$F_2(x)=\thr_{m,k}(h_1(x),\ldots,h_m(x))$.

We show that
\(\ns_\delta(F_2)=\Omega(\sqrt{k\delta\log(m/k)})\).
Let \(x\sim\pmcube{n_0}\) and \(y\sim N_\delta(x)\).
As above, we view \((Z^x_i,Z^y_i)\) as a pair of Gaussian variables
with correlation \(1-2\delta\).
For each \(i\in[m]\), let \(\RR_i\) be the event that
\(Z^x_i<t\) and \(Z^y_i\ge t\), exactly \(k-1\) of the remaining pairs satisfy
\(Z^x_\ell\ge t\) and \(Z^y_\ell\ge t\), and all other remaining pairs satisfy
\(Z^x_\ell<t\) and \(Z^y_\ell<t\).
Hence \(F_2(x)\ne F_2(y)\) on \(\RR_i\).
Let $g$ and $g'$ be standard Gaussian variables with correlation \(1-2\delta\).
Set
\[
    a=\Pr[g<t,\ g'\ge t]\enspace,\qquad
    b=\Pr[g\ge t,\ g'\ge t]\enspace,\qquad
    r=\Pr[g<t,\ g'<t]\enspace.
\]
Then, we have $\Pr[\RR_i] = a \cdot \binom{m-1}{k-1}\cdot b^{k-1} \cdot r^{m-k}$.

We first lower bound \(a\).  Write \(\rho=1-2\delta\) and
\(\sqrt{1-\rho^2}=\Theta(\sqrt{\delta})\).  We may write
\(g'=\rho g+\sqrt{1-\rho^2}\cdot w\), where \(w\sim\NN(0,1)\) is independent of \(g\).
Let \(I=[t-\eta\sqrt{1-\rho^2},t-\eta\sqrt{1-\rho^2}/2]\), where \(\eta>0\) is a sufficiently
small absolute constant.  For every \(s\in I\),
\begin{equation} \label{eq:gaussian-conditional-probability-1}
    \Pr[g'\ge t\mid g=s]
    ~=~
    \Pr\!\left[w\ge \frac{t-\rho s}{\sqrt{1-\rho^2}}\right]
    ~=~
    \Omega(1) \enspace,
\end{equation}
since
\[
	\frac{t-\rho s}{\sqrt{1-\rho^2}} ~\leq~
    \frac{t-\rho(t - \eta\sqrt{1-\rho^2})}{\sqrt{1-\rho^2}}
    ~=~
    \frac{(1-\rho)t}{\sqrt{1-\rho^2}} + {\eta\rho}
    ~=~
    O(t\sqrt{\delta})+O(1)
    ~=~
    O(1) \enspace,
\]
where last step uses \(t=O(\sqrt{\log(m/k)})\) and \(\delta\le c_0/(k^2\log(m/k))\).
Moreover, for every \(s\in I\), we have
\[
    \frac{\gpdf(s)}{\gpdf(t)}
    ~=~
    \e^{{(t^2-s^2)}/{2}}
    ~=~
    \e^{ t(t-s)-{(t-s)^2}/{2} }
    ~\geq~
    \e^{ -{(t-s)^2}/{2} }
    ~\geq~
    \e^{ -{\eta^2(1-\rho^2)}/{2} }
    ~=~
    \e^{-\Theta(\delta)}
    ~=~
    \Omega(1) \enspace.
\]
Therefore, we have
\begin{equation}\label{eq:gaussian-conditional-probability-2}
    \Pr[g\in I]
    ~=~
    \int_I \gpdf(s) \,\d s
    ~=~
    \Omega\!\br{\abs{I}\cdot \gpdf(t)}
    ~=~
    \Omega\!\br{\sqrt{\delta}\cdot \gpdf(t)}
    \enspace.
\end{equation}
Combining \cref{eq:gaussian-conditional-probability-1,eq:gaussian-conditional-probability-2}, we get
\[
	a
    ~=~
    \Pr[g<t,\ g'\ge t]
    ~\ge~
    \Pr[ g\in I, g'\ge t]
    ~=~
    \Omega\!\br{\sqrt{\delta}\cdot \gpdf(t)} \enspace.
\]
By \cref{clm:gaussian-quantile-density}, we get
\begin{equation}\label{eq:lb-a}
    a ~=~ \Omega\!\left(\frac{k}{m}\sqrt{\delta\cdot \log\frac{m}{k}}\right) \enspace.
\end{equation}

We also need an upper bound on \(a\).
Note that by \cref{eq:gaussian-conditional-probability-1}, we have
\[
    a
    ~=~
    \int_{-\infty}^t
    \gpdf(s)\cdot
    \Pr_{w\sim\NN(0,1)}\!\left[
        w\ge \frac{t-\rho s}{\sqrt{1-\rho^2}}
    \right]\,\d s .
\]
Since $t-\rho s=(1-\rho)t+\rho(t-s)\ge \rho(t-s)$ and \(\rho\ge 1/2\) for sufficiently small \(c_0\),
the standard tail bound gives
\[
    \Pr\!\left[
        w\ge \frac{t-\rho s}{\sqrt{1-\rho^2}}
    \right]
    ~\le~
    \Pr\!\left[
        w\ge \frac{\rho(t-s)}{\sqrt{1-\rho^2}}
    \right]
    ~\le~
    \exp\!\left(-\frac{(t-s)^2}{8(1-\rho^2)}\right) \enspace.
\]
Also, $\gpdf(s)
    =
    \gpdf(t) \cdot \e^{ t(t-s)-{(t-s)^2}/{2} }
    \le
    \gpdf(t) \cdot \e^{t(t-s)}$.
Thus,
we obtain
\begin{align*}
    a
    ~&\le~
    \gpdf(t) \cdot
    \int_{-\infty}^t
    \exp\!\left(t(t-s)-\frac{(t-s)^2}{8(1-\rho^2)}\right)\d s \\
    ~&=~
    \gpdf(t) \cdot
    \sqrt{1-\rho^2} \cdot
    \int_0^\infty
    \exp\!\left(\sqrt{1-\rho^2}\cdot tu-\frac{u^2}{8}\right)\d u \\
    ~&=~
    O(\sqrt{\delta}\cdot \gpdf(t)) \enspace,
\end{align*}
where the first equality uses the change of variable \(u=(t-s)/\sqrt{1-\rho^2}\),
and the last step uses
\(\sqrt{1-\rho^2}=\Theta(\sqrt{\delta})\),
\(t\sqrt{1-\rho^2}=O(1)\), and the fact that
\(\int_0^\infty \exp(Cu-u^2/8)\,\d u=O(1)\) for every constant \(C\).
By \cref{lem:oned} with \(r=1\) and \(x=-t\), and using
\(\gcdf(-t)=p\), we have
\[
    \gpdf(t)
    ~=~
    \gpdf(-t)
    ~=~
    \gcdf'(-t)
    ~\le~
    C_{1,1}\cdot \gcdf(-t)
    \left(1+\log\frac1{\gcdf(-t)}\right)^{1/2}
    =~
    O\!\left(p\sqrt{\log\frac1p}\right) \enspace,
\]
where the last step uses \(p\le 1/2\).
Therefore, we have
\begin{equation}\label{eq:ub-a}
    a ~=~ O\!\left(\frac{k}{m}\sqrt{\delta\cdot \log\frac{m}{k}}\right) \enspace.
\end{equation}

Now we lower bound \(b\) and \(r\). Note that
\[
    b=\Pr[g'\ge t]-a=p-a
    \qquad\text{and}\qquad
    r=\Pr[g<t]-a=1-p-a \enspace.
\]
Moreover, by \cref{eq:ub-a}, we have \(a\le p/(2k)\), since \(\delta\le c_0/(k^2\log(m/k))\) and \(c_0\) is sufficiently small.
Thus, we have
\begin{equation} \label{eq:lb-b}
    b^{k-1}
    ~=~
    (p-a)^{k-1}
    ~\ge~
    p^{k-1}\left(1-\frac1{2k}\right)^{k-1}
    =~
    \Omega(p^{k-1}) \enspace.
\end{equation}
On the other hand, we have
\begin{equation}\label{eq:lb-r}
    r^{m-k}
    ~=~
    (1-p-a)^{m-k}
    ~=~
    (1-p)^{m-k}
    \left(1-\frac{a}{1-p}\right)^{m-k}
    =~
    \Omega((1-p)^{m-k}) \enspace,
\end{equation}
where the last step follows from
$
    \frac{(m-k)a}{1-p} = ma = O(k\sqrt{\delta\log(m/k)}) = O(\sqrt{c_0}) \ll 1
$.
Therefore, combining \cref{eq:lb-a,eq:lb-b,eq:lb-r}, we get
\begin{align*}
    \Pr[\RR_i]
    ~=~
    a\cdot \binom{m-1}{k-1} \cdot b^{k-1}r^{m-k}
    =
    \Omega\!\left(
        \frac{k}{m}\sqrt{\delta\log\frac{m}{k}}\cdot
        \binom{m-1}{k-1} \cdot p^{k-1}(1-p)^{m-k}
    \right) \enspace.
\end{align*}
Using \cref{eq:binomial-lower-bound}, we obtain
\[
    \Pr[\RR_i]
    ~=~
    \Omega\!\left(
        \frac{\sqrt{k}}{m}\sqrt{\delta\log\frac{m}{k}}
    \right) \enspace.
\]
The events \(\RR_1,\ldots,\RR_m\) are disjoint, and each implies
\(F_2(x)\ne F_2(y)\).  Thus, we have
\[
    \ns_\delta(F_2)
    ~\ge~
    \sum_{i=1}^m \Pr[\RR_i]
    ~=~
    \Omega\!\left(\sqrt{\delta\cdot k\log\frac{m}{k}}\right)\enspace. \qedhere
\]
\end{proof}

\section{\texorpdfstring{Proof of \cref{thm:boolean-threshold-band-lower}}{Proof of Theorem 4.7}}
\label{app:boolean-threshold-band-lower}

\booleanlb*

\begin{proof}
Let \(q\) be an integer, and let \(n_0=mq\).
For \(u\sim\pmcube{n_0}\), define \(Z^u\in\R^m\) by
\[
    Z^u_i
    =
    \frac{\sum_{j=(i-1)q+1}^{iq} u_j}{\sqrt q}\,,
    \qquad \forall i\in [m] \enspace.
\]
As in the proof of \cref{prop:threshold-ns-lower-bound},
we take \(q\) sufficiently large
and treat \(Z^u=(Z^u_1,\ldots,Z^u_m)\) as an \(m\)-dimensional standard Gaussian vector
for all probability estimates below.

We will give two constructions of row-normalized matrices $A\in\R^{m\times n}$
and thresholds $b\in\R^m$ for all $n\geq n_0$.
The first has anticoncentration probability \(\Omega(\Lambda k)\), and the second has
anticoncentration probability \(\Omega(\Lambda\sqrt{k\log(m/k)})\).
The desired lower bound then follows by taking the maximum.

\vspace{-1em}
\paragraph{Construction 1.}
Choose \(2k-1\) thresholds
$-1<t_1<t_2<\cdots<t_{2k-1}<1$
such that adjacent thresholds are separated by at least \(2\Lambda\).
This is possible because \(\Lambda\le c_0/k\) for sufficiently small \(c_0>0\).
For \(i=1,\ldots,2k-1\), define
\[
    A_{i,j}
    =
    \begin{cases}
        -q^{-1/2}\,, & \text{if } i \text{ is odd and } 1\le j\le q \enspace,\\
        ~q^{-1/2}\,, & \text{if } i \text{ is even and } 1\le j\le q \enspace,\\
        ~0\,, & \text{otherwise}\enspace.
    \end{cases}
\]
For \(i=2k,\ldots,m\), let
\[
    A_{i,j}
    =
    \begin{cases}
        q^{-1/2}\,, & \text{if } 1\le j\le q \enspace,\\
        0\,, & \text{otherwise}\enspace.
    \end{cases}
\]
Set \(b_i=-t_i\) for odd \(i\le 2k-1\), \(b_i=t_i\) for even
\(i\le 2k-1\), and \(b_i=-q\) for \(i=2k,\ldots,m\).
Since \(Z^u_1\ge -\sqrt q\), the last \(m-2k+1\) rows are never satisfied
even after shifting the threshold by \(\Lambda\).
All rows have $2$-norm \(1\) and are \(q^{-1/2}\)-regular.
We can choose \(q\) sufficiently large so that
$q^{-1/2} \leq \Lambda$.

It is easy to check that if \(Z^u_1\in(t_i-\Lambda,t_i+\Lambda)\) for some \(i\le 2k-1\),
then $Au\in \ort_{k,b+\Lambda\cdot\onevec}$ but $Au \notin \ort_{k,b-\Lambda\cdot\onevec}$.
Indeed, the first \(2k-1\) halfspaces alternate between
the tests \(Z^u_1\ge t_j\) and \(Z^u_1\le t_j\).
When \(Z^u_1\in(t_i-\Lambda,t_i+\Lambda)\),
passing from \(b-\Lambda\cdot\onevec\) to \(b+\Lambda\cdot\onevec\) changes
only the \(i\)-th test from unsatisfied to satisfied,
increasing the number of satisfied rows from \(k-1\) to \(k\).
Hence, we have
\begin{align*}
    \Pr\!
    \left[
        Au\in \ort_{k,b+\Lambda\cdot\onevec}
        \setminus
        \ort_{k,b-\Lambda\cdot\onevec}
    \right]
    ~\ge~
    \sum_{j=1}^{2k-1}\Pr[|Z^u_1-t_j| < \Lambda]
    ~=~
    \Omega(\Lambda k) \enspace,
\end{align*}
where the last step uses the fact that the Gaussian density is bounded below by
an absolute constant on \([-1,1]\).

\vspace{-1em}
\paragraph{Construction 2.}
Let \(t\) be a threshold to be determined later.
For \(i=1,\ldots,m\), define
\[
    A_{i,j}
    =
    \begin{cases}
        -q^{-1/2}\,, & (i-1)q<j\le iq\enspace,\\
        0\,, & \text{otherwise}\enspace,
    \end{cases}
    \qquad
    b_i=-t \enspace.
\]
Thus the \(i\)-th halfspace is \(\ind[Z^u_i\ge t]\).
Again, all rows have $2$-norm \(1\) and are \(q^{-1/2}\)-regular.

For a lower bound, it suffices to consider the event that exactly \(k-1\)
of the variables \(Z^u_i\) are at least \(t+\Lambda\), and among the remaining
variables at least one lies in the interval \([t-\Lambda,t+\Lambda)\).
Let $g$ be a standard Gaussian variable, and let
\[
	a ~=~ \Pr[g\ge t+\Lambda]\enspace, \qquad
	r ~=~ \Pr[ g\leq t-\Lambda] \enspace.
\]
Then, we have
\begin{align}
    \Pr\!
    \left[
        Au\in \ort_{k,b+\Lambda\cdot\onevec}
        \setminus
        \ort_{k,b-\Lambda\cdot\onevec}
    \right]
    ~&\ge~
    \binom{m}{k-1} \cdot a^{k-1}
    \left((1-a)^{m-k+1}-r^{m-k+1}\right) \nonumber\\
    ~&=~
    \binom{m}{k-1}\cdot a^{k-1}(1-a)^{m-k+1}
    \left(1-\left(\frac{r}{1-a}\right)^{m-k+1}\right) \enspace. \label{eq:ba-1}
\end{align}
We choose \(t\) so that
\[
	a ~=~ \Pr[g\ge t+\Lambda] ~=~ \frac{k}{m}\enspace.
\]
Since \(k\le m/2\), we have \(t+\Lambda\ge0\).
By \cref{lem:oned} with derivative order \(1\) and \(x=-(t+\Lambda)\),
and using \(a=\gcdf(-(t+\Lambda))=k/m\), we have
$\gpdf(t+\Lambda) =
    O\!\left(
        \frac{k}{m}\sqrt{\log\frac{m}{k}}
    \right)$.
By \cref{clm:gaussian-quantile-density}, we also have
$\gpdf(t+\Lambda) =
    \Omega\!\left(
        \frac{k}{m}\sqrt{\log\frac{m}{k}}
    \right)$.
This implies that
\[
    \gpdf(t+\Lambda)
    ~=~
    \Theta\!\left(
        \frac{k}{m}\sqrt{\log\frac{m}{k}}
    \right) \enspace.
\]
Moreover, we know that \(t+\Lambda=O(\sqrt{\log(m/k)})\) by the standard Gaussian tail bound.
Since \(\Lambda\le c_0/(k\sqrt{\log(m/k)})\),
the Gaussian density on \([t-\Lambda,t+\Lambda]\) is within a constant factor of
\(\gpdf(t+\Lambda)\).
Hence, we have
\[
    (1-a)-r
    ~=~
    \Pr[t-\Lambda\le g<t+\Lambda]
    ~=~
    \Theta\!\left(
        \Lambda\cdot \frac{k}{m}\sqrt{\log\frac{m}{k}}
    \right) \enspace.
\] 
Since \(1-a\ge 1/2\) and \(\Lambda\le c_0/(k\sqrt{\log(m/k)})\),
choosing \(c_0>0\) sufficiently small ensures that
$(m-k+1)\frac{(1-a)-r}{1-a}$ is at most a small absolute constant.
Using \(1-(1-x)^z=\Omega(zx)\) for \(0 < zx < 1\),
we get
\begin{align}
    1-\left(\frac{r}{1-a}\right)^{m-k+1}
    &=~
    1-\left(1-\frac{(1-a)-r}{1-a}\right)^{m-k+1} \nonumber\\
    &=~
    \Omega\!\left(
        (m-k+1)\frac{(1-a)-r}{1-a}
    \right) \nonumber\\
    &=~
    \Omega\!\left(
        \Lambda \cdot k\sqrt{\log\frac{m}{k}}
    \right) \enspace. \label{eq:ba-2} 
\end{align}
By the same argument as in \cref{eq:binomial-lower-bound},
\begin{equation}
    \binom{m}{k-1}\cdot a^{k-1} \cdot (1-a)^{m-k+1}
    ~=~
    \Omega\!\left(\frac1{\sqrt{k}}\right) \enspace. \label{eq:ba-3}
\end{equation}
Combining \cref{eq:ba-1,eq:ba-2,eq:ba-3} yields
\[
\Pr\!
    \left[
        Au\in \ort_{k,b+\Lambda\cdot\onevec}
        \setminus
        \ort_{k,b-\Lambda\cdot\onevec}
    \right]
    ~=~
    \Omega\!\left(
        \Lambda \sqrt{k\cdot \log\frac{m}{k}}
    \right) \enspace. \qedhere
\]
\end{proof}

\newpage
\phantomsection
\addcontentsline{toc}{section}{References}
\bibliographystyle{alpha}
\bibliography{ref}

\end{document}